\documentclass[10pt,doublesided, doublecolumn,final]{IEEEtran}
\usepackage{float}
\usepackage{algorithm}
\usepackage[noend]{algorithmic}
\usepackage[geometry]{ifsym}
\usepackage {ifsym}
\usepackage{amssymb}
\usepackage{yhmath}
\usepackage{amscd}
\usepackage{dsfont}
\usepackage{amsmath}
\usepackage{epsfig}
\usepackage{graphics}
\usepackage{psfrag}
\usepackage{rotating}
\usepackage{amsmath}
\usepackage{amsfonts}
\usepackage{url}
\usepackage{color}
\usepackage{float}
\usepackage{balance} 

\newtheorem{lemma}{Lemma}

\newtheorem{theorem}{Theorem}
\newtheorem{definition}{Definition}

\newcommand{\ind}[1]{\mathds{1}_{\left\lbrace #1 \right\rbrace}}

\newcommand{\bs}{\boldsymbol}

\newcommand{\ds}{\displaystyle}

\newcommand{\pr}[1]{\mathrm{Pr} \left(#1\right)}
\newcommand{\SNR}{\mathrm{SNR}}
\newcommand{\INR}{\mathrm{INR}}

\newcommand{\Cldic}{\mathcal{C}_{\mathrm{LDIC}}}
\newcommand{\Nldic}{\mathcal{N}_{\mathrm{LDIC}}}
\newcommand{\Cldicfb}{\mathcal{C}_{\mathrm{LDIC/FB}}}
\newcommand{\Nldicfb}{\mathcal{N}_{\mathrm{LDIC/FB}}}
\newcommand{\Bldic}{\mathcal{B}_{\mathrm{LDIC}}}
\newcommand{\Bldicfb}{\mathcal{B}_{\mathrm{LDIC/FB}}}
\newcommand{\Cgic}{\mathcal{C}_{\mathrm{GIC}}}
\newcommand{\Ngic}{\mathcal{N}_{\mathrm{GIC}}}
\newcommand{\Cgicfb}{\mathcal{C}_{\mathrm{GIC/FB}}}
\newcommand{\Ngicfb}{\mathcal{N}_{\mathrm{GIC/FB}}}
\newcommand{\underBgic}{\underline{\mathcal{B}}_{\mathrm{GIC}}}
\newcommand{\overBgic}{\overline{\mathcal{B}}_{\mathrm{GIC}}}
\newcommand{\Bgicfb}{\mathcal{B}_{\mathrm{GIC/FB}}}
\newcommand{\underR}{\underline{\mathcal{R}}}
\newcommand{\overR}{\overline{\mathcal{R}}}
\newcommand{\underRfb}{\underline{\mathcal{R}}_{\mathrm{FB}}}
\newcommand{\overRfb}{\overline{\mathcal{R}}_{\mathrm{FB}}}

\newcommand{\GameNF}{\mathcal{G} = \left(\mathcal{K}, \left\lbrace\mathcal{A}_k \right\rbrace_{k \in \mathcal{K}},\left\lbrace
u_{k}\right\rbrace_{ k \in \mathcal{K}}\right)}
\newcommand{\gameNF}{\mathcal{G}}
  
\begin{document}
\title{Perfect Output Feedback in the \\ Two-User Decentralized Interference Channel}
\author{Samir M. Perlaza, 
Ravi Tandon, 
H. Vincent Poor, 
and Zhu Han
\thanks{Samir M. Perlaza (samir.perlaza@inria.fr) is  with the Institut National de Recherche en Informatique et en Automatique (INRIA), Lyon, France. He is also with the Department of Electrical Engineering at Princeton University, Princeton, NJ. }
\thanks{Ravi Tandon (tandonr@vt.edu) is with the Department of Computer Science and the Discovery Analytics Center at Virginia Tech, Blacksburg, VA.}
\thanks{H. Vincent Poor (poor@princeton.edu) is with the Department of Electrical Engineering at Princeton University, Princeton, NJ.}
\thanks{Zhu Han (zhan2@uh.edu) is with the Department of Electrical and Computer Engineering at University of Houston, Houston, TX.}
\thanks{This research was supported in part by the Army Research Office under MURI Grant W911NF-11-1-0036; the US National Science Foundation under grants CCF-1422090, CCF-1456921, CNS-1443917, CNS-1265268,  CNS-0953377, ECCS-1405121, ECCS-1405121 and ECCS-1343210; the China National Science Foundation grant NSFC-61428101 and the European Commission under Individual Fellowship Marie Sk\l{}odowska-Curie Action No. 659316 (CYBERNETS).
}
\thanks{Part of this work was presented at the 50th Allerton Conference on Communication, Control, and Computing (Allerton), Monticello, IL, October, 2012.}
}

\maketitle
\begin{abstract} In this paper, the $\eta$-Nash equilibrium ($\eta$-NE) region of the two-user Gaussian interference channel (IC) with perfect output feedback is approximated to within $1$ bit/s/Hz and $\eta$ arbitrarily close to $1$  bit/s/Hz. The relevance of the $\eta$-NE region is that it provides the set of rate-pairs that are achievable and stable in the IC when both transmitter-receiver pairs autonomously tune their own transmit-receive configurations seeking an $\eta$-optimal individual transmission rate. 
Therefore, any rate tuple outside the  $\eta$-NE region is not stable as there always exists one link able to increase by at least $\eta$ bits/s/Hz its own transmission rate by updating its own transmit-receive configuration.   
The main insights that arise from this work are: $(i)$ The $\eta$-NE region achieved with feedback is larger than or equal to the $\eta$-NE region without feedback. More importantly, for each rate pair achievable at an $\eta$-NE without feedback, there exists at least one rate pair achievable at an  $\eta$-NE with feedback that is weakly Pareto superior. $(ii)$ There always exists an $\eta$-NE transmit-receive configuration that achieves a rate pair that is at most $1$ bit/s/Hz per user away from the outer bound of the capacity region.
\end{abstract}
\begin{IEEEkeywords}
Interference channels, feedback communications, Gaussian channels, wireless networks, distributed information systems.
\end{IEEEkeywords}

\section{Introduction}\label{SecIntroduction}

\IEEEPARstart{I}{n} point-to-point communications, perfect output feedback does not increase the capacity either in the discrete or the continuous memoryless channel \cite{Shannon-IRETIT-1956, Kadota-TIT-1971, Dobrushin-TPA-1958}. 
At most, feedback increases the capacity by a bounded number of bits per channel use in channels with memory. This is the case for colored additive Gaussian noise \cite{Cover-TIT-1989, Butman-TIT-1969} and stationary first-order moving average Gaussian noise \cite{Kim-TIT-2006}. 
The same can be said for some multiuser channels in which the capacity region is broadened only by a limited number of bits per channel use. This is the case for the memoryless Gaussian multiple access channel (MAC) \cite{Cover-TIT-1981, Gaarder-TIT-1975, Ozarow-TIT-1984, Thomas-TIT-1987}. 
In the discrete memoryless broadcast channel (BC), there exists evidence that feedback increases the capacity region \cite{ Kramer-TIT-2003, Dueck-IC-1980,Shyevitz-TIT-2013}. However, in particular cases such as the physically degraded BC, the opposite has been formally proven \cite{El-Gamal-TIT-1978}.  

Feedback substantially enlarges the capacity region of the two-user memoryless Gaussian interference channel (IC) \cite{Suh-TIT-2011}. The same effect is observed in some special cases with a larger number of users, e.g., in the symmetric $K$-user cyclic $Z$-interference channel \cite{Tandon-TIT-2013} and the fully connected $K$-user IC \cite{Mohaher-TIT-2013}. The two-user linear deterministic IC with partial feedback has been considered in \cite{Quintero-TR-2015, SyQuoc-TIT-2013} and \cite{Quintero-ITW-2015}.
In the particular case of the two-user memoryless Gaussian IC, in the very strong interference regime, the gain provided by feedback can be arbitrarily large when the interference to noise ratios (INRs) grow to infinity. One of the reasons why feedback provides such a surprising benefit relies on the fact that it creates an alternative path to the existing point-to-point paths. For instance, in the two-user IC, feedback creates a path from transmitter $1$ (resp. transmitter $2$) to receiver $1$ (resp. receiver $2$) in which signals that are received at receiver $2$ (resp. receiver $1$) are fed back to transmitter $2$ (resp. transmitter $1$) which decodes the messages and re-transmits them to receiver $1$ (resp. receiver $2$). This implies a type of cooperation in which transmitters engage each other to transmit each other's messages.

In decentralized multiuser channels, the benefits of feedback are less well understood and the existing results from the centralized perspective do not apply immediately. In a decentralized network, each transmitter-receiver link acts autonomously and tunes its individual transmit-receive configuration aiming to optimize a given performance metric. 
Therefore, in decentralized networks, a competitive scenario arises in which the individual improvement of one link often implies the detriment of others due to the mutual interference. From this point of view, in decentralized networks, the notion of capacity region is shifted to the notion of equilibrium region. Such a region varies depending on the associated notion of equilibrium, e.g., Nash equilibrium (NE) \cite{Nash-PNAS-1950},  $\eta$-Nash equilibrium ($\eta$-NE) \cite{Nisan-Book-2007}, correlated equilibrium \cite{Aumann-JME-1974}, satisfaction equilibrium \cite{Perlaza-JSTSP-2012}, etc. In particular, when each individual link aims to selfishly optimize its individual transmission rate by tuning its transmit-receive configuration, the equilibrium region is a subregion of the capacity region and it must be understood in terms of the $\eta$-NE. Once an $\eta$-NE is achieved, none of the links has a particular interest in unilaterally deviating from the actual transmit-receive configuration as any deviation would bring an improvement of at most $\eta$ bits/s/Hz. When, $\eta = 0$, an $\eta$-NE corresponds to an NE. Essentially, any deviation from an NE implies no gain or even a loss in the individual rate of the deviating transmitter. Therefore, any rate tuple outside the NE-region is not stable as there always exists at least one link that is able to increase its own transmission rate by updating its own transmit-receive configuration.   

An approximate characterization of the  $\eta$-NE region of the decentralized Gaussian IC without feedback is presented in \cite{Berry-TIT-2011}, with $\eta \geqslant 0$ arbitrarily small. This characterization implies two important points. First, in all the interference regimes, the $\eta$-NE region is non-empty, which verifies some of the existing results in \cite{Chung-ISIT-2003, Rose-ICC-2011} and \cite{Yu-ISIT-2000}. 
Second, the individual rates achievable at an $\eta$-NE are both lower and upper bounded. The lower bound corresponds to the rate achieved by treating interference as noise, whereas the upper bound requires partial decoding of the interference. Interestingly, in some cases of the strong and very strong interference regimes, it is shown that the $\eta$-NE region equals the capacity region. Conversely, in all the other cases, the $\eta$-NE region is a subregion of the capacity region and often, it does not contain all the strictly Pareto optimal rate pairs, e.g., the rate pairs on the boundary of the sum-capacity. 

In the case of the IC with feedback, conventional wisdom leads to the idea that the $\eta$-NE region must not be different from the $\eta$-NE region of the IC without feedback. This follows from the fact that feedback can be seen as an altruistic action in which the benefit is not for the transmitter-receiver pair that implements feedback but rather for the other pair \cite{Perlaza-ALLERTON-2012}. Note for instance that, the alternative path from transmitter $i$ to receiver $i$ mentioned above appears thanks to the feedback from receiver $j$ to transmitter $j$. Note also that when a transmitter-receiver pair uses feedback to cancel the interference of a bit received during channel use $t-1$, the corresponding interfering bit must be re-transmitted during channel use $t$ such that it can be reliably decoded in order to effectively cancel its effect during channel use $t-1$. This implies that, subject to a feedback delay, such a transmitter-receiver pair  can opt for not transmitting the information bit during channel use $t-1$ and instead, transmitting it at channel use $t$ at the place of the interfering bit. This reasoning shows that for every individual rate that is achievable with feedback, there always exists an alternative transmit-receive configuration in which the previous channel outputs obtained by feedback are not used and such a configuration achieves the same transmission rate. Therefore, intuitively,  feedback should not be useful in decentralized channels given that transmitter-receiver pairs whose individual interest is their own transmission rate do not have a particular interest in using it.

However, this paper shows the opposite.
Even in the strictly competitive scenario in which both links are selfish, the use of feedback can be shown to be individually advantageous  and thus, transmitter-receiver pairs might opt to use it in some cases. 
This is basically because, when one transmitter-receiver pair uses feedback, it motivates the others to use it, which leads to a mutually beneficial situation and thus, to an equilibrium.
This observation leads to two of the most important conclusions of this work: $(i)$ The $\eta$-NE region achieved with feedback is larger than or equal to the $\eta$-NE region without feedback. More importantly, for each rate pair achievable at an $\eta$-NE without feedback, there exists at least one rate pair achievable at an $\eta$-NE with feedback that is, at least, weakly Pareto superior; and $(ii)$ There always exists an $\eta$-NE transmit-receive configuration pair that achieves a rate pair that is at most $1$ bit/s/Hz per user away from the outer bound of the capacity region even when the network is fully decentralized.

The remainder of the paper is structured as follows. In Sec. \ref{SecProblemFormulation}, the decentralized IC with feedback is formally introduced and its equivalent game theoretic model is presented. In Sec. \ref{SecLDICFB}, the $\eta$-NE region of the linear deterministic IC (LD-IC) with feedback (LD-IC-FB) is fully characterized for any $\eta \geqslant 0$. In Sec. \ref{SecGICFB}, using the intuition obtained from the LD-IC-FB, the $\eta$-NE region of the GIC with feedback is approximated to within one bit/s/Hz and $\eta \geqslant 1$ bit/s/Hz. This approximation inherits the one-bit precision of the approximation of the capacity region of this channel \cite{Suh-TIT-2011}. This work is concluded by Sec. \ref{SecConclusions}.
 
\section{Problem Formulation}\label{SecProblemFormulation}
\subsection{Channel Model}
Consider the fully decentralized two-user interference channel with perfect channel-output feedback in Fig. \ref{FigGIC}. Transmitter $i$, with $i \in \lbrace 1, 2 \rbrace$, communicates with receiver $i$ subject to the interference produced by transmitter $j \in \lbrace 1, 2\rbrace\setminus \lbrace i \rbrace$. 
At each block, transmitter $i$ sends $M_i$ information bits $b_{i,1}, \ldots, b_{i,M_i}$ by transmitting the codeword $\bs{X}_i = ( X_{i,1}, \ldots, X_{i,N_i})^{\sf T} \in \mathcal{X}_i^{N_i}$, where $\mathcal{X}_i^{N_i}$ denotes the codebook of transmitter $i$ and $N_i$ denotes the corresponding block-length in channel uses. All information bits are independent and identically distributed (i.i.d.) following a uniform probability distribution.
 \begin{figure}[t]
 \centerline{\epsfig{figure=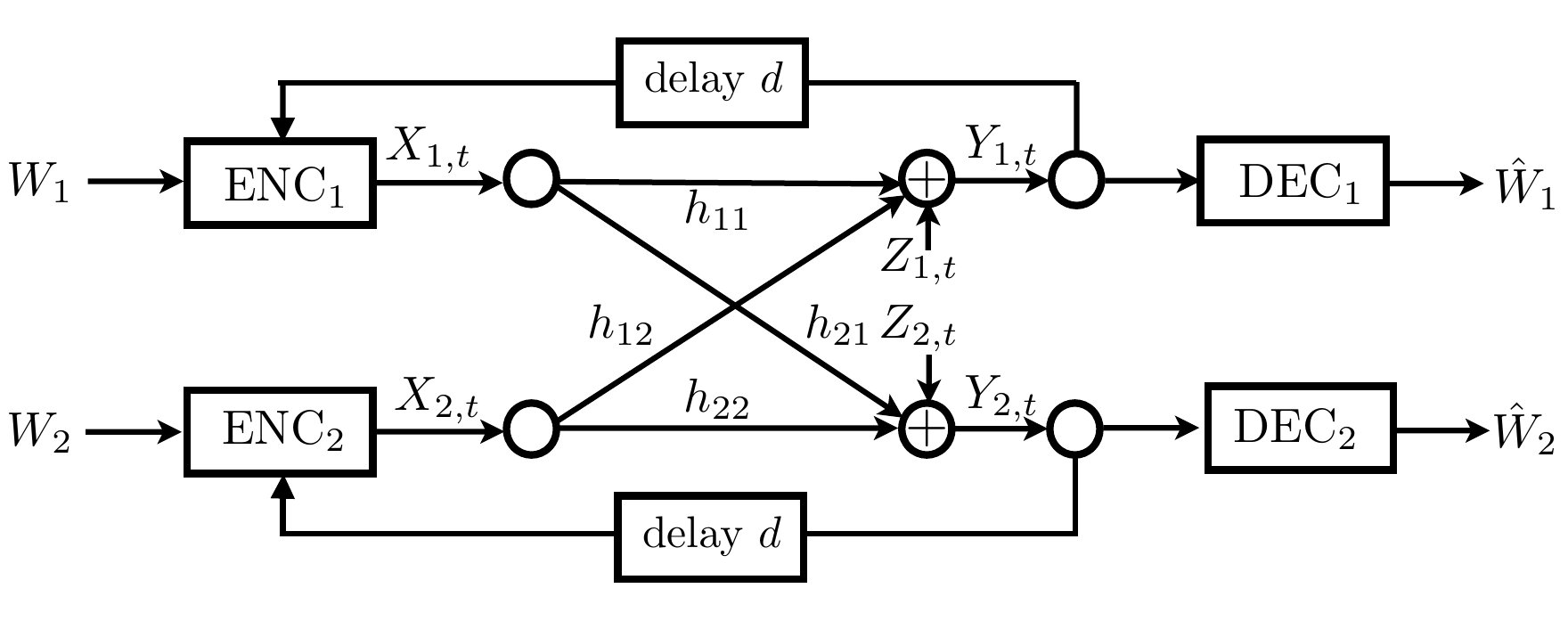,width=.5\textwidth}}
  \caption{Two-user Gaussian interference channel with perfect channel-output feedback at channel use $t$.}
  \label{FigGIC}
\end{figure}
The channel output at receiver $i$ is denoted by  $\bs{Y}_i = (Y_{i,1}, \ldots, Y_{i,N_i} )^{\sf T}$ and during channel use $t$, it holds that
\begin{eqnarray}\label{EqGICsignals}
Y_{i,t} & = & h_{ii} X_{i,t} + h_{ij} X_{j,t} + Z_{i,t},
\end{eqnarray}
where $\bs{Z}_{i} = (Z_{i,1}, \ldots, Z_{i,N_i})^{\sf T}$ represents the noise observed at receiver $i$. At each channel use $ t \in \lbrace 1, \ldots, N_i\rbrace$, the noise terms $Z_{i,t}$ are independent real Gaussian random variables with zero means and unit variances. 

The channel coefficient from transmitter $j$ to receiver $i$ is denoted by $h_{ij}$ and it is assumed to be a real. 
The channel-input vector $\bs{X}_i =  (X_{i,1}, \ldots, X_{i,N_i} )^{\sf T}$ has zero mean entries and it is subject to the variance constraint 
\begin{equation}\label{EqPowerConstraint}
\frac{1}{N_i}\sum_{t = 1}^{N_i} \left| X_{i,t}\right|^2 \leqslant P_i,
\end{equation}
with $P_i = 1$ the average transmit power of transmitter $i$. 
 
A perfect feedback link from receiver $i$ to transmitter $i$, with finite delay $d > 0 $ channel uses,  allows at the end of each channel use $t > d$, the observation of the channel output $Y_{i,t-d}$ at transmitter $i$. In the following, without loss of generality, the feedback delay is assumed to be $d =1$ channel use. 

The encoder of transmitter $i$ can be modeled as a set of deterministic mappings $f_i^{(1)}, \ldots, f_i^{(N_i)}$, with $f_i^{(1)}: \lbrace 0,1 \rbrace^{M_i} \times \mathds{N} \rightarrow \mathcal{X}_i$ and $\forall t \in \lbrace 2, \ldots, N_i\rbrace$, $f_i^{(t)}: \lbrace 0,1 \rbrace^{M_i} \times \mathds{N} \times \mathds{R}^{t-1} \rightarrow \mathcal{X}_i$, such that 
\begin{IEEEeqnarray}{lcl}\label{EqEnconderf}
X_{i,1} &=& f_i^{(1)}\big(b_{i,1},\ldots,b_{i,M_i}, \Omega_i \big) \mbox{ and }\\
X_{i,t} &=& f_i^{(t)}\big(b_{i,1},\ldots,b_{i,M_i}, \Omega_i, Y_{i,1}, \ldots, Y_{i,t-1} \big),
\end{IEEEeqnarray}
where $\Omega_i$ is an additional index randomly generated. The index $\Omega_i$ is assumed to be known by both transmitter $i$ and receiver $i$, while totally unknown by transmitter $j$ and receiver $j$.

At the end of the transmission (of B blocks), receiver $i$ uses the sequence $Y_{i,1}, \ldots, Y_{i, B \, N_i}$ to  generate the estimates $\hat{b}_{i,1}, \ldots, \hat{b}_{i, B \, M_i}$ of the transmitted bits via a decoding function $g_i: \mathds{R}^{B \, N_i} \rightarrow \lbrace 0,1 \rbrace^{B \, M_i}$, such that $(\hat{b}_{i,1}, \ldots, \hat{b}_{i, B \, M_i})^{\sf{T}} = g(Y_{i,1}, \ldots, Y_{i, B \, N_i})$.   
The average bit error probability at receiver $i$, denoted by $p_i$, is calculated as follows:
\begin{equation}
p_i = \frac{1}{B M_i} \ds\sum_{\ell = 1}^{B M_i} \ind{\hat{b}_{i,\ell} \neq {b}_{i,\ell} }.
\end{equation}

A rate pair $(R_1, R_2) \in \mathds{R}_{+}^{2}$ is said to be achievable if it satisfies the following definition.
\begin{definition}[Achievable Rate Pairs]\label{DefAchievableRatePairs}\empty{
The rate pair $(R_1, R_2) \in \mathds{R}_{+}^{2}$ is achievable if there exists at least one pair of codebooks $\mathcal{X}_1^{N_1}$ and $\mathcal{X}_2^{N_2}$ with codewords of length $N_1$ and $N_2$, respectively,  with the corresponding encoding functions $f_1^{(1)}, \ldots, f_1^{(N_1)}$ and $f_2^{(1)}, \ldots, f_2^{(N_2)}$ and decoding functions $g_1$ and $g_2$, such that the average bit error probability can be made arbitrarily small by letting the block lengths $N_1$ and $N_2$ grow to infinity.
}
\end{definition}

The Gaussian IC with feedback in Fig. \ref{FigGIC} can be fully described by four parameters: (a) the signal to noise ratios $\SNR_i = h_{ii}^2$ and (b) the interference to noise ratios $\INR_{ij} = h_{ij}^2$. 
The aim of transmitter $i$ is to autonomously choose its transmit-receive configuration $s_i$ in order to maximize its achievable rate $R_i$. More specifically, the transmit-receive configuration $s_i$ can be described in terms of the number of information bits per block $M_i$, the block-length $N_i$, the codebook $\mathcal{X}_i^{N_i}$, the encoding functions $f_i^{(1)}, \ldots, f_i^{(N_i)}$, etc. 
Note that the rate achieved by transmitter-receiver pair $i$ depends on both configurations $s_1$ and $s_2$ due to the mutual interference naturally arising in the interference channel. This reveals the competitive interaction between both links in the decentralized interference channel. The following section models this interaction using tools from game theory.

\subsection{Game Formulation}

The competitive interaction of the two transmitter-receiver pairs in the interference channel can be modeled by the following game in normal form:
\begin{equation}\label{EqGame}
\GameNF.
\end{equation}
The set $\mathcal{K} = \lbrace 1, 2 \rbrace$ is the set of players, that is, the set of transmitter-receiver pairs. The sets $\mathcal{A}_1$ and $\mathcal{A}_2$ are the sets of actions of players $1$ and $2$, respectively. An action of a player $i$, which is denoted by $s_i \in \mathcal{A}_i$, is basically its transmit-receive configuration as described above. 
The utility function of player $i$ is $u_i: \mathcal{A}_1 \times \mathcal{A}_2 \rightarrow \mathds{R}_+$ and it is defined  as the achieved rate of transmitter $i$,
\begin{equation}
\label{EqUtility}
\small
u_i(s_1,s_2) = \left\lbrace
\begin{array}{lcl}
R_i(s_1,s_2), & \mbox{if} &  p_i  < \epsilon \\
0, & &  \mbox{otherwise}
\end{array},
\right.
\end{equation}
where $\epsilon > 0$ is an arbitrarily small number and $R_i(s_1,s_2)$ denotes a transmission rate achievable with the configurations $s_1$ and $s_2$, i.e., $p_i< \epsilon$. 
Often, the rate $R_i(s_1,s_2)$ is written as $R_i$ for the sake of simplicity. However, every non-negative achievable rate is associated with a particular transmit-receive configuration pair $(s_1,s_2)$ that achieves it. It is worth noting that there might exist several transmit-receive configurations that achieve the same rate pair $(R_1, R_2)$ and distinction between the different transmit-receive configurations is made only when needed. 

A class of transmit-receive configurations $\bs{s}^* = (s_1^*, s_2^*) \in \mathcal{A}_1 \times \mathcal{A}_2$ that are particularly important in the analysis of this game are referred to as $\eta$-Nash equilibria ($\eta$-NE) and satisfy the following conditions.

\begin{definition}[$\eta$-NE \cite{Nisan-Book-2007}] \label{DefEtaNE} \emph{
In the game $\GameNF$, an action profile  $(s_1^*, s_2^*)$ is an $\eta$-NE if $\forall i \in \mathcal{K}$ and $\forall s_i \in \mathcal{A}_i$, it holds that
\begin{equation}\label{EqNashEquilibrium}
u_i (s_i , s_j^*) \leqslant u_i (s_i^*, s_j^*) + \eta.
\end{equation}
}
\end{definition}

From Def. \ref{DefEtaNE}, it becomes clear that if $(s_1^*, s_2^*)$ is an $\eta$-NE, then none of the transmitters can increase its own transmission rate more than $\eta$ bits/s/Hz by changing its own transmit-receive configuration and keeping the average bit error probability arbitrarily close to zero. Thus, at a given $\eta$-NE, every transmitter achieves a utility (transmission rate) that is $\eta$-close to its maximum achievable rate given the transmit-receive configuration of the other transmitter. Note that if $\eta = 0$, then the classical definition of NE  is obtained \cite{Nash-PNAS-1950}. The relevance of the notion of equilibrium is that at any $\eta$-NE, every transmitter-receiver pair's configuration is $\eta$-optimal with respect to the configuration of the other transmitter-receiver pairs. The following investigates the set of rate pairs that can be achieved at an $\eta$-NE. This set of rate pairs is known as the $\eta$-NE region.

\begin{definition}[$\eta$-NE Region] \label{DefNERegion} \emph{
Let $\eta\geqslant 0$. An achievable rate pair $(R_1,R_2)$ is said to be in the $\eta$-NE region of the game $\GameNF$ if there exists a pair $(s_1^*, s_2^*) \in \mathcal{A}_1 \times \mathcal{A}_2$  that is  an  $\eta$-NE and the following holds:
\begin{eqnarray}
u_1 (s_1^* , s_2^*)  =  R_1 & \mbox{ and } & u_2 (s_1^* , s_2^*)  =  R_2. 
\end{eqnarray}
}
\end{definition}

The following section studies the $\eta$-NE region of the game $\gameNF$ in \eqref{EqGame}, with $\eta \geqslant 0$ arbitrarily small, using a linear deterministic model of decentralized IC.

\section{Linear Deterministic Interference Channel with Feedback}\label{SecLDICFB}

The linear deterministic approximation of the  GIC was introduced in \cite{Avestimehr-ALLERTON-2007}. 
In general, the linear deterministic model deemphasizes the effect of background noise and focuses on the signal interactions. Furthermore, as shown in \cite{Avestimehr-TIT-2011}, it provides valuable insights that can be used to study the corresponding Gaussian model. 

The linear deterministic IC is described by four parameters: $n_{11}, n_{22}, n_{12},$ and $n_{21}$, where $n_{ii}$ captures the signal strength from transmitter $i$ to receiver $i$, and $n_{ij}$ captures the interference strength from transmitter $j$ to receiver $i$. The input-output relation during channel use $t$, with $t \in \lbrace 1, \ldots, \max(N_1,N_2) \rbrace$ is given as follows:
\begin{IEEEeqnarray}{lcl}
\nonumber
\bs{Y}_{1,t}&=& \bs{S}^{q-n_{11}} \bs{X}_{1,t} + \bs{S}^{q-n_{12}} \bs{X}_{2,t}, \mbox{ and }\\
\label{EqLDICsignals}
\bs{Y}_{2,t}&=& \bs{S}^{q-n_{21}} \bs{X}_{1,t} + \bs{S}^{q-n_{22}} \bs{X}_{2,t},
\end{IEEEeqnarray}
where 
\begin{equation}
\label{Eqq}
q=\ds\max \left( n_{11}, n_{22}, n_{12}, n_{21}\right)
\end{equation}
and $\bs{X}_{i,t}$ and $\bs{Y}_{i,t}$ are both $q$-dimensional vectors whose components are binary. Additions and multiplications are defined over the binary field, and $\bs{S}$ is a $q\times q$ lower shift matrix of the form
\begin{align}
\bs{S}=
\left[
  \begin{array}{cccccc}
    0 & 0 &0 & \cdots & 0 \\
    1 & 0 & 0 &\cdots & 0\\
    0 & 1 & 0 &\cdots  & \vdots\\
    \vdots & \ddots & \ddots  & \ddots & 0\\
    0 & \cdots & 0 & 1& 0
  \end{array}
\right].\nonumber
\end{align}
Note that at the end of a transmission of $B$ blocks, i.e., $B\cdot\max\left(N_1, N_2 \right)$ channel uses, the channel input and channel output can be represented by the super-vectors 
\begin{IEEEeqnarray}{lcl}
\label{EqUnderlineX}
\underline{\bs{X}}_i &=& \left( {\bs{X}^{(1)}_i}^{\sf{T}}, \ldots, {\bs{X}^{(B)}_i}^{\sf{T}} \right) ^{\sf{T}}\mbox{ and } \\
\label{EqUnderlineY}
\underline{\bs{Y}}_i &=& \left( {\bs{Y}^{(1)}_i}^{\sf{T}}, \ldots, {\bs{Y}^{(B)}_i}^{\sf{T}} \right)^{\sf{T}},
\end{IEEEeqnarray}
where at each block $b$, it holds that  
\begin{IEEEeqnarray}{lcl}
\label{EqXt}
\bs{X}^{(b)}_i &=&  \left( {\bs{X}_{i,1}^{(b)}}^{\sf{T}}, \ldots, {\bs{X}_{i,N_i}^{(b)}}^{\sf{T}} \right)^{\sf{T}} \mbox{ and }\\
\label{EqYt}
\bs{Y}^{(b)}_i &= & \left( {\bs{Y}_{i,1}^{(b)}}^{\sf{T}}, \ldots, {\bs{Y}_{i,N_i}^{(b)}}^{\sf{T}} \right)^{\sf{T}},
\end{IEEEeqnarray}
with $\bs{X}_{i,t}^{(b)}$ and $\bs{Y}_{i,t}^{(b)}$, two $q$-dimensional binary vectors for all $t \in \lbrace 1, \ldots, \max(N_1,N_2) \rbrace$ and $b \in \lbrace 1, \ldots, B \rbrace$.
The parameters $n_{ii}$ and $n_{ji}$ correspond to $\lfloor \frac{1}{2}\log_{2}(\mbox{SNR}_{i}) \rfloor$  and  $\lfloor \frac{1}{2}\log
(\mbox{INR}_{ji})\rfloor$, respectively, where $\SNR_i$  and $\INR_{ji}$ are the parameters of the GIC in Fig. \ref{FigGIC}. For a detailed discussion about the connections between the LD-IC and the GIC, the reader is referred to \cite{Bresler-ETT-2008}. 

\subsection{Preliminaries and Existing Results}\label{SecPreliminaries}

In the following, some of the existing results used to fully characterize the $\eta$-NE region of the LD-IC-FB are briefly presented. These results are described using the notation adopted in this paper and it might be different from the one used when they were first introduced.

\subsubsection{Capacity of the LD-IC without Feedback}

The capacity region of the two-user LD-IC without feedback is denoted by $\Cldic$ and it is fully characterized by Lemma $4$ in \cite{Bresler-ETT-2008}. 
 
\begin{lemma} [Lemma $4$ in \cite{Bresler-ETT-2008}] \label{LemmaCldic}\emph{
The capacity region $\Cldic$ of the LD-IC without feedback corresponds to the set of non-negative rate pairs $(R_1,R_2)$ satisfying
\begin{eqnarray}
 R_i \; \qquad \leqslant & n_{ii}, \mbox{ with } i \in \lbrace 1, 2 \rbrace,\nonumber\\
 R_1 + R_2 \leqslant  &(n_{11} - n_{12})^+ + \max(n_{22},n_{12}),\nonumber\\
 R_1 + R_2 \leqslant &  (n_{22} - n_{21})^+ + \max(n_{11},n_{21}),\nonumber\\
 R_1 + R_2   \leqslant &  \max\Big(n_{21}, (n_{11} - n_{12})^+ \Big) \nonumber\\
& +  \max\Big(n_{12}, (n_{22} - n_{21})^+ \Big),\label{EqRegionC}\\
  2R_1 + R_2  \leqslant  & \max\big(n_{11} , n_{21}\big) + (n_{11} - n_{12})^+ \nonumber\\
   & + \max\big(n_{12}, (n_{22} - n_{21})^+\big),\nonumber\\
  R_1 + 2R_2  \leqslant &  \max\big(n_{22} , n_{12}\big) + (n_{22} - n_{21})^+ \nonumber\\
 & + \max\big(n_{21}, (n_{11} - n_{12})^+\big).\nonumber
\end{eqnarray}
}
\end{lemma} 

Note that the capacity region shown in Lemma \ref{LemmaCldic} is a particular case of the capacity region presented in \cite{ElGamal-TIT-1982} that applies to a larger class of deterministic interference channels.

\subsubsection{$\eta$-NE Region of the LD-IC without Feedback}

The $\eta$-NE region of the linear deterministic interference channel without feedback is denoted by $\Nldic$ and it is fully characterized in \cite{Berry-TIT-2011} in terms of the set
\begin{equation}
\Bldic =  \left\lbrace (R_1, R_2) \in \mathds{R}^{2}: 
L_i \leqslant  R_i \leqslant \tilde{U}_i, \forall i \in \lbrace 1, 2 \rbrace \right\rbrace,
\end{equation}
where, $\forall i \in \lbrace 1, 2 \rbrace$,
\begin{eqnarray}
\label{EqLi}
L_i & \stackrel{\mathrm{def}}{=} & \left( n_{ii} - n_{ij} \right)^+,\\
\label{EqUi}
\tilde{U}_i &  \stackrel{\mathrm{def}}{=} & \left\lbrace \begin{array}{lcl}
n_{ii} - \min\big( L_j, n_{ij}\big) & \mbox{ if } & n_{ij} \leqslant n_{ii},\\
\min\Big( \big(n_{ij} - L_j\big)^+, n_{ii} \Big) & \mbox{ if } & n_{ij}> n_{ii}.
\end{array}\right.
\end{eqnarray}

Using this notation, the $\eta$-NE region $\Nldic$ is formalized by the following lemma.

\begin{lemma} [Theorem $1$ in \cite{Berry-TIT-2011}] \label{LemmaNldic} \emph{
The $\eta$-NE region of the two-user LD-IC without feedback $\Nldic$ is
\begin{eqnarray}
\Nldic & = &  \Bldic \cap \Cldic.
\end{eqnarray}
}
\end{lemma} 

\subsubsection{Capacity of the LD-IC with Feedback}

The capacity region of the linear deterministic interference channel with feedback is denoted by $\Cldicfb$ and it is fully characterized  by Corollary $1$ in \cite{Suh-TIT-2011}.

\begin{lemma} [Corrollary $1$ in \cite{Suh-TIT-2011}] \label{LemmaCldicfb} \emph{
The capacity region $\Cldicfb$ of the two-user LD-IC with feedback corresponds to the set of non-negative rate pairs $(R_1,R_2)$ satisfying
\begin{IEEEeqnarray}{lcl}
  \label{EqRegionCwFB1}
R_1 &\leqslant & \min\Big( \max\big( n_{11}, n_{12} \big), \max\big( n_{11}, n_{21}\big) \Big), \\
  \label{EqRegionCwFB2}
R_2  &\leqslant&  \min\Big( \max\big( n_{22}, n_{21} \big), \max\big( n_{22}, n_{12}\big) \Big), \\
  \label{EqRegionCwFB3}
R_1 + R_2 & \leqslant &   \min\Big( \max\big(n_{22}, n_{21}\big) + \big(n_{11} - n_{21}\big)^+ ,\\
& &  \max\big(n_{11}, n_{12}\big) + \big(n_{22} - n_{12}\big)^+ \Big) . \nonumber
\end{IEEEeqnarray}
}
\end{lemma} 

\subsection{Main Result}\label{SecLDIC}

\begin{figure*}[t] 
 \centerline{\epsfig{figure=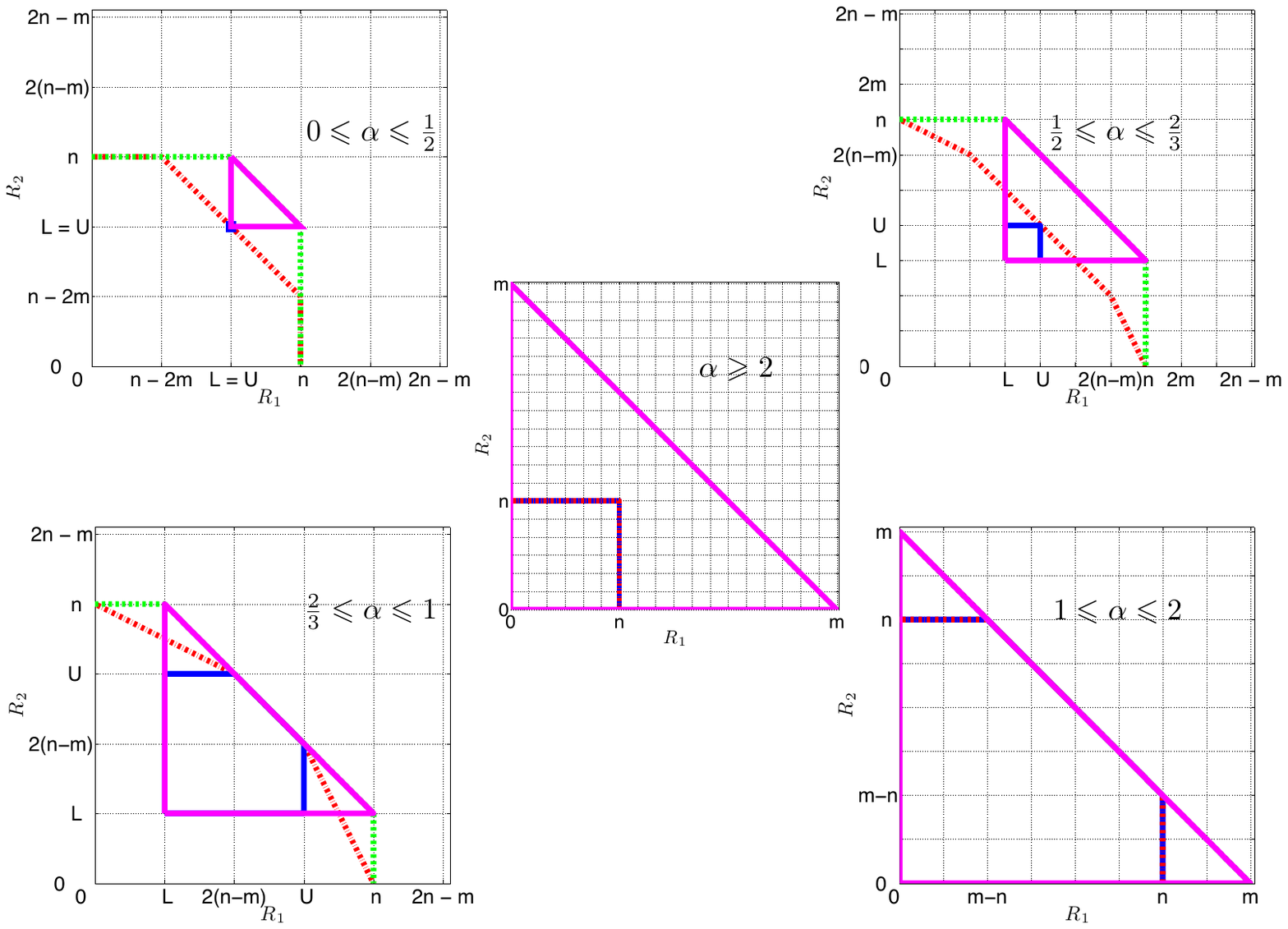,width=.85\textwidth}}
 \vspace{-0.1in}
  \caption{Illustration of $\Cldic$ (red dotted line), $\Nldic$ (solid blue line), $\Cldicfb$ (green dotted line), and $\Nldicfb$ (magenta solid line) in all interference regimes.}
  \label{FigRegions}
\end{figure*}

The $\eta$-NE region of the LD-IC with feedback and $\eta$ arbitrarily small is denoted by $\Nldicfb$.
In order to define the set $\Nldicfb$, first consider the following open set:
\begin{equation}
\Bldicfb =  \left\lbrace (R_1, R_2) \in \mathds{R}^{2}: 
\left( L_i - \eta \right)^+ \leqslant R_i, \forall i \in \lbrace 1, 2 \rbrace \right\rbrace,
\end{equation}
where $L_i$ is defined in \eqref{EqLi}.
The main result for the LD-IC with feedback is stated in terms of the sets $\Bldicfb$ and $\Cldicfb$.
\begin{theorem}\label{TheoremMainResultLDICFB}\emph{
Let $\eta \geqslant 0$. For a two-user linear deterministic IC with feedback, the $\eta$-NE region $\Nldicfb$ satisfies 
\begin{equation}
\Nldicfb  =   \Bldicfb  \cap \Cldicfb.
\end{equation}
}
\end{theorem}

The proof of Theorem \ref{TheoremMainResultLDICFB} is presented in Sec. \ref{SecProofsLDICFB}. However, before presenting the proof, some important observations are discussed. 

Note that the lowest individual rate achievable at an $\eta$-NE, i.e., $(L_i - \eta)^+$, is the same with and without feedback. That is, $\forall i \in \lbrace 1, 2 \rbrace$,
\begin{IEEEeqnarray}{lcl}
\label{EqWorseIndividualRatesNE}
\qquad \ds\inf_{R_i} \Nldic = \ds\inf_{R_i}  \Nldicfb = (L_i - \eta)^+. 
\end{IEEEeqnarray}
These two observations unveil several important facts. On the one hand, this implies the following inclusion:
\begin{equation}\label{EqInclusionNERegions}
\Nldic \subseteq \Nldicfb.
\end{equation}

As shown by the examples in the next section, strict inclusion holds for all interference regimes in the symmetric case. Conversely, strict equality holds in particular cases, consider for instance, a non-symmetric case in which $n_{11} = n_{22} =  n > 0$, $n_{21} \geqslant  2 n $ and $n_{12} = 0$.

On the other hand, the equality in  \eqref{EqWorseIndividualRatesNE} also shows that increasing the space of actions by letting both transmitter-receiver pairs choose whether to use or not the observations of the previous channel-outputs via the feedback link does not induce new NEs with lower individual rates. This statement might appear obvious, however, it has been shown that increasing the set of actions of players might reduce their individual rates or even the sum-rate at an NE. In particular, this effect has been observed in the parallel Gaussian IC and parallel Gaussian MAC when transmitters are allowed to use larger sets of transmit-receive configurations \cite{Rose-ICC-2011, Perlaza-JWCN-2012}. In the general realm of game theory, this effect is often referred to as a Braess-type paradox \cite{Braess-U-1969}. Fortunately, Braess-type paradoxes do not appear in the game $\GameNF$.

A final observation is that in some cases, all Pareto optimal rate pairs of the capacity region $\Cldicfb$ are achievable at an NE. The following lemma formalizes this observation.
\begin{lemma}[Sum-Rate Optimality]\label{LemmaSumMaxNE}\emph{
In the game $\GameNF$, the $\eta$-NE region $\Nldicfb$, with $\eta$ arbitrarily small, includes all the set of sum-optimal rates pairs $(R_1,R_2)$ of the capacity region $\Cldicfb$, if and only if:
\begin{IEEEeqnarray}{lcl}
\left(n_{22} - n_{21} \right)^+ - \left(n_{22} - n_{12} \right)^+ & \leqslant & \big( \left( n_{12} - n_{11}\right)^+ \\
\nonumber
& & - \left( n_{21} - n_{11} \right)^+ \big)^+\mbox{ and } \\
\left(n_{11} - n_{12} \right)^+ - \left(n_{11} - n_{21} \right)^+ & \leqslant & \big( \left( n_{21} - n_{22}\right)^+ \\
\nonumber
& & - \left( n_{12} - n_{22} \right)^+ \big)^+.
\end{IEEEeqnarray}
}
\end{lemma}
The proof of Lemma \ref{LemmaSumMaxNE} is presented in Appendix \ref{AppProofOfLemmaSumMaxNE}.
Lemma \ref{LemmaSumMaxNE} highlights that there exist cases in which the set of sum-optimal rate pairs of the capacity region is not entirely inside the $\eta$-NE region.
An interesting observation from Lemma \ref{LemmaSumMaxNE} is that if the direct links have larger capacity than the cross-interference links, i.e., $n_{ii} \geqslant \max\left( n_{ji}, n_{ij}\right)$, then the set $\Nldicfb$ includes all the sum-optimal rate pairs only if the capacities of the cross links are identical, i.e., $n_{21} = n_{12}$.

The following section presents some examples to provide some intuition about the impact of feedback in decentralized ICs using a symmetric channel.

\subsection{Examples}
Consider a symmetric linear deterministic IC with perfect channel-output feedback in which, $n = n_{11} = n_{22}$ and $m = n_{12} = n_{21}$, with normalized cross gain $\alpha = \frac{m}{n}$.
The regions $\Cldic$, $\Nldic$, $\Cldicfb$, and $\Nldicfb$ are plotted in Fig.  \ref{FigRegions} for different interference regimes, i.e., very weak interference ($0 \leqslant \alpha \leqslant \frac{1}{2}$), weak interference ($\frac{1}{2} < \alpha \leqslant \frac{2}{3}$), moderate interference  ($ \frac{2}{3} < \alpha \leqslant 1$), strong interference  ($1 < \alpha \leqslant 2$) and very strong interference  ($\alpha \geqslant 2$),  respectively. 

From Fig. \ref{FigRegions}, it becomes clear that at least in the symmetric case, the use of feedback increases the number of rate pairs that are achievable at an NE in all the interference regimes.
Another important observation from Fig. \ref{FigRegions} is that, at least in the symmetric case, all the sum-optimal rate pairs are achievable at an NE when feedback is available.  This drastically contrasts with the case in which feedback is not available, i.e., in the LD-IC without feedback. Note that Fig. \ref{FigRegions} shows that when $0 \leqslant \alpha \leqslant \frac{2}{3}$, only one rate pair of the infinitely many sum-optimal rate pairs is achieved at an NE in the LD-IC without feedback.
Alternatively, in  the strong interference regime, i.e., $\alpha > 1$,  $ \Nldicfb =  \Cldicfb$ and $ \Nldic =  \Cldic$, and thus, all the achievable rates of the IC with and without feedback are also achievable at an NE in the symmetric case.

\subsubsection{Achievability in the Very Weak Interference Regime}

Consider the scenario of very weak interference, for instance, let $\alpha = \frac{1}{3}$, with $m = 2$ and $n = 6$. 
From Theorem \ref{TheoremMainResultLDICFB}, it follows that the $\eta$-NE region is $\Nldicfb = \lbrace (R_1, R_2) \in \mathds{R}^{2}: \forall i \, R_i \geqslant 4, R_1 + R_2 \leqslant 10 \rbrace$. In Fig. \ref{FigRegions}, the region $\Nldicfb$ corresponds to the convex hull of the points $(4,4)$, $(6,4)$ and $(4,6)$. The rate pair $(4,4)$ is achieved by treating interference as noise and the use of feedback is not required;  the rate pairs $(6,4)$ and $(4,6)$ are achieved when one of the transmitter-receiver pairs uses feedback to cancel the interference of the other.

\paragraph{Achievability of $(4,4)$}\label{Sec44}

The rate pair $(4,4)$ is achievable when both transmitters use their top $(n-m)^+ = 4$ levels, which are interference-free, to transmit new bits at every channel use.  
Note that any attempt by a transmitter to increase its rate by using its $m=2$ lowest levels would bound its probability of error away from zero since those levels are subject to the interference of the $m=2$ top levels of the other transmitter (see Fig.  \ref{FigExamples4x4}). For transmitter $1$ (resp. transmitter $2$) to be able to use its $m=2$  bottom levels, transmitter-receiver pair $1$ (resp. transmitter-receiver pair $2$) must implement an interference cancellation technique, e.g., channel-output feedback. However, as shown in the following example, this would at most, guarantee the same rate of $4$ bits per channel use.  
Thus, these receive-transmit configurations in which transmitters use their top $(n-m)^+ = 4$ levels form an NE. Interestingly, as seen in Fig. \ref{FigRegions}, this is the worst $\eta$-NE in terms of both individual rates and sum-rate.

 \begin{figure}[t]
 \centerline{\epsfig{figure=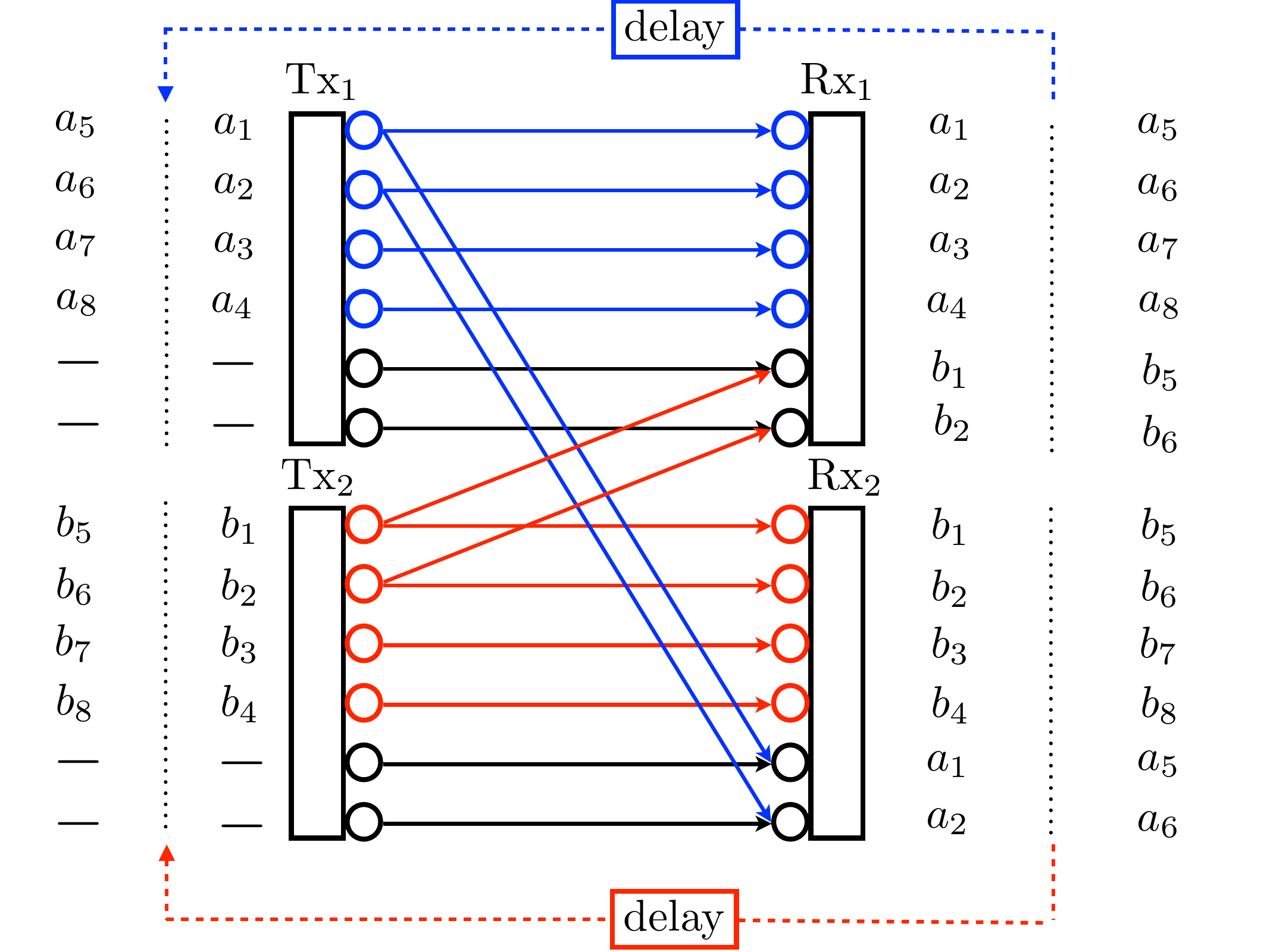,width=.5\textwidth}}
  \caption{Coding scheme for achieving the rate pair $(4,4)$ at an NE in the symmetric LD-IC with feedback, with $n = 6$ and $m = 2$.}
  \label{FigExamples4x4}
\end{figure}

\paragraph{Achievability of $(6,4)$ and $(4,6)$}

\begin{figure}[t]
 \centerline{\epsfig{figure=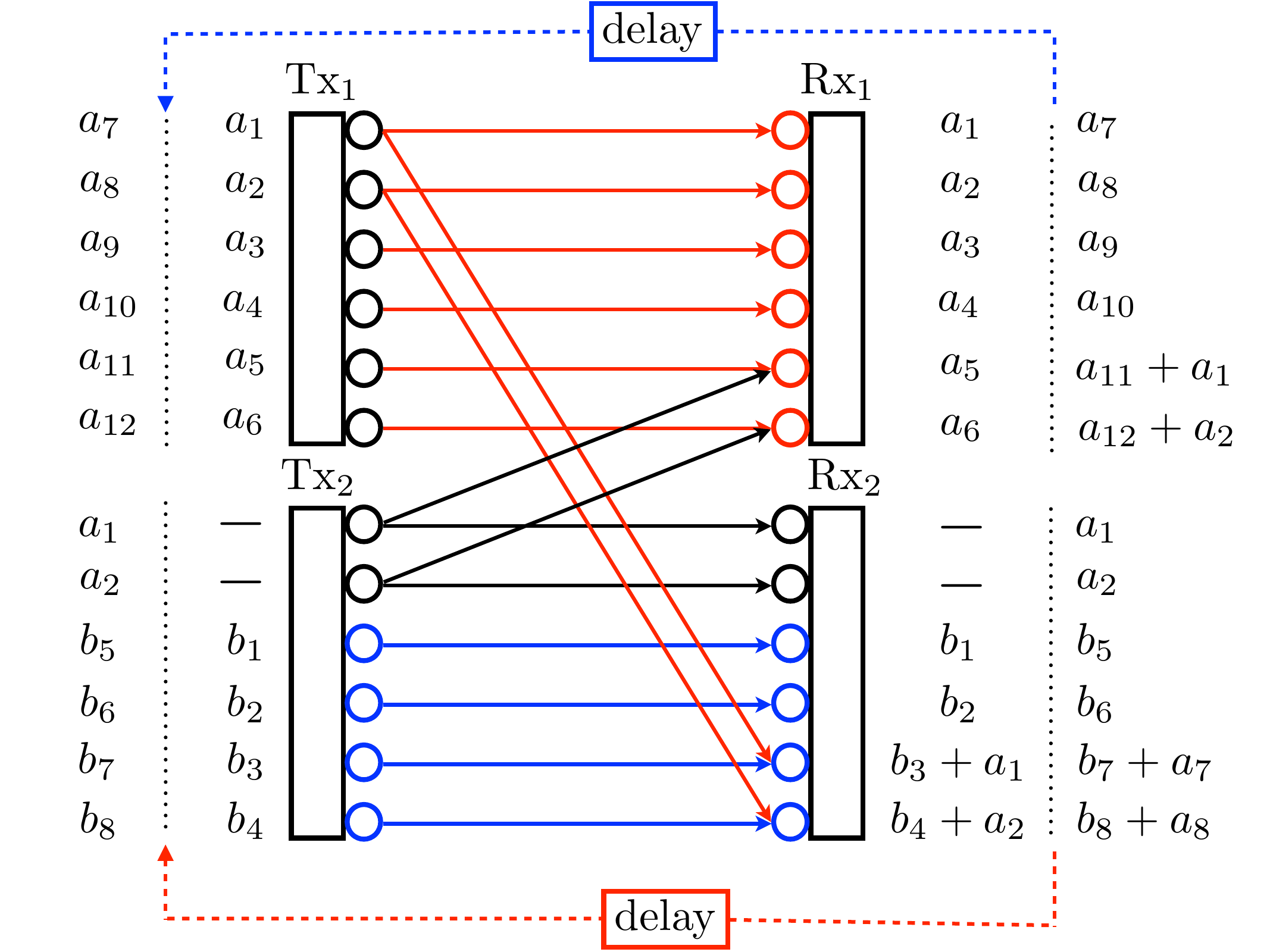,width=.5\textwidth}}
  \caption{Coding scheme for achieving the rate pair $(6,4)$ at an NE in the symmetric LD-IC with feedback, with $n = 6$ and $m = 2$.}
  \label{FigExamples6x4}
\end{figure}

The rate pair $(6,4)$ is achievable at an NE when transmitter $1$ uses all its $n = 6$ levels to transmit new bits at each channel use, while transmitter $2$ uses the following transmit-receive configuration: $(a)$ At channel use $t$, the top $m = 2$ levels of transmitter $2$ are used to re-transmit bits that have been sent by transmitter $1$ during channel use $t-1$ and that have produced interference in the bottom $m=2$ levels of receiver $2$; $(b)$ The bottom $(n-m)^+ = 4$ levels are used to transmit new bits  at each channel use.
Note that at channel use $t$, with $t > 1$, the top bits of receiver $2$ are used to clean the interference affecting the $m=2$ bottom bits in channel use $t-1$. 
Note also that canceling the interference produced by bits $a_1$ and $a_2$ at receiver $2$ during channel use $t=1$  requires transmitter $2$ to re-transmit $a_1$ and $a_2$ over its  $m=2$ (interference-free) top levels during channel use $t=2$.
A key observation is that, when transmitter $2$ retransmits $a_1$ and $a_2$ during channel use $t=2$, these bits produce interference at receiver $1$ that can be cancelled since $a_1$ and $a_2$ are received interference-free during channel use $t =1$.
The same occurs with bits $a_7$ and $a_8$ during channel uses $t=2$ and so on.  Hence, for each interfering bit that transmitter $2$ re-transmits on its top $m$ levels per channel use, it allows transmitter $1$ to transmit a new bit per channel use in the corresponding level of its $m$ bottom levels. Thus, when transmitter-receiver pair $2$ uses feedback, it benefits in the sense than it can cancel the interference and use its $m$ bottom levels. However, it unintentionally benefits transmitter-receiver pair $1$ as it is allowed to increase its individual rate.

\textbf{Remark 1:} 
Note that transmitter-receiver pair $1$ achieves a rate $R_1= \frac{6 N_1}{N_1} = 6$ bits per channel use and transmitter-receiver pair $2$ achieves a rate $R_2 = \frac{4N_2 -4}{N_2} = 4 - \eta$ bits per channel use, with $\eta = \frac{4}{N_2}$, as in the first and the last channel use only $2$ new bits are transmitted. Hence, the strategy described above is an $\eta$-NE as the strategy in which the top $(n-m)^+ = 4$ levels of transmitter $1$ are used to transmit new bits at each channel use would achieve a rate $R_2' = 4$ bits per channel use. That is, an improvement of $\eta = \frac{4}{N_2}$ is feasible by unilateral deviation of transmitter-receiver pair $2$. However, such an improvement $\eta$ can be made arbitrarily small by transmitter-receiver pair $2$ by increasing its block-length $N_2$.

\textbf{Remark 2:} 
Note that when transmitter-receiver pair $2$ uses feedback to re-transmit at channel use $t$ the bits that produced interference during channel use $t-1$, it does not increase its own transmission rate, indeed it decreases by $\eta = \frac{4}{N_2}$ bits per channel use, with respect to a transmit-receive configuration in which it transmits new bits over the $(n-m)^+ = 4$ top levels at each channel use. Hence, since the improvement is bounded by $\eta = \frac{4}{N_2}$ bits per channel use, it indifferently uses either configuration at an $\eta$-NE. 
However, for each interfering bit that transmitter $2$ re-transmits on its top $m$ levels per channel use, it allows transmitter $1$ to transmit a new bit per channel use in the corresponding level of its $m$ bottom levels.  
From this point of view, transmitter-receiver pair $2$ exhibits an \emph{altruistic behavior}. That is, it uses an alternative action that is $\eta$-optimal with respect to its own individual interest, but it allows transmitter-receiver pair $1$ to achieve a higher transmission rate.

 \subsubsection{Examples in the Strong Interference Regime}

Consider the scenario of strong interference, for instance, let $\alpha = \frac{3}{2}$, with $m = 3$ and $n = 2$. 
From Theorem \ref{TheoremMainResultLDICFB}, it follows that the $\eta$-NE region is $\Nldicfb = \lbrace (R_1, R_2) \in \mathds{R}^{2}: \forall i \, R_i \geqslant 0, R_1 + R_2 \leqslant 3 \rbrace$. In Fig. \ref{FigRegions}, the region $\Nldicfb$ corresponds to the convex hull of the points $(0,0)$, $(3,0)$ and $(0,3)$. The rate pair $(0,0)$ is achieved when interference is treated as noise and feedback is not required; the rate pairs $(3,0)$ or $(0,3)$ are achieved when one of the transmitter-receiver pairs uses feedback. As a complementary example, the achievability of the rate pairs $(2,0)$ and $(0,2)$ is presented to highlight the relevance of the random indices $\Omega_i$ in \eqref{EqEnconderf}.

\paragraph{Achievability of $(0,0)$}\label{Sec00}

The rate pair $(0,0)$ is achievable as an NE when transmitter $i$, with $i \in \lbrace 1, 2\rbrace$, uses all its levels to send  at each channel use randomly generated bits that are uniquely known by transmitter $i$ and receiver $i$. The fact that these bits are known by the intended receiver implies that it does not exists an effective transfer of information, which justifies $R_1 = R_2 = 0$.  
This transmit-receive configuration is an NE. This is due to the fact that in the strong interference regime, all levels of receiver $i$ are subject to the interference of transmitter $j$. Hence, simultaneous transmission and reliable decoding of new bits at each channel use is simply not possible. Moreover, no individual deviation from this strategy increases the individual rates. 

\textbf{Remark 3:} 
Note that transmitting randomly generated bits known at the intended receiver does not increase the rate of the corresponding transmitter-receiver pair and does not bound the probability of error away from zero in case they are not decoded. However, at the non-intended receiver, these randomly generated bits produce interference and thus, constrain the other transmitter from sending new bits at each channel use. 
Interestingly, if the random bits are not sent by transmitter $1$, transmitter $2$ would be able to send new  bits and achieve a rate $R_2 > 0$.

Transmitter-receiver pair $i$ can increase its own rate  when either transmitter-receiver pair $j$ does not transmit or when it re-transmits its interfering bits, as shown in the following example.

\paragraph{Achievability of $(3,0)$ and $(0,3)$}\label{Sec30}
The rate pairs $(3,0)$ and $(0,3)$ are achievable as an NE when one of the transmitters uses all its $m = 3$ levels to transmit new bits at each channel use and the other transmitter re-transmits these bits at each channel use $t+1$ (see Fig. \ref{FigExamples3x0}). 
Without any loss of generality, let $R_1 = 3$ bits per channel use. Under this condition, it follows from \eqref{EqRegionCwFB1}-\eqref{EqRegionCwFB3} that the maximum rate that transmitter-receiver pair $2$ can achieve is $R_2 = 0$ bits per channel use.
The key observation is that transmitter $2$ does not increase its own rate by sending new bits to its own receiver. Hence, either it remains silent or it re-transmits the channel-output fed back by receiver $2$ and in both cases it achieves a rate $R_2 = 0$. Nonetheless, when transmitter-receiver pair $2$ re-transmits the interfering bits, it provides an alternative path from transmitter $1$ to receiver $1$ and thus, it allows the achievability of  $R_1 = 3$ bits per channel use.

\textbf{Remark 4:} 
Note that in this case the use of feedback by transmitter-receiver pair $2$ is strictly beneficial to transmitter-receiver pair $1$ as it allows the achievability of a rate $R_1 = 3$ bits per channel use. Nonetheless, it does not increase nor decrease the rate of transmitter-receiver pair $2$. Hence, since there does not exist a strict improvement in the rate of transmitter-receiver pair $2$ by the use of feedback, this transmit-receive configuration can be seen as an altruistic decision, while it still constitutes an NE.

 \begin{figure}[t]
 \centerline{\epsfig{figure=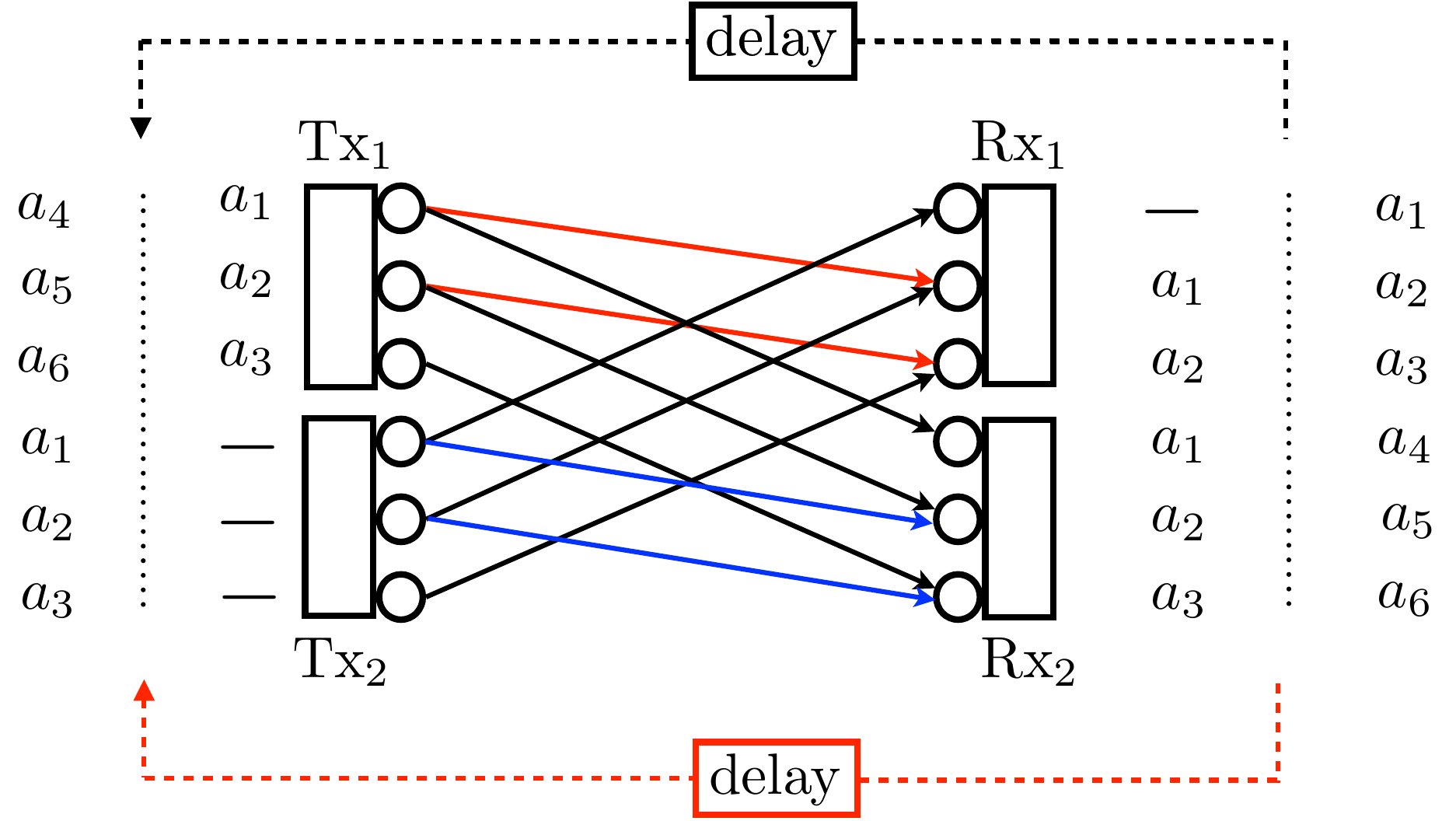,width=.5\textwidth}}
  \caption{Coding scheme for achieving the rate pair $(3,0)$ at an NE in the symmetric LD-IC with feedback, with $n = 3$ and $m = 2$.}
  \label{FigExamples3x0}
\end{figure}

 \subsection{Proof of Theorem \ref{TheoremMainResultLDICFB}}\label{SecProofsLDICFB}
To prove Theorem \ref{TheoremMainResultLDICFB}, the first step is to show that a rate pair $(R_1, R_2)$, with $R_i \leqslant  L_i - \eta$ and $i \in \lbrace 1, 2 \rbrace$, is not achievable at an $\eta$-equilibrium for an arbitrarily small $\eta$. That is, 
\begin{equation}
\label{EqPart1LDICFB}
\Nldicfb \subseteq \Cldicfb  \cap \Bldicfb.
\end{equation}
The second step is to show that any point in $\Cldicfb \cap \Bldicfb$ can be achievable at an $\eta$-equilibrium $\forall \eta>0$. That is, 
\begin{equation}
\label{EqPart2LDICFB}
\Nldicfb \supseteq \Cldicfb \cap \Bldicfb,
\end{equation}
which proves Theorem \ref{TheoremMainResultLDICFB}.

\subsubsection{Proof of \eqref{EqPart1LDICFB}}\label{SecNonEquilibriumRatePairs}
The proof of \eqref{EqPart1LDICFB} is completed by the following lemma.
\begin{lemma}\label{LemmaLDICFBa}\emph{
A rate pair $(R_1, R_2) \in \Cldicfb$, with either
$ R_1  <  \left( \left( n_{11} - n_{12} \right)^+ - \eta \right)^{+}$ or 
$R_2  < \left(\left( n_{22} - n_{21} \right)^+ - \eta \right)^{+}$
is not achievable at an $\eta$-equilibrium, with $\eta \geqslant 0$ arbitrarily small.
}
\end{lemma}

\begin{proof}
Let $(s_1^*, s_2^*)$ be an $\eta$-NE transmit-receive configuration pair such that users $1$ and $2$ achieve the rates $R_1(s_1^*,s_2^*)$ and $R_2(s_1^*,s_2^*)$, respectively. Assume, without loss of generality, that $R_1(s^*_1,s^*_2) < \left(\left( n_{11} - n_{12} \right)^+ - \eta \right)^+$. Then, let $s'_1 \in \mathcal{A}_1$ be a transmit-receive configuration in which transmitter $1$ uses its $\left( n_{11} - n_{12} \right)^{+}$ top levels, which are interference free, to transmit new bits at each channel use $t$. Hence, it achieves a rate $R_1(s'_1,s_2^*) \geqslant  \left( n_{11} - n_{12} \right)^+$.  Note also that the utility improvement $R_1(s'_1,s_2^*) - R_1(s_1^*,s_2^*) > \eta$ is always possible independently of the current transmit-receive configuration $s_2$ of user $2$. Thus, it follows that the transmit-receive configuration pair $(s_1^*, s_2^*)$ is not an $\eta$-equilibrium. This completes the proof.
\end{proof}

\subsubsection{Proof of \eqref{EqPart2LDICFB}}

To continue with the second part of the proof of Theorem \ref{TheoremMainResultLDICFB}, consider a modification of the feedback coding scheme presented in \cite{Suh-TIT-2011}. The novelty consists of allowing users to introduce some random symbols into their common messages as done in \cite{Berry-TIT-2011} for the case of the IC without feedback. 
Let $\mathcal{X}_i = \lbrace 0, 1 \rbrace^{\max(n_{ii}, n_{ji})}$ be the set of all possible binary vectors of dimension $\max(n_{ii}, n_{ji})$. Let also $\mathcal{X}_{i,P} = \lbrace 0, 1 \rbrace^{(n_{ii} - n_{ji})^+}$ and $\mathcal{X}_{i,C} = \lbrace 0, 1 \rbrace^{n_{ji}}$ be the sets of all possible binary vectors of dimensions $(n_{ii} - n_{ji})^+$  and $n_{ji}$, respectively, such that $\mathcal{X}_i = \mathcal{X}_{i,P}  \times \mathcal{X}_{i,C}$. 
Thus, at each channel use $t$, transmitter $i$ sends the bits $\bs{X}_{i,t}\in \mathcal{X}_{i}$. Vector $\bs{X}_{i,t}\in \mathcal{X}_{i}$ can be expressed as a concatenation of the components of the vectors $\bs{X}_{i,t,P} \in \mathcal{X}_{i,P}$ and $\bs{X}_{i,t,C} \in \mathcal{X}_{i,C}$. The bits in $\bs{X}_{i,t,P}$ are exclusively seen by receiver $i$ in the case in which $n_{ii} > n_{ij}$ and thus, they play the role of a private message. 
The bits in $\bs{X}_{i,t,C}$ are seen either by both receivers or at least by receiver $j$. Thus, in the case in which these bits  $\bs{X}_{i,t,C}$ are seen by both receivers, they play the role of a common message. 
Let the message index $W_i^{(b)} \in \lbrace 1, \ldots, 2^{M_i}\rbrace$ during block $b$ be expressed in terms of indices $W_{i,P}^{(b)} \in \lbrace 1, \ldots, 2^{N_i R_{i,P}}\rbrace$ and $W_{i,C}^{(b)} \in \lbrace 1, \ldots, 2^{N_i R_{i,C}}\rbrace$, that is, $W_i^{(b)} = (W_{i,P}^{(b)}, W_{i,C}^{(b)})$. The rate $R_{i,P}$ is the number of bits exclusively and reliably  decoded by receiver $i$ per channel use; and the rate  $R_{i,C}$ is the number of bits reliably decoded by at least receiver $j$ per channel use, respectively. More specifically, $R_i = R_{i,P} + R_{i,C} = \frac{M_i}{N_i}$. 

The encoding of the message index $W_i^{(b)}$, during block $b$, is made following Markov superposition coding. At the end of channel use $t-1$, transmitter $i$ decodes from the feedback signal $\bs{Y}^{(b)}_{i,t-1}$, the common symbol sent by transmitter $j$, i.e., $\bs{X}^{(b)}_{j,t-1,C}$. 
Using $\bs{X}^{(b)}_{j,t-1,C}$ and its own common symbol $\bs{X}^{(b)}_{i,t-1,C}$, transmitter $i$ uses the encoding functions $f_{i,C}^{(t)}:  \mathcal{X}_{i,C} \times \mathcal{X}_{j,C} \times \lbrace 1, \ldots, 2^{N_i R_{i,C}}\rbrace \times \lbrace 1, \ldots, 2^{N_i R_{i,R}}\rbrace \rightarrow \lbrace 0,1 \rbrace^{n_{ji}}$ to generate the next common symbol $\bs{X}^{(b)}_{i,t,C} = f_{i,C}^{(t)}\left(\bs{X}^{(b)}_{i,t-1,C}, \bs{X}^{(b)}_{j,t-1,C}, W_{i,C}^{(b)}, \Omega_{i}^{(b)} \right)$, with the symbols of the first block $\bs{X}^{(1)}_{i,0,C}$ and $\bs{X}^{(1)}_{j,0,C}$, as well as,  the symbols of the last block $\bs{X}^{(B)}_{i,N_i,C}$ and $\bs{X}^{(B)}_{j,N_j,C}$, known at both receivers.
$\Omega_i^{(b)} \in \lbrace 1, \ldots, 2^{N_i R_{i,R}}\rbrace$ is the index of the randomly generated message index that is assumed to be known by both transmitter $i$ and receiver $i$. $R_{i,R}$ is the rate that represents the number of transmitted bits that are known by both transmitter $i$ and receiver $i$ per channel use.
To generate, $\bs{X}_{i,t,P}$, transmitter $i$ uses the encoding function $f_{i,P}^{(t)}:  \mathcal{X}_{i,C} \times \mathcal{X}_{j,C} \times \mathcal{X}_{i,C} \times \lbrace 1, \ldots, 2^{N_i R_{i,P}}\rbrace  \rightarrow \lbrace 0,1 \rbrace^{( n_{ii} - n_{ji} )^{+}}$, such that  $\bs{X}_{i,t,P} = f_{i,P}^{(t)} \left(\bs{X}^{(b)}_{i,t-1,C}, \bs{X}^{(b)}_{j,t-1,C}, \bs{X}^{(b)}_{i,t,C}, W_{i,P}^{(b)}  \right)$.
Note that the rate $R_{i,R}$ does not have an impact on the number of bits effectively transmitted by transmitter-receiver pair $i$  as $R_{i,R}$ bits are already known by receiver $i$ in each channel use.
This superposition coding implies that the pair of symbols  $(\bs{X}^{(b)}_{i,t-1,C}, \bs{X}^{(b)}_{j,t-1,C})$ determines a center of a cloud of symbols, the symbols $\bs{X}^{(b)}_{i,t,C}$ determines a smaller cloud of symbols inside the previous cloud and finally, the private symbol index ${W}_{i,P}^{(b)}$ determines a private symbol inside the smaller cloud.
This coding scheme is referred to as a \emph{randomized Han-Kobayashi coding scheme with feedback} and it is thoroughly described in Appendix \ref{AppProofOfLemmaLDICFBb}. 

The proof of \eqref{EqPart2LDICFB} uses the following results:
\begin{itemize}
\item Lemma \ref{LemmaLDICFBb} proves that the randomized Han-Kobayashi scheme with feedback achieves all the rate pairs $(R_{1},R_{2}) \in  \Cldicfb$; 
\item Lemma \ref{LemmaLDICFBc} provides the maximum rate improvement that a transmitter-receiver pair can obtain when it deviates from the randomized Han-Kobayashi scheme with feedback;
\item Lemma \ref{LemmaLDICFBd} proves that when the rates of the random components $R_{1,R}$ and $R_{2,R}$ are properly chosen, the randomized Han-Kobayashi scheme with feedback is an $\eta$-NE, with $\eta \geqslant 0$; and 
\item Lemma \ref{LemmaLDICFBe} shows that for all rate pairs in $\Cldicfb \cap \Bldicfb$ there always exists a randomized Han-Kobayashi scheme with feedback that is an $\eta$-NE and achieves such a rate pair. 
\end{itemize}
This verifies that $\Cldicfb \cap \Bldicfb \subseteq \Nldicfb$ and completes the proof of \eqref{EqPart2LDICFB}.

\begin{lemma}\label{LemmaLDICFBb}
\emph{ The achievable region of the randomized Han-Kobayashi coding scheme with feedback in the linear deterministic IC is the set of tuples
$(R_{1,C}, R_{1,R},R_{1,P},R_{2,C}, R_{2,R},R_{2,P})$ that satisfy the following conditions:
\begin{equation}\label{EqHK3}
\left\lbrace 
\begin{array}{lcl} 
R_{1,C} + R_{1,R} & \leqslant & n_{21}\\
R_{2,P}  & \leqslant & \left(n_{22} - n_{12} \right)^+ \\
R_{1,C} + R_{1,P} + R_{2,C} + R_{2,R} & \leqslant & \max(n_{11},n_{12})\\
R_{2,C} + R_{2,R} & \leqslant & n_{12}\\
R_{1,P}  & \leqslant & \left(n_{11} - n_{21} \right)^+ \\
R_{2,C} + R_{2,P} + R_{1,C} + R_{1,R} & \leqslant & \max(n_{22},n_{21}).
\end{array}
\right. 
\end{equation}
}
\end{lemma}
The proof of Lemma \ref{LemmaLDICFBb} is presented in Appendix \ref{AppProofOfLemmaLDICFBb}.

The set of inequalities in \eqref{EqHK3} can be written in terms of the transmission rates $R_1 = R_{1,P} + R_{1,C}$ and $R_2 = R_{2,P} + R_{2,C}$, which yields the following conditions:
\begin{equation}\label{EqHK4}
\left\lbrace 
\begin{array}{lcl} 
R_{1,R}  & \leqslant &  n_{21} ,\\
R_{1} + R_{1,R} & \leqslant &\max(n_{11},n_{21}),\\
R_{1}  +  R_{2,R} & \leqslant &\max(n_{11},n_{12}),\\
R_{1}  + R_{2} + R_{2,R}   & \leqslant & \max(n_{11},n_{12}) + (n_{22} - n_{12})^{+},\\
R_{2,R}  & \leqslant &  n_{12}, \\
R_{2} + R_{2,R} & \leqslant &\max(n_{22},n_{12}),\\
R_{2} + R_{1,R} & \leqslant &\max(n_{22},n_{21}),\\
R_{2}  + R_{1} + R_{1,R}   & \leqslant &  \max(n_{22},n_{21})  + (n_{11} - n_{21})^{+}.
\end{array}
\right. 
\end{equation}
Note that $\forall (R_1,R_2) \in \Cldicfb$, there always exists a $R_{1,R} \geqslant 0$ and $R_{2,R} \geqslant 0$, such that $(R_1,R_{1,R},R_2,R_{2,R})$ satisfies the conditions in \eqref{EqHK4}.
Therefore, the relevance of Lemma \ref{LemmaLDICFBb} relies on the implication that any rate pair $(R_1,R_2) \in \Cldicfb$ is achievable by the randomized Han-Kobayashi coding scheme with feedback, under the assumption that the random common rates $R_{1,R}$ and $R_{2,R}$ are chosen accordingly to the conditions in  \eqref{EqHK4}.

The following lemma shows that when both transmitter-receiver links use the randomized Han-Kobayashi scheme with feedback and one of them unilaterally changes its coding scheme, it obtains a rate improvement that can be upper bounded. 
\begin{lemma}\label{LemmaLDICFBc}
\emph{Let $\eta \geqslant 0$ be an arbitrarily small number and let the rate tuple $\bs{R} = (R_{1,C} - \frac{\eta}{6}, R_{1,R} - \frac{\eta}{6},R_{1,P}- \frac{\eta}{6},R_{2,C} - \frac{\eta}{6}, R_{2,R} - \frac{\eta}{6}, R_{2,P}- \frac{\eta}{6})$ be achievable with the randomized Han-Kobayashi coding scheme with feedback such that $R_1 = R_{1,P} + R_{1,C} - \frac{1}{3}\eta$ and $R_2 = R_{2,P} + R_{2,C}- \frac{1}{3}\eta$. Then, any unilateral deviation of player $i$ by using any other coding scheme leads to a transmission rate $\tilde{R}_i$ that satisfies
\begin{IEEEeqnarray}{lcl}
\tilde{R}_{i} & \leqslant & \max\left( n_{ii}, n_{ij} \right) - (R_{j,C} + R_{j,R}) +  \frac{2}{3}\eta.
\end{IEEEeqnarray}
}
\end{lemma}

 \begin{IEEEproof} 
From Lemma \ref{LemmaLDICFBb}, it is known that for all rate tuples $(R_1,R_2) \in \Cldicfb$ there always exists a rate tuple $\bs{R} = (R_{1,C}, R_{1,R} ,R_{1,P},R_{2,C} , R_{2,R} ,R_{2,P})$, with $R_1 = R_{1,P} + R_{1,C}$ and $R_2 = R_{2,P} + R_{2,C}$ that satisfies \eqref{EqHK4}.
Assume that both transmitters achieve the rates $\bs{R}$ by using the randomized Han-Kobayashi scheme with feedback described in Appendix \ref{AppProofOfLemmaLDICFBb}.

Without loss of generality, let transmitter $1$ change its transmit-receive configuration while transmitter-receiver pair $2$ remains unchanged.
Let $B$ denote the number of blocks sent by both transmitters after the transmit-receive configuration change of transmitter-receiver pair $1$.  
Note that this new configuration can be arbitrary, i.e., it may or may not use feedback, and it may or may not use any random symbols. It can also use a new block length $\tilde{N}_1 \neq N_1$.  Hence, the total duration of this transmission is 
\begin{equation}
T=B\max(\tilde{N}_1, N_2)
\end{equation}
channel uses.
Denote by $W_1^{(b)}$ and $\Omega_{1}^{(b)}$ the message index and the random index of transmitter-receiver pair  $1$ during block $b$ after its deviation, with $b \in \lbrace 1, \ldots, B \rbrace$, $\bs{W}_{1} = \left(W_{1}^{(1)}, \ldots, W_{1}^{(B)} \right)$ and $\bs{\Omega}_{1} = \left(\Omega_{1}^{(1)}, \ldots, \Omega_{1}^{(B)} \right)$. 
Let also $\underline{\bs{\tilde{X}}}_1^{(b)} = \left( {\bs{\tilde{X}}_{1,1}^{(b)}}^{\sf{T}}, \ldots, {\bs{\tilde{X}}_{1,\tilde{N}_1}^{(b)}}^{\sf{T}} \right)^{\sf{T}}$ and $\underline{\bs{\tilde{Y}}}_1^{(b)} = \left( {\bs{\tilde{Y}}_{1,1}^{(b)}}^{\sf{T}}, \ldots, {\bs{\tilde{Y}}_{1,\tilde{N}_1}^{(b)}}^{\sf{T}} \right)^{\sf{T}}$ be the corresponding vector of outputs of transmitter $1$ and inputs  to receiver $1$ during block $b$, with $\bs{\tilde{X}}_{1,m}^{(b)}$ and $\bs{\tilde{Y}}_{1,m}^{(b)}$ two $q$-dimensional binary vectors for all $m \in \lbrace 1, \ldots, \tilde{N}_1\rbrace$ and for all $b \in \lbrace 1, \ldots, B \rbrace$. Hence, an upper bound for $\tilde{R}_1$ is obtained from the following inequalities:
\begin{align}
\label{EqConditionTildeR1LDICFB}
& T \tilde{R}_{1}   = H \left( \bs{W}_1 \right) = H \left( \bs{W}_1 | \bs{\Omega}_1 \right)  \nonumber\\
%
& = I\left( \bs{W}_1; \underline{\bs{\tilde{Y}}}_{1}^{(1)}, \ldots,\underline{\bs{\tilde{Y}}}_{1}^{(B)} |  \bs{\Omega}_1 \right) \nonumber \\ 
&\quad + H \left( \bs{W}_1 |  \bs{\Omega}_1, \underline{\bs{\tilde{Y}}}_{1}^{(1)}, \ldots,\underline{\bs{\tilde{Y}}}_{1}^{(B)} \right)  \nonumber\\
& \stackrel{(a)}{\leqslant} I\left( \bs{W}_1; \underline{\bs{\tilde{Y}}}_{1}^{(1)}, \ldots,\underline{\bs{\tilde{Y}}}_{1}^{(B)} |  \bs{\Omega}_1 \right) + T \delta_1(N_1')  \nonumber\\
& =  H\left(\underline{\bs{\tilde{Y}}}_{1}^{(1)}, \ldots,\underline{\bs{\tilde{Y}}}_{1}^{(B)} |  \bs{\Omega}_1 \right) - H\left(\underline{\bs{\tilde{Y}}}_{1}^{(1)}, \ldots,\underline{\bs{\tilde{Y}}}_{1}^{(B)} |  \bs{W}_{1}, \bs{\Omega}_1 \right) \nonumber\\
& \quad + T \delta_1(N_1') \nonumber \\
&  \stackrel{(b)}{\leqslant}   T \cdot\max\left( n_{11}, n_{12} \right)\nonumber
\\
& \quad - H\left(\underline{\bs{\tilde{Y}}}_{1}^{(1)}, \ldots,\underline{\bs{\tilde{Y}}}_{1}^{(B)} |  \bs{W}_{1}, \bs{\Omega}_1 \right) + T \delta_1(N_1'),
\end{align}
where, \newline
$(a)$ follows from Fano's inequality, as the rate $\tilde{R}_{1}$ is achievable by assumption, and thus, the all message indices $W_1^{(1)}, \ldots, W_1^{(B)}$ can be reliably decodable from $\underline{\bs{\tilde{Y}}}_1^{(1)}, \ldots, \underline{\bs{\tilde{Y}}}_1^{(B)}$ and $\bs{\Omega}_1$ after the deviation of transmitter-receiver pair $1$; and 
\newline
$(b)$ follows from the fact that  for all $b \in \lbrace 1, \ldots, B\rbrace$ and for all $m \in  \lbrace 1, \ldots, \tilde{N}_1 \rbrace$, $H(\tilde{\bs{Y}}_{1,m}^{(b)}|\tilde{\bs{Y}}_{1,1}^{(1)},\ldots,\tilde{\bs{Y}}_{1,m-1}^{(b)}, \bs{\Omega}_1 ) \leqslant H(\tilde{\bs{Y}}_{1,m}^{(b)}) \leqslant \max\left( n_{11}, n_{12} \right)$.

To refine this upper bound, the term $H\left(\underline{\bs{\tilde{Y}}}_{1}^{(1)}, \ldots,\underline{\bs{\tilde{Y}}}_{1}^{(B)} |  \bs{W}_{1}, \bs{\Omega}_1 \right)$ in \eqref{EqConditionTildeR1LDICFB} can be lower bounded.  Denote by $W_{2,C}^{(b)}$ and $\Omega_{2}^{(b)}$ the common message index and the random index of transmitter-receiver pair $2$ during block $b$ after the deviation of transmitter-receiver pair $1$, with $\bs{W}_{2,C} = \left(W_{2,C}^{(1)}, \ldots, W_{2,C}^{(B)} \right)$ and $\bs{\Omega}_{2,C} = \left(\Omega_{2,C}^{(1)}, \ldots, \Omega_{2,C}^{(B)} \right)$. Hence, the following holds:
\begin{align}
\nonumber
& T (R_{2,C} + R_{2,R})  = H(\bs{W}_{2,C}, \bs{\Omega}_2)\\
\nonumber
& \stackrel{(d)}{=} H(\bs{W}_{2,C}, \bs{\Omega}_2 | \bs{W}_{1}, \bs{\Omega}_1) \\
\nonumber
&  =  I(\bs{W}_{2,C}, \bs{\Omega}_2; \underline{\bs{\tilde{Y}}}_{1}^{(1)}, \ldots, \underline{\bs{\tilde{Y}}}_{1}^{(B)} | \bs{W}_{1}, \bs{\Omega}_1) \\
\nonumber
&  + H\left( \bs{W}_{2,C}, \bs{\Omega}_2| \underline{\bs{\tilde{Y}}}_{1}^{(1)}, \ldots, \underline{\bs{\tilde{Y}}}_{1}^{(B)}, \bs{W}_1, \bs{\Omega}_1 \right)\\
\nonumber
& \stackrel{(e)}{=}   I(\bs{W}_{2,C}, \bs{\Omega}_2; \underline{\bs{\tilde{Y}}}_{1}^{(1)}, \ldots, \underline{\bs{\tilde{Y}}}_{1}^{(B)} | \bs{W}_{1}, \bs{\Omega}_1) \\
\nonumber
&  + H\Big( \bs{W}_{2,C}, \bs{\Omega}_2| \underline{\bs{\tilde{Y}}}_{1}^{(1)}, \ldots, \underline{\bs{\tilde{Y}}}_{1}^{(B)}, \bs{W}_1, \bs{\Omega}_1, \underline{\bs{\tilde{X}}}_{1}^{(1)}, \ldots, \underline{\bs{\tilde{X}}}_{1}^{(B)} \Big)\\
\nonumber
&  \stackrel{(f)}{=}  I(\bs{W}_{2,C}, \bs{\Omega}_2; \underline{\bs{\tilde{Y}}}_{1}^{(1)}, \ldots, \underline{\bs{\tilde{Y}}}_{1}^{(B)} | \bs{W}_{1}, \bs{\Omega}_1) \\
\nonumber
 & + H\Big( \bs{W}_{2,C}, \bs{\Omega}_2| \underline{\bs{X}}_{2,C}^{(1)}, \ldots, \underline{\bs{X}}_{2,C}^{(B)}, \bs{W}_1, \bs{\Omega}_1, \underline{\bs{\tilde{X}}}_{1}^{(1)}, \ldots, \underline{\bs{\tilde{X}}}_{1}^{(B)} \Big)\\
\nonumber
&  \stackrel{(g)}{\leqslant}  I(\bs{W}_{2,C}, \bs{\Omega}_2; \underline{\bs{\tilde{Y}}}_{1}^{(1)}, \ldots, \underline{\bs{\tilde{Y}}}_{1}^{(B)} | \bs{W}_{1}, \bs{\Omega}_1)  + T \delta(N_2)\\
\label{EqConditionHY1W1}
& \leqslant  H(\underline{\bs{\tilde{Y}}}_{1}^{(1)}, \ldots, \underline{\bs{\tilde{Y}}}_{1}^{(B)} | \bs{W}_{1}, \bs{\Omega}_1) + T \delta(N_2),
\end{align}
where, \newline 
$(d)$ follows from the independence of the index messages $W_1^{(b)}$, $\Omega_{1}^{(b)}$, $W_2^{(b)}$, $\Omega_{2}^{(b)}$, for all $b \in \lbrace 1, \ldots, B \rbrace$; \newline
$(e)$ follows from the fact that the output of transmitter $i$ at the $m$-th channel use of block $b$ is a deterministic function of $W_{i}^{(b)}$ and the previous channel outputs $\underline{\tilde{\bs{Y}}}_{1}^{(1)}, \ldots,  \underline{\tilde{\bs{Y}}}_{\tilde{N}_1}^{(b-1)}, \tilde{\bs{Y}}_{1,1}^{(b)}, \ldots, \tilde{\bs{Y}}_{m-1}^{(b)}$; \newline
$(f)$ follows from the signal construction in \eqref{EqLDICsignals}; and finally, \newline
$(g)$ follows from Fano's inequality as the message indices $W_{2,C}^{(1)}, \ldots, W_{2,C}^{(B)}$ and $\Omega_{2}^{(1)}, \ldots, \Omega_{2}^{(B)}$ can be reliably decoded by receiver $1$ using the signals $\underline{\bs{X}}_{2,C}^{(1)}, \ldots, \underline{\bs{X}}_{2,C}^{(B)}$ as transmitter-receiver pair $2$ did not change its transmit-receive configuration. 

Substituting \eqref{EqConditionHY1W1} into \eqref{EqConditionTildeR1LDICFB}, it follows that:
\begin{IEEEeqnarray}{lcl}
\nonumber
\tilde{R}_{1} & \leqslant & \max\left( n_{11}, n_{12} \right) - \left( R_{2,C} + R_{2,R}\right) + \delta_{1}(N'_1) +  \delta_2(N_2).
\end{IEEEeqnarray}
Note that $ \delta_{1}(\tilde{N}_1)$ and $\delta_2(N_2)$ are monotonically decreasing functions of $\tilde{N}_1$ and $N_2$, respectively. Hence, there always exists an $\eta \geqslant 0$, such that
\begin{IEEEeqnarray}{lcl}
\nonumber
\tilde{R}_{1} & \leqslant &   \max\left( n_{11}, n_{12} \right) - \left( R_{2,C} + R_{2,R}\right) + \frac{2}{3}\eta.
\end{IEEEeqnarray}
The same can be proved for the other transmitter-receiver pair $2$ and this completes the proof.
\end{IEEEproof}
Lemma \ref{LemmaLDICFBc} reveals the relevance of the random symbols $\Omega_{1}$ and $\Omega_{2}$ used by the randomized Han-Kobayashi scheme in the construction of the common words $\bs{X}_{1,C}^{(b)}$ and $\bs{X}_{2,C}^{(b)}$ during block $b$, respectively. Even though the random symbols used by transmitter $j$ do not increase the effective transmission of data of the transmitter-receiver pair $j$, they strongly limit the rate improvement transmitter $i$ can obtain by deviating by the randomized Han-Kobayashi scheme.
This observation can be used to show that the randomized Han-Kobayashi scheme with feedback can be  an $\eta$-NE, when both $R_{1,R}$ and $R_{2,R}$ are properly chosen. For instance, for any achievable rate pair $(R_1,R_2) \in \Bldicfb \cap \Cldicfb$, there exists a randomized Han-Kobayashi scheme with feedback that achieves the rate tuple $\bs{R} = (R_{1,C} - \frac{\eta}{6}, R_{1,R} - \frac{\eta}{6}, R_{1,P}- \frac{\eta}{6},R_{2,C} - \frac{\eta}{6}, R_{2,R} - \frac{\eta}{6},R_{2,P}- \frac{\eta}{6})$, with $R_i = R_{i,P} + R_{i,C} - \frac{1}{3}\eta$ and $\eta$ arbitrarily small. Denote by $\tilde{R}_{i,\max} =  \max\left( n_{ii}, n_{ij} \right) - (R_{j,C} + R_{j,R}) + \frac{2}{3}\eta$ the maximum rate transmitter-receiver pair $i$ can obtain  by unilaterally deviating from its randomized Han-Kobayashi scheme. Then, if $\tilde{R}_{i,\max} - R_{i} \leqslant \eta$, any improvement obtained by either transmitter deviating from its randomized Han-Kobayashi scheme is bounded by $\eta$.  
The following lemma formalizes this observation.

\begin{lemma}\label{LemmaLDICFBd}
\emph{Let $\eta \geqslant 0$ be an arbitrarily small number and let the rate tuple $\bs{R} = (R_{1,C} - \frac{\eta}{6}, R_{1,R} - \frac{\eta}{6},R_{1,P} - \frac{\eta}{6},R_{2,C} - \frac{\eta}{6}, R_{2,R} - \frac{\eta}{6},R_{2,P}- \frac{\eta}{6})$ be achievable with the randomized Han-Kobayashi coding scheme with feedback and satisfy $\forall i \in \lbrace 1, 2 \rbrace$, 
\begin{IEEEeqnarray}{rcl}
\label{EqConditionNEa}
R_{i,P} + R_{i,C} & \geqslant & L_i - \frac{2}{3}\eta; and \\
\label{EqConditionNEc}
R_{i,C}  + R_{i,P} + R_{j,C} + R_{j,R} & = & \max(n_{ii},n_{ij}) + \frac{2}{3}\eta .
\end{IEEEeqnarray}
Then, the rate pair $(R_1,R_2)$, with $R_i = R_{i,C} + R_{i,P} - \frac{1}{3}\eta$ is a utility pair achieved at an $\eta$-NE equilibrium.
}
\end{lemma}
\begin{IEEEproof}
Let $(s_1^*,s_2^*) \in \mathcal{A}_1 \times \mathcal{A}_2$ be a transmit-receive configuration pair, in which the individual strategy $s_i^*$ is a randomized Han-Kobayashi scheme with feedback satisfying conditions \eqref{EqConditionNEa}-\eqref{EqConditionNEc}. 
From the assumptions of the lemma, it follows that $(s_1^*,s_2^*) $ is an $\eta$-NE at which $u_1(s_1^*,s_2^*) = R_{1,C} + R_{1,P} - \frac{1}{3}\eta$ and $u_2(s_1^*,s_2^*) = R_{2,C} + R_{2,P} - \frac{1}{3}\eta$.  
 
 \noindent
Consider that such a transmit-receive configuration pair  $(s_1^*,s_2^*)$ is not an $\eta$-NE. Then, from Def. \ref{DefEtaNE}, there exist at least one $i \in \lbrace 1, 2 \rbrace$ and at least one strategy $s_i \in \mathcal{A}_i$ such that the utility $u_i$ is improved by at least $\eta$ bits per channel use when player $i$ deviates from $s_i^*$ to $s_i$. Without loss of generality, let $ i = 1$ be the deviating user and denote by $\tilde{R}_{1}$ the rate achieved after the deviation. Then, 
\begin{equation}\label{EqConditionNEz}
u_1(s_1,s_2^*) = \tilde{R}_1 \geqslant u_1(s_1^*,s_2^*) + \eta = R_{1,C} + R_{1,P} + \frac{2}{3}\eta.
\end{equation}
However, from Lemma \ref{LemmaLDICFBc}, it follows that
\begin{IEEEeqnarray}{lcl}
\label{EqR1tildeLDICFB}
\tilde{R}_{1} & \leqslant & \max\left( n_{11}, n_{12} \right) - (R_{2,C} + R_{2,R}) +  \frac{2}{3}\eta,
\end{IEEEeqnarray}
and from the assumption in \eqref{EqConditionNEc}, with $i = 1$, i.e.,
\begin{equation}
\label{EqR2cR2r}
R_{2,C}+ R_{2,R} = \max(n_{11},n_{12}) - (R_{1,C} + R_{1,P}) + \frac{2}{3}\eta, 
\end{equation}
it follows that
\begin{IEEEeqnarray}{lcl}
\label{EqR1tildeLDICFB1}
\tilde{R}_{1} & \leqslant & R_{1,C} + R_{1,P},
\end{IEEEeqnarray}
which is a contradiction for any $\eta \geqslant 0$, given the initial assumption \eqref{EqConditionNEz}. Hence, this proves that  there does not exist another coding scheme that brings an individual utility improvement higher than $\eta$. 
Note that from \eqref{EqConditionNEa}, with $i =1$, it follows that
\begin{equation}
\label{EqR1cR1p}
R_{1,C} + R_{1,P} \geqslant  (n_{11} - n_{12})^{+} + \frac{1}{3}\eta. 
\end{equation}
Then, combining  \eqref{EqR2cR2r} and \eqref{EqR1cR1p}, it yields $R_{2,C}+R_{2,R} < n_{12} + \frac{1}{3}\eta$. This verifies that there always exists a rate $R_{2,R}$ that simultaneously satisfies  \eqref{EqConditionNEa}-\eqref{EqConditionNEc} and the corresponding conditions in \eqref{EqHK3}.
The same can be proved for the other transmitter-receiver pair. This completes the proof.
\end{IEEEproof}
The following lemma shows that all the rate pairs $(R_{1}, R_{2}) \in \Cldicfb  \cap \Bldicfb$ can be achieved by at least one $\eta$-NE.
\begin{lemma}\label{LemmaLDICFBe}\emph{
Let $\eta \geqslant 0$ be an arbitrarily small number. Then, for all rate pairs $(R_{1}, R_{2}) \in  \Cldicfb  \cap \Bldicfb$, there always exists at least one $\eta$-NE transmit-receive configuration pair $(s_1^*,s_2^*) \in \mathcal{A}_1\times\mathcal{A}_2$, such that $u_1(s_1^*,s_2^*) = R_{1}$ and $u_2(s_1^*,s_2^*) = R_{2}$.
}
\end{lemma}
\begin{proof}
From Lemma \ref{LemmaLDICFBd}, it is known that a transmit-receive configuration pair $(s_1^*,s_2^*)$ in which each player's transmit-receive configuration is the randomized Han-Kobayshi scheme with feedback satisfying conditions \eqref{EqConditionNEa}- \eqref{EqConditionNEc} is an $\eta$-NE and achieves any rate tuple $(R_{1,C}, R_{1,R}, R_{1,P},R_{2,C}, R_{2,R}, R_{2,P})$. Thus, from the conditions in \eqref{EqHK3}, \eqref{EqConditionNEa}-\eqref{EqConditionNEc}, it follows that
\begin{equation}\label{EqNERegion1}
\left\lbrace 
\begin{array}{ll} 
R_{1,C} + {R}_{1,P}  & \geqslant  (n_{11} - n_{12})^{+},\\
R_{1,C} + R_{1,R}   & \leqslant  n_{21},\\
R_{2,P}  & \leqslant  \left(n_{22} - n_{12} \right)^+, \\
R_{1,C} + R_{1,P} + R_{2,C} + R_{2,R} & =  \max(n_{11},n_{12}),\\
R_{2,C} + {R}_{2,P}  & \geqslant  (n_{22} -n_{21})^{+},\\
R_{2,C} + R_{2,R}   & \leqslant  n_{12},\\
R_{1,P}  & \leqslant  \left(n_{11} - n_{21} \right)^+, \\
R_{2,C} + R_{2,P} + R_{1,C} + R_{1,R} & =  \max(n_{22},n_{21}).
\end{array}
\right. 
\end{equation}
The region characterized by \eqref{EqNERegion1} can be written in terms of $R_{1} = R_{1,C} + R_{1,P}$ and $R_{2} = R_{2,C} + R_{2,P}$. This yields
\begin{equation}\label{EqNERegion2}
\left\lbrace 
\begin{array}{lcl} 
R_{1} & \geqslant & (n_{11} -n_{12})^{+}\\
\nonumber
R_{1,R} & \leqslant & n_{21}\\
\nonumber
R_{1} + R_{1,R} & \leqslant & \max\left(n_{11}, n_{21} \right)\\
\nonumber
R_{1} + R_{2,R} & \leqslant & \max\left(n_{11}, n_{12} \right)\\
R_{1} + R_{2} + R_{2,R} & \leqslant  & \max\left(n_{11},n_{12}\right) + \left(n_{22} - n_{12}\right)^{+}\\
R_{2} & \geqslant & (n_{22} -n_{21})^{+}\\
\nonumber
R_{2,R} & \leqslant &n_{12}\\
\nonumber
R_{2} + R_{2,R} & \leqslant & \max\left(n_{22}, n_{12} \right)\\
\nonumber
R_{2} + R_{1,R} & \leqslant & \max\left(n_{22}, n_{21} \right)\\
R_{2} + R_{1} + R_{1,R} & \leqslant  & \max\left(n_{22},n_{21}\right) + \left(n_{11} - n_{21}\right)^{+}.
\end{array}
\right. 
\end{equation}
Finally, it is proved by inspection that for all $(R_{1},R_{2})\in \Cldicfb \cap \Bldicfb$, with $R_{i} = R_{i,C} + {R}_{i,P}$ and $i \in \lbrace 1, 2 \rbrace$, there always exist an $R_{1,R} \geqslant 0$ and  an $R_{2,R} \geqslant 0$ such that conditions \eqref{EqNERegion2} are always met and thus, the rate pair $(R_{1},R_{2})$ can be achieved at an $\eta$-NE.
This completes the proof.
\end{proof}

\subsection{Discussion}

 \begin{figure*}[t]
 \centerline{\epsfig{figure=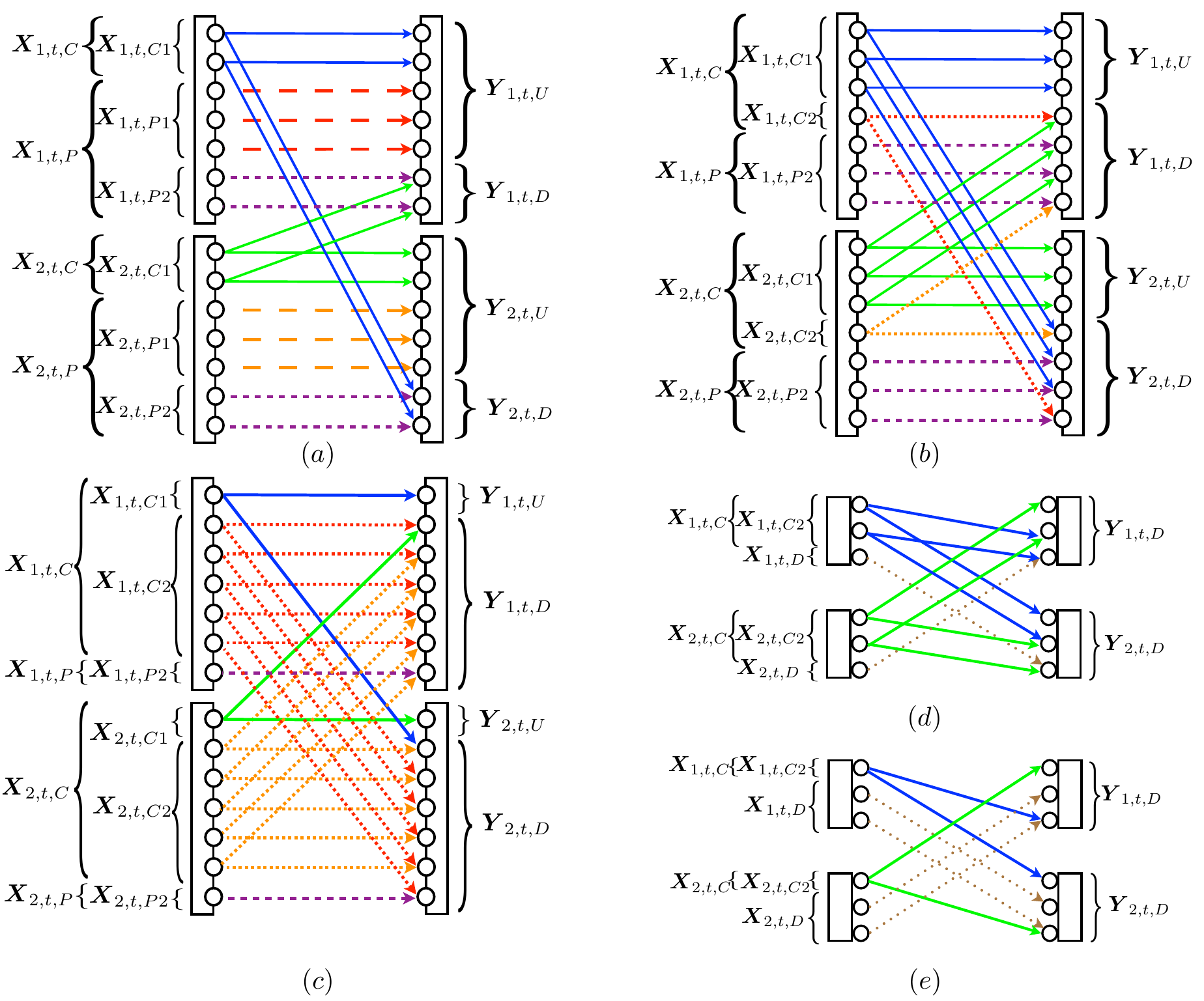,width=.85\textwidth}}
  \caption{Different components of channel input $\bs{X}_{i,t} = \left( \bs{X}_{i,t,C}, \bs{X}_{i,t,P},\bs{X}_{i,t,D}\right)$ and channel output $\bs{Y}_{i,t} = \left( \bs{Y}_{i,t,U}, \bs{Y}_{i,t,D}\right)$ at channel use $t$, with $t \in \lbrace 1, \ldots, \max(N_1, N_2)\rbrace$ and $i \in \lbrace 1, 2 \rbrace$, in the very weak interference regime $(a)$, weak interference regime $(b)$, moderate interference regime $(c)$, strong interference regime $(d)$ and very strong interference regime $(e)$.}
\label{FigMessagesLDICFB}
\end{figure*}

In this section, some properties of the $\eta$-NE transmit-receive configurations are highlighted. For this purpose, special notation is introduced. At each channel use $t$, the channel input vector $\bs{X}_{i,t}$ in \eqref{EqXt} can be written as the concatenation of three vectors: 
$\bs{X}_{i,t,P}$, $\bs{X}_{i,t,C}$ and $\bs{X}_{i,t,D}$, i.e., $\bs{X}_{i,t} = \left(\bs{X}_{i,t,C}^{\sf{T}},\bs{X}_{i,t,P}^{\sf{T}}, \bs{X}_{i,t,D}^{\sf{T}}\right)^{\sf{T}}$, as shown in Fig. \ref{FigMessagesLDICFB}. More specifically, 
\begin{itemize}
\item $\bs{X}_{i,t,C}$ contains the levels at transmitter $i$ that are observed at both receiver $i$ and receiver $j$ and thus,
\begin{IEEEeqnarray}{lcl}
\dim{\bs{X}_{i,t,C}} & = & \min\left( n_{ii}, n_{ji} \right) ; 
\end{IEEEeqnarray}
\item $\bs{X}_{i,t,P}$ contains the levels at transmitter $i$ that are observed only at receiver $i$, and thus,
\begin{IEEEeqnarray}{lcl}
\dim{\bs{X}_{i,t,P}} & = & (n_{ii} - n_{ji})^+; \mbox{ and }
\end{IEEEeqnarray}
\item $\bs{X}_{i,t,D}$ contains the levels at transmitter $i$ that are observed only at receiver $j$, and thus,
\begin{IEEEeqnarray}{lcl}
\dim{\bs{X}_{i,t,D}} & = & (n_{ji} - n_{ii})^+.
\end{IEEEeqnarray}
\end{itemize}
Note that vectors $\bs{X}_{i,t,P}$ and $\bs{X}_{i,t,D}$ do not exist simultaneously. The former exists when $n_{ii} > n_{ji}$, while the latter exists when $n_{ii} < n_{ji}$.  
Vector $\bs{X}_{i,t,C}$ can be written as the concatenation of two vectors: $\bs{X}_{i,t,C1}$ and $\bs{X}_{i,t,C2}$, i.e., $\bs{X}_{i,t,C} = \left(\bs{X}_{i,t,C1}^{\sf{T}}, \bs{X}_{i,t,C2}^{\sf{T}}\right)^{\sf{T}}$.
Vector $\bs{X}_{i,t,C1}$ (resp. $\bs{X}_{i,t,C2}$) contains the levels of $\bs{X}_{i,t,C}$ that are seen at receiver $i$ without any interference (resp. with interference of transmitter $j$). Hence,
\begin{IEEEeqnarray}{lcl}
\dim{\bs{X}_{i,t,C1}} & = & \min\left( \left(n_{ii} - n_{ij}\right)^{+}, n_{ji}\right), \mbox{ and }\\
\nonumber
\dim{\bs{X}_{i,t,C2}} & = &  \min\left( n_{ii}, n_{ji} \right) -  \min\left( \left( n_{ii} - n_{ij}\right)^{+}, n_{ji}\right).
\end{IEEEeqnarray}
Vector $\bs{X}_{i,t,P}$ can also be written as the concatenation of two vectors: $\bs{X}_{i,t,P1}$ and $\bs{X}_{i,t,P2}$, i.e., $\bs{X}_{i,t,P} = \left(\bs{X}_{i,t,P1}^{\sf{T}}, \bs{X}_{i,t,P2}^{\sf{T}}\right)^{\sf{T}}$. Vector $\bs{X}_{i,t,P1}$ (resp. $\bs{X}_{i,t,P2}$) contains the levels of $\bs{X}_{i,t,P}$ that are seen at receiver $i$ without any interference (resp. with interference of transmitter $j$). Hence,
\begin{IEEEeqnarray}{lcl}
\dim{\bs{X}_{i,t,P1}} & = & \left( \left( n_{ii} - n_{ji} \right)^{+} - n_{ij}\right)^{+} \mbox{ and }\\
\dim{\bs{X}_{i,t,P2}} & = &  \min\left( \left( n_{ii} - n_{ji} \right)^{+} , n_{ij}\right)^{+} .
\end{IEEEeqnarray}
The channel output $\bs{Y}_{i,t}$ can also be written as a concatenation of two vectors: $\bs{Y}_{i,t,U}$ and $\bs{Y}_{i,t,D}$, i.e., $\bs{Y}_{i,t} = \left(\bs{Y}_{i,t,U}^{\sf{T}}, \bs{Y}_{i,t,D}^{\sf{T}}\right)^{\sf{T}}$. Vectors $\bs{Y}_{i,t,D}$ and $\bs{Y}_{i,t,U}$ contain the levels in $\bs{Y}_{i,t}$ that are received with and without the interference from transmitter $j$, respectively. Hence,
\begin{IEEEeqnarray}{lcl}
\dim{\bs{Y}_{i,t,U}} & = & (n_{ii} - n_{ij})^+,\\
\dim{\bs{Y}_{i,t,D}} & = & n_{ij}.
\end{IEEEeqnarray}

Using this notation, the first property of an $\eta$-NE transmit-receive configuration is highlighted by the following remark.

\textbf{Remark $\mathbf{5}$:} A necessary condition for a transmit-receive configuration pair ($s_1^*,s_{2}^*)$ to be an $\eta$-NE, is that independently of $s_j^*$, $s_i^*$ requires that the levels  $\bs{X}_{i,t,C1}$ and $\bs{X}_{i,t,P1}$ are used at each channel use $t$ either to transmit new bits to receiver $i$ and increase the individual rate; or to re-transmit bits previously transmitted by transmitter $j$ that are needed at the receiver $i$ for canceling interference. This reasoning is aligned with intuition as these levels are seen interference-free at receiver $i$. If $s_i^*$ does not use these bits for at least one of these two purposes at each channel use, there always exists another configuration that does this and thus, it contradicts the fact that ($s_1^*,s_{2}^*)$ is an $\eta$-NE, with $\eta$ arbitrarily small.
 
\textbf{Remark $\mathbf{6}$:} All levels in $\bs{X}_{i,t,C2}$ are potentially subject to the interference of transmitter $j$. Hence, if at an $\eta$-NE, transmitter $j$ is interfering at levels $\bs{X}_{i,t,C2}$ with bits that are independent of bits previously sent by transmitter $i$, no new bits can be sent by transmitter $i$ via $\bs{X}_{i,t,C2}$ to increase its individual rate, as the interference cannot be cancelled. Alternatively, if transmitter $j$ is interfering at levels $\bs{X}_{i,t,C2}$ with bits previously sent by transmitter $i$ or not interfering at all, then new bits are sent by transmitter $i$ at each channel use via $\bs{X}_{i,t,C2}$. This implies that there exists a set of NEs transmit-receive configurations in which some levels (or all levels) in $\bs{X}_{i,t,C2}$ are not used to transmit new information. Similarly, all levels in $\bs{X}_{i,t,D}$ are seen only by receiver $j$. Thus, if transmitter $j$ does not re-transmit these bits to receiver $i$, these levels are not used by transmitter $i$ to send new bits at each channel use $t$. Hence, there exists a set of NEs at which some levels (or all levels) in $\bs{X}_{i,t,D}$ are not used by transmitter $i$ to increase its own rate.

Finally, from Remark $5$ and Remark $6$, it becomes clear that for every bit of interference that is cancelled by receiver $i$, transmitter $i$ needs to re-transmit the corresponding interfering bit via a level that is being reliably decoded at receiver $i$. Hence, the same rate can be obtained by a transmit-receive configuration in which feedback is not used. However, as both configurations provide the same rate, transmitter-receiver pair $i$ indifferently uses either of them, as it does not make any difference with respect to the utility function in \eqref{EqUtility}. 
Nonetheless, the situation is dramatically different for transmitter-receiver $j$ when transmitter-receiver pair $i$ decides whether to or not to re-transmit interfering bits. Indeed, for every bit that transmitter-receiver $i$ decides to retransmit to cancel interference, transmitter-receiver pair $j$ gains one additional bit per channel use. This is basically because that re-transmitted bit has been already reliably decoded by receiver $j$ and thus, its interference can be cancelled. This highlights the altruistic nature of using feedback, as first suggested in \cite{Perlaza-ALLERTON-2012}.

\section{Gaussian Interference Channel with Feedback}\label{SecGICFB}

This section presents an approximation to within $1$ bit per channel use of the $\eta$-NE region of the GIC with feedback (GIC-FB), with $\eta \geqslant 1$ bits per channel use. This result is given in terms of existing inner and outer bounds of the capacity region of the GIC-FB, which are briefly described hereunder.

\subsection{Preliminaries and Existing Results}

The following definition provides a formal description of a class of bounds known as ``approximation to within $b$ units''.
\begin{definition}[Approximation to within $\xi$ units]\label{DefApproximationInNbis}\emph{
A closed and convex region $\mathcal{X} \subset \mathds{R}_{+}^n$ is approximated to within $\xi$ units if there exist two sets $\underline{\mathcal{X}}$ and $\overline{\mathcal{X}}$ such that $\underline{\mathcal{X}} \subseteq \mathcal{X}  \subseteq \overline{\mathcal{X}}$ and $\forall \bs{x}= \left(x_1, \ldots, x_n\right) \in \overline{\mathcal{X}}$ then $\left((x_1-\xi)^+, \ldots, (x_n-\xi)^+\right) \in \underline{\mathcal{X}}$.
}
\end{definition}

Using Def. \ref{DefApproximationInNbis} existing results can be easily described.

\subsubsection{Capacity Region of the Gaussian IC}

The capacity region of the real GIC without feedback is denoted by $\Cgic$. An exact characterization of $\Cgic$ is known only for the case of the very weak interference regime \cite{Shang-TIT-2009, Motahari-TIT-2009,Annapureddy-TIT-2009} and the very strong interference regime \cite{Han-TIT-1981, Carleial-TIT-1978}. In all the other regimes, the capacity region is approximated to within half a bit per dimension \cite{Etkin-TIT-2008} (see Def. \ref{DefApproximationInNbis}). The approximation in \cite{Etkin-TIT-2008} is given in terms of two regions: (a) A region $\underR$ that is achievable with a ``simplified'' Han-Kobayashi scheme; and (b) an outer bound of the capacity region, denoted by $\overR$. The full descriptions of both $\underR$ and $\overR$  are available in \cite{Etkin-TIT-2008}.

\subsubsection{$\eta$-NE Region of the Gaussian IC}
The $\eta$-NE region of the GIC without feedback is denoted by $\Ngic$ and it has been approximated to within half a bit per dimension in \cite{Berry-TIT-2011}. This approximation is given in terms of two other regions: (a) $\underBgic$ that acts as an inner bound; and (b) $\overBgic$ that acts as an outer bound. Here, for the case of the real GIC, it follows that
\begin{IEEEeqnarray}{lcl}
\nonumber
\overBgic & = & \lbrace (R_1,R_2) : L_i \leqslant R_i \leqslant \tilde{U}_i, \; i \in \lbrace 1,2\rbrace \rbrace, \\
\nonumber
\underBgic & = & \lbrace (R_1,R_2) : L_i \leqslant R_i \leqslant \max\left\lbrace \tilde{U}_i- 1, L_i\right\rbrace \; i \in \lbrace 1,2\rbrace \rbrace,
\end{IEEEeqnarray}
and
\begin{IEEEeqnarray}{lcl}
\label{EqLiGaussian}
L_i & = &   \frac{1}{2}\log\left( 1 + \frac{\SNR_i}{1 + \INR_{ij}}\right)
\end{IEEEeqnarray}
and
\begin{IEEEeqnarray}{lcl}
\nonumber
\tilde{U}_i & = &    \frac{1}{2} \min\Bigg( 
\log\left(1 + \SNR_i + \INR_{ij} \right)  \\
\nonumber
& & -\log\left( 1 + \frac{\left[\SNR_{j} - \max\left( \INR_{ji}, \SNR_j/ \INR_{ij} \right]^+ \right) }{1 + \INR_{ji} + \max\left( \INR_{ji},  \SNR_j/\INR_{ij}\right)}\right),\\
\label{EqUiGaussian}
& &  \log\left( 1 + \SNR_i \right)
\Bigg).
\end{IEEEeqnarray}
Note that $L_i$ is the rate achieved by the transmitter-receiver pair $i$ when it saturates the power constraint in \eqref{EqPowerConstraint} and treats interference as noise. Following this notation, the $\eta$-NE region of the two-user GIC can be written as in the following lemma.
\begin{lemma}[Theorem $2$ in \cite{Berry-TIT-2011}]\label{LemmaNgic}\emph{
Let $\eta \geqslant 0$. The $\eta$-NE region of the two-user GIC $\Ngic$ is approximated to within half a bit per dimension by the regions
$\underR \cap \underBgic$ and $\overR \cap \overBgic$ and thus, $\underR \cap \underBgic \subseteq \Ngic \subseteq \overR \cap \overBgic$.
}
\end{lemma}

\subsubsection{Capacity Region of the Gaussian IC with Feedback}

One of the approximations of $\Cgicfb$ is given by Suh and Tse in \cite{Suh-TIT-2011}. This approximation is to within one bit (Def. \ref{DefApproximationInNbis}) for the case of the real GIC and it is given in terms of two regions: 
(a) a region $\underRfb$ achievable with a simplified Han-Kobayashi scheme with feedback that uses block Markov encoding and backward decoding (Theorem $2$ in \cite{Suh-TIT-2011}); and 
(b) an outer-bound region (Theorem $3$ in \cite{Suh-TIT-2011}), denoted in the following by $\overRfb$. 
The set of rate pairs $(R_1, R_2) \in \underRfb$ satisfy the following set of inequalities for a given $\rho \in [0,1]$: 
\begin{IEEEeqnarray}{l}
\nonumber
R_1  \leqslant  \frac{1}{2} \log\Big( 1 + \SNR_{1} + \INR_{12}  + 2\rho \sqrt{\SNR_1 \INR_{12}} \Big) \\
\label{EqUnderRfb1}
\qquad  - \frac{1}{2},\\
\nonumber
R_1  \leqslant  \frac{1}{2}\log\left( 1 + (1-\rho)\INR_{21}\right) + \frac{1}{2}\log\left( 2 + \frac{\SNR_1}{\INR_{21}} \right) \\
\label{EqUnderRfb2}
\qquad - 1, \quad\\
\nonumber
R_2  \leqslant  \frac{1}{2}\log\Big( 1 + \SNR_{2} + \INR_{21}+ 2\rho \sqrt{\SNR_2 \INR_{21}} \Big)\\
\label{EqUnderRfb3}
\qquad  - \frac{1}{2},\\ 
\nonumber
R_2  \leqslant   \frac{1}{2} \log\left( 1 + (1-\rho)\INR_{12}\right)  +  \frac{1}{2} \log\left( 2 + \frac{\SNR_2}{\INR_{12}} \right) \\
\label{EqUnderRfb4}
\qquad - 1\\
\label{EqUnderRfb5}
R_1 + R_2  \leqslant  \frac{1}{2} \log\left( 2 + \frac{\SNR_{1}}{\INR_{21}} \right)  - 1   \\
       \nonumber
\qquad  +  \frac{1}{2} \log\Big( 1 +  \SNR_2+ \INR_{21} + 2\rho\sqrt{\SNR_2  \INR_{21}} \Big),\\
\label{EqUnderRfb6}
R_1 + R_2  \leqslant   \frac{1}{2}\log\left( 2 + \frac{\SNR_{2}}{\INR_{12}} \right)  - 1 \\
\nonumber
\qquad +  \frac{1}{2} \log\Big( 1 +  \SNR_1 + \INR_{12} + 2\rho\sqrt{\SNR_1 \INR_{12}} \Big) .
\end{IEEEeqnarray}
The region $\overRfb$ is an outer bound of the capacity region, i.e., $\Cgicfb \subseteq \overRfb$. The region $\overRfb$ is the set of rate pairs $(R_1, R_2)$ that satisfy the following set of inequalities, with $\rho \in [0,1]$: %
\begin{IEEEeqnarray}{l}
\label{EqUpperBoundR1Capacity1}
R_1 \leqslant  \frac{1}{2}\log\Big( 1 + \SNR_{1} + \INR_{12}  + 2\rho \sqrt{\SNR_{1} \INR_{12}} \Big), \quad\\  
\nonumber
R_1  \leqslant  \frac{1}{2} \log\left( 1 + (1-\rho^2)\INR_{21}\right) \\
\label{EqUpperBoundR1Capacity2}
\qquad  + \frac{1}{2} \log\left( 1 + \frac{(1-\rho^2)\SNR_1}{1+ (1-\rho^2)\INR_{21}} \right), \\
\label{EqUpperBoundR1Capacity3}
R_2  \leqslant  \frac{1}{2} \log\Big( 1 + \SNR_{2} + \INR_{21} + 2\rho \sqrt{\SNR_{2} \INR_{21}} \Big),\\  
\nonumber
R_2  \leqslant  \frac{1}{2} \log\left( 1 + (1-\rho^2)\INR_{12}\right) \\
\label{EqUpperBoundR1Capacity4}
\qquad + \frac{1}{2}\log\left( 1 + \frac{(1-\rho^2)\SNR_2}{1+ (1-\rho^2)\INR_{12}} \right),\\
\label{EqUpperBoundR1Capacity5}
R_1 + R_2  \leqslant  \frac{1}{2} \log\left( 1 + \frac{(1-\rho^2)\SNR_{1}}{1 + (1-\rho^2)\INR_{21}} \right) \\
\nonumber    
\quad + \frac{1}{2} \log\Big( 1 +  \SNR_2 + \INR_{21} + 2\rho\sqrt{\SNR_2\INR_{21}} \Big) \mbox{, and }\\
\label{EqUpperBoundR1Capacity6}
R_1 + R_2  \leqslant  \frac{1}{2} \log\left( 1 + \frac{(1-\rho^2)\SNR_{2}}{1 + (1-\rho^2)\INR_{12}} \right) \\
\nonumber    
\quad + \frac{1}{2} \log\Big( 1 +  \SNR_1 + \INR_{12} + 2\rho\sqrt{\SNR_1 \INR_{12}} \Big).
\end{IEEEeqnarray}

The approximation of $\Cgicfb$ is described in terms of $\underRfb$ and $\overRfb$ by the following lemma for the case of the real GIC.

\begin{lemma}[Theorem $4$ in \cite{Suh-TIT-2011}]\label{LemmaCgicfb}\emph{
The capacity region $\Cgicfb$ is approximated to within one bit by the regions $\underRfb$ and $\overRfb$.
}
\end{lemma}

\subsection{Main Results}

In this subsection, the $\eta$-NE region of the real GIC-FB $\Ngicfb$ is approximated to within one bit per channel use (Def. \ref{DefApproximationInNbis}), with $\eta \geqslant 1$ bit per channel use.
This approximation is given in terms of three regions: $\underRfb$, $\overRfb$ and $\Bgicfb$, where  the closed region $\Bgicfb$ is
\begin{IEEEeqnarray}{lcl}
\label{EqBox}
\Bgicfb & = & \lbrace (R_1,R_2) : (L_i - \eta)^{+} \leqslant R_i , \; i \in \lbrace 1,2\rbrace \rbrace, 
\end{IEEEeqnarray}
with $L_i$ given by \eqref{EqLiGaussian}.
Using these elements, the main result is given by the following theorem.
\begin{theorem}[$\eta$-NE Region of the GIC with Feedback]\label{TheoremMainResultGICFB}\emph{
Let $\eta \geqslant 1$.
Then, the  $\eta$-NE region $\Ngicfb$ of the real Gaussian interference channel with perfect output feedback is approximated to within $1$ bit  by the regions $\underRfb \cap \Bgicfb$ and $\overRfb \cap \Bgicfb$ and it satisfies:
\begin{equation}
\underRfb \cap \Bgicfb \subseteq \Ngicfb \subseteq \overRfb \cap \Bgicfb.
\end{equation}
}
\end{theorem}

It is worth noting that Theorem \ref{TheoremMainResultGICFB} is analogous to Theorem \ref{TheoremMainResultLDICFB}. In the case of the LD-IC-FB (Theorem \ref{TheoremMainResultLDICFB}), the $\eta$-NE region is fully characterized for any $\eta\geqslant 0$ arbitrarily small. In the case of the decentralized real GIC-FB (Theorem \ref{TheoremMainResultGICFB}),  the $\eta$-NE region is approximated to within $1$ bit and for any $\eta \geqslant 1$. This implies that for all $(R_1,R_2) \in \Ngicfb$, any unilateral deviation from the equilibrium strategy that achieves these rates does not bring a rate improvement larger than $1$ bit per channel use. 
The relevance of Theorem \ref{TheoremMainResultGICFB} relies on two important implications:    
$(a)$ If the pair of configurations $(s_1, s_2)$ is an $\eta$-NE, then players $1$ and $2$ always achieve a rate equal to or higher than $(L_1 -\eta)^+$ and $(L_2 -\eta)^+$, with  $L_1$ and $L_{2}$ as in \eqref{EqLiGaussian}, respectively; and $(b)$ There always exists an $\eta$-NE transmit-receive configuration pair $(s_1,s_2)$ that achieves a rate pair $(R_1(s_1,s_2), R_2(s_1,s_2))$ that is at most $1$ bit per channel use per user away from the outer bound of the capacity region.

\subsection{Proof of Theorem \ref{TheoremMainResultGICFB}}

The proof of Theorem \ref{TheoremMainResultGICFB} closely follows along the same lines as the proof of Theorem \ref{TheoremMainResultLDICFB}. In the first part of this proof, given an $\eta \geqslant 0$, it is shown that $R_1 \geqslant  (L_1 -\eta)^+$ and $R_2 \geqslant  (L_2 -\eta)^+$ are necessary conditions for the rate pair $(R_1, R_2)$ to be achievable at an $\eta$-NE. That is, 
\begin{equation}
\label{EqPartI}
\Ngicfb \subseteq  \overRfb  \cap \Bgicfb.
\end{equation}
In the second part of the proof, it  is shown that any point in $\underRfb  \cap \Bgicfb$ is an $\eta$-NE, with $ \eta \geqslant 1$, that is, 
\begin{equation}
\label{EqPartII}
\underRfb \cap \Bgicfb \subseteq \Ngicfb,
\end{equation}
which proves Theorem \ref{TheoremMainResultGICFB}.

\subsubsection{Proof of \eqref{EqPartI}}
The proof of \eqref{EqPartI} is completed by the following lemma.

\begin{lemma}\label{LemmaGICFBa}\emph{
A rate pair $(R_1, R_2) \in \Cgicfb$, with either
$ R_1  < \left( \frac{1}{2} \log(1 + \frac{\SNR_{1}}{1 + \INR_{12}}) - \eta \right)^+$ or 
$R_2  <  \left(\frac{1}{2} \log(1 + \frac{\SNR_{2}}{1 + \INR_{21}}) - \eta \right)^+$
is not an $\eta$-NE, for a given $\eta \geqslant 0$.
}
\end{lemma}

\begin{proof}
Let $(s_1^*,s_2^*)$ be a transmit-receive configuration pair such that users achieve the rate pair $R_1 = R_1(s_1^*,s_2^*)$ and $R_2 = R_2(s_1^*,s_2^*)$, respectively, and assume $(s_1^*, s_2^*)$ is an $\eta$-NE. Hence, from Def. \ref{DefEtaNE}, it holds that any rate improvement of a transmitter-receiver pair that unilaterally deviates from $(s_1^*, s_2^*)$ is upper bounded by $\eta$. 
Without loss of generality, let $R_1(s_1^*,s_2^*) <\left( \frac{1}{2}\log(1 + \frac{\SNR_{1}}{1 + \INR_{12}}) - \eta \right)^{+}$. Then, note that independently of the transmit-receive configuration of transmitter-receiver pair $2$, transmitter-receiver pair  $1$ can always use a transmit-receive configuration $s_1'$ in which transmitter $1$ saturates the average power constraint \eqref{EqPowerConstraint} and interference is treated as noise at receiver $1$. Thus, transmitter-receiver pair  $1$ is always able to achieve the rate $R(s_1', s_2^*) = \frac{1}{2} \log(1 + \frac{\SNR_{1}}{1 + \INR_{12}})$, which implies that a utility improvement $R(s_1', s_2^*) - R(s_1^*, s_2^*) > \eta$ is always possible.  Thus, from Def. \ref{DefEtaNE}, the assumption that the rate pair $(R_1(s_1^*, s_2^*),R_2(s_1^*, s_2^*))$ is an $\eta$-NE does not hold. This completes the proof.
\end{proof}

\subsubsection{Proof of \eqref{EqPartII}}
 
Consider the randomized Han-Kobayashi scheme with feedback introduced in Sec. \ref{SecLDIC}, for the case of the LD-IC model. This coding scheme can be extended to the Gaussian case by letting $\mathcal{X}_{i,P}$ and $\mathcal{X}_{i,C}$ be the set of private and common codewords of length $N_i$ symbols for transmitter $i$ such that  $\forall \bs{X}_{i,C} \in \mathcal{X}_{i,C}$, $\frac{1}{N_i} \mathds{E}[\bs{X}_{i,C}^{\sf{T}} \bs{X}_{i,C}] \leqslant \lambda_{i,C}$ and $\forall \bs{X}_{i,P} \in \mathcal{X}_{i,P}$, $\frac{1}{N_i} \mathds{E}[\bs{X}_{i,P}^{\sf{T}} \bs{X}_{i,P}] \leqslant \lambda_{i,P}$. The terms $\lambda_{i,P}$ and $\lambda_{i,C}$ are the fractions of power assigned to the common and private codewords, i.e.,  $\lambda_{i,C} + \lambda_{i,P} \leqslant 1$.  
As suggested in \cite{Suh-TIT-2011}, the fraction $\lambda_{i,P}$ is chosen such that the interference produced at receiver $j$ is at the level of the noise, i.e., $\lambda_{i,P} \INR_{ji} \leqslant 1$ and thus, $\forall i \in \lbrace 1, 2 \rbrace$,
\begin{equation}\label{EqLambdaip}
\lambda_{i,P} =\left\lbrace
\begin{array}{l}
\min\Big(\frac{1}{\INR_{ji}}, 1 \Big), \text{ if } \INR_{ij} < \SNR_{i},\\
0, \qquad \qquad \qquad  \text{ otherwise}.
\end{array}
\right.
\end{equation}
This choice of the power allocation reproduces the main assumption of the linear deterministic model in which the private messages, when they exist, do not appear in the other receiver as they are seen at a lower or equal level than the noise. 
More interestingly, note that by using this power allocation, transmitter $i$ uses message splitting only in the very weak, weak and moderate interference regimes ($\INR_{ij} < \SNR_{i}$). For instance, in the very weak interference regime, e.g., $\INR_{ji} < 1$, no common message is used by transmitter $i$ as the channel between transmitter $i$ and receiver $j$ has a very small capacity and thus,  it privileges the private messages, i.e., $\lambda_{i,P} = 1$. 
In the very strong interference regimes ($\INR_{ij} > \SNR_{i}$), transmitter $i$ uses the alternative path provided by feedback to communicate with receiver $i$, i.e., the link \emph{transmitter~$i$ - receiver~$j$ - transmitter~$j$ - receiver~$i$}, and thus, no private message is used, i.e., $\lambda_{i,P} = 0$.  

The proof of \eqref{EqPartII} is immediate from the following lemmas. In particular,  Lemma \ref{LemmaGICFBb} states that the randomized Han-Kobayashi scheme with feedback achieves all the rate pairs $(R_{1},R_{2}) \in  \underRfb$;
Lemma  \ref{LemmaGICFBc} provides the maximum rate improvement that a given transmitter-receiver pair achieves by unilateral deviation from the randomized Han-Kobayashi scheme with feedback; 
 Lemma \ref{LemmaGICFBd} states that when the rates of the random components $R_{1,R}$ and $R_{2,R}$ are properly chosen, the randomized Han-Kobayashi scheme with feedback forms an $\eta$-NE, with $\eta \geqslant 1$; 
Lemma \ref{LemmaGICFBe} shows that $\underRfb \cap \Bgicfb \subseteq \Ngicfb$; and finally, Lemma \ref{LemmaGICFBf} states that the regions  $\underRfb \cap \Bgicfb$ and $\overRfb \cap \Bgicfb$ approximate the $\eta$-NE region within $1$ bit, and this completes the proof of \eqref{EqPartII}.

\begin{lemma}\label{LemmaGICFBb}
\emph{The achievable region of the randomized Han-Kobayashi coding scheme with feedback in the GIC-FB is the set of non-negative rates 
$(R_{1,C}, R_{1,R},R_{1,P}, R_{2,C}, R_{2,R},R_{2,P} )$ that satisfy, $\forall i \in \lbrace 1, 2 \rbrace$ and $\rho \in [0,1]$,
 \begin{IEEEeqnarray}{l} 
\label{EqAchievableRatesGICFB1}
R_{i,C} + R_{i,R}   \leqslant  \frac{1}{2} \log\big(1+ (1-\rho)\INR_{ji} \big) - \frac{1}{2},\\
\label{EqAchievableRatesGICFB2}
R_{i,P}   \leqslant \frac{1}{2}\log\Big( 2 + \frac{\SNR_{i}}{\INR_{ji}} \Big) - \frac{1}{2}, \\
\label{EqAchievableRatesGICFB3}%
R_{i,C} + R_{i,P}  + R_{j,C} + R_{j,R} \\
\nonumber
\qquad   \leqslant \frac{1}{2}\log\Big(1 + \SNR_{i} + \INR_{ij} + 2 \rho \sqrt{\SNR_{i} \INR_{ij}} \Big) -\frac{1}{2}.
\end{IEEEeqnarray}
}
\end{lemma}
The proof of Lemma \ref{LemmaGICFBb} is presented in Appendix \ref{AppProofOfLemmaGICFBb}.  
The set of inequalities in Lemma \ref{LemmaGICFBb} can be written in terms of $R_1 = R_{1,C} + R_{1,P}$ and $R_2 = R_{2,C} + R_{2,P}$. This yields the following set of conditions:
\begin{IEEEeqnarray}{l}
\label{EqAchievableRatesGICFB1a}
R_{1,R}  \leqslant  \frac{1}{2} \log\Big( 1 + (1-\rho) \INR_{21} \Big) -\frac{1}{2},\\
\nonumber
R_{1} + R_{1,R}  \leqslant  \frac{1}{2}\log\Big( 1 + (1-\rho) \INR_{21} \Big) \\
\label{EqAchievableRatesGICFB2a}
\qquad + \frac{1}{2}\log\Big( 2 + \frac{\SNR_1}{\INR_{21} }\Big)  - 1,\\
\nonumber
R_{1} + R_{2,R}   \leqslant  \frac{1}{2}\log\Big( 1 + \SNR_{1} + \INR_{12} \\
\label{EqAchievableRatesGICFB3a}
\qquad + 2\rho\sqrt{\SNR_1 \INR_{12} } \Big)  -\frac{1}{2},\\ 
\label{EqAchievableRatesGICFB4a}
R_{2,R}  \leqslant \frac{1}{2}\log\Big( 1 + (1-\rho) \INR_{12} \Big) -\frac{1}{2},\\
\nonumber
R_{2} + R_{2,R}  \leqslant \frac{1}{2}\log\Big( 1 + (1-\rho) \INR_{12} \Big) \\
\label{EqAchievableRatesGICFB5a}
\qquad   + \frac{1}{2}\log\Big( 2 + \frac{\SNR_2}{\INR_{12} }\Big) - 1,\\
\nonumber
R_{2} + R_{1,R}  \leqslant \frac{1}{2}\log\Big( 1 + \SNR_{2} + \INR_{21} \\
\label{EqAchievableRatesGICFB6a}
\qquad   + 2\rho\sqrt{\SNR_2 \INR_{21} } \Big) -\frac{1}{2},\\
\label{EqAchievableRatesGICFB7a}
R_{1}  + R_{2} + R_{2,R}  \leqslant  \frac{1}{2}\log\Big( 2 + \frac{\SNR_2}{\INR_{12} }\Big)\\
\nonumber
\qquad +  \frac{1}{2}\log\Big( 1 + \SNR_{1} + \INR_{12} + 2\rho\sqrt{\SNR_1 \INR_{12} } \Big)  - 1,
\end{IEEEeqnarray}
\begin{IEEEeqnarray}{l}
\label{EqAchievableRatesGICFB8a}
R_{2}    + R_{1} + R_{1,R}  \leqslant \frac{1}{2}\log\Big( 2 + \frac{\SNR_1}{\INR_{21} }\Big)\\
\nonumber
\qquad  + \frac{1}{2}\log\Big( 1 + \SNR_{2} + \INR_{21} + 2\rho\sqrt{\SNR_2 \INR_{21} } \Big)  - 1.
\end{IEEEeqnarray}
It is worth noting that for any rate pair $(R_1,R_2)$ that satisfies \eqref{EqUnderRfb1} - \eqref{EqUnderRfb6}, that is $ \forall (R_1,R_2) \in \underRfb$, there always exist a $R_{1,R} \geqslant 0$ and  a $R_{2,R} \geqslant 0$ such that that  $(R_1,R_{1,R},R_2,R_{2,R})$ satisfies \eqref{EqAchievableRatesGICFB1a} - \eqref{EqAchievableRatesGICFB8a}. 
This implies that any rate pair $(R_1,R_2) \in \underRfb$ is also achievable by the randomized Han-Kobayashi scheme with feedback as long as the rates  $R_{1,R}$ and $R_{2,R}$ are properly chosen.

The following lemma determines the maximum rate improvement that can be achieved by a transmitter that unilaterally deviates from a strategy pair in which both transmitters use the randomized Han-Kobayashi scheme with feedback. The statement of the lemma as well as its proof are analogous to Lemma \ref{LemmaLDICFBc} in the LD-IC-FB case.

\begin{lemma}\label{LemmaGICFBc}
\emph{Let the rate tuple $\bs{R} = (R_{1,C}, R_{1,R},R_{1,P},R_{2,C}, R_{2,R} ,R_{2,P})$ be achievable with the randomized Han-Kobayashi coding scheme with feedback such that $R_1 = R_{1,P} + R_{1,C}$ and $R_2 = R_{2,P} + R_{2,C}$. Then, any unilateral deviation of transmitter-receiver pair $i$ by using any other coding scheme leads to a transmission rate $\tilde{R}_i$ that satisfies
\begin{IEEEeqnarray}{ccl}
\nonumber
\tilde{R}_{i} & \leqslant  & \frac{1}{2}\log\big( 1 + \SNR_{i} + \INR_{ij}  + 2 \sqrt{\SNR_i \INR_{ij} }  \big) \\
& & - (R_{j,C} + R_{j,R}).
\end{IEEEeqnarray}
}
\end{lemma}
\begin{IEEEproof}
From Lemma \ref{LemmaGICFBb}, it is known that for all rate tuples $(R_1,R_2) \in \underRfb$ there always exists a rate tuple $\bs{R} = (R_{1,C}, R_{1,R} ,R_{1,P},R_{2,C} , R_{2,R} ,R_{2,P})$ , with $R_1 = R_{1,P} + R_{1,C}$ and $R_2 = R_{2,P} + R_{2,C}$ that satisfies \eqref{EqAchievableRatesGICFB1} - \eqref{EqAchievableRatesGICFB3}.
Assume that both transmitters achieve the rates $\bs{R}$ by using the Han-Kobayashi scheme following the code construction in \eqref{EqConstructionX_i} and the power allocation \eqref{EqGaussianInputDistribution1} - \eqref{EqGaussianInputDistribution3}, in Appendix \ref{AppProofOfLemmaGICFBb}.

Without loss of generality, let transmitter $1$ change its transmit-receive configuration while the transmitter-receiver pair $2$ remains unchanged.
Let $B$ denote the number of blocks sent by both transmitters after the transmit-receive configuration change of transmitter-receiver pair $1$.  
Note that the new transmit-receive configuration of transmitter-receiver pair $1$ can be arbitrary, i.e., it may or may not use feedback, and it may or may not use any random symbols. It can also use a new block length $\tilde{N}_1 \neq N_1$.  Hence, the total duration of the transmission after the deviation is 
\begin{equation}
T=B\max(\tilde{N}_1, N_2)
\end{equation}
channel uses.
Denote by $W_1^{(b)}$ and $\Omega_{1}^{(b)}$ the message index and the random index of transmitter-receiver pair  $1$ during block $b$ after its deviation, with $b \in \lbrace 1, \ldots, B \rbrace$, $\bs{W}_{1} = \left(W_{1}^{(1)}, \ldots, W_{1}^{(B)} \right)$ and $\bs{\Omega}_{1} = \left(\Omega_{1}^{(1)}, \ldots, \Omega_{1}^{(B)} \right)$. 
Let also $\bs{\tilde{X}}_1^{(b)} = \left( \tilde{X}_{1,1}^{(b)}, \ldots, \tilde{X}_{1,\tilde{N}_1}^{(b)} \right)$ and $\bs{\tilde{Y}}_1^{(b)} = \left( \tilde{Y}_{1,1}^{(b)}, \ldots, \tilde{Y}_{1,\tilde{N}_1}^{(b)} \right)$ be the corresponding vector of outputs of transmitter $1$ and inputs  to receiver $1$ during block $b$. Hence, an upper bound for $\tilde{R}_1$ is obtained from the following inequalities:
\begin{align}
\label{EqConditionTildeR1}
& T \tilde{R}_{1}  \nonumber\\
& = H \left( \bs{W}_1 \right) = H \left( \bs{W}_1 | \bs{\Omega}_1 \right)  \nonumber\\
& = I\left( \bs{W}_1; \bs{\tilde{Y}}_{1}^{(1)}, \ldots,\bs{\tilde{Y}}_{1}^{(B)} |  \bs{\Omega}_1 \right) \nonumber \\
& \quad + H \left( \bs{W}_1 |  \bs{\Omega}_1, \bs{\tilde{Y}}_{1}^{(1)}, \ldots,\bs{\tilde{Y}}_{1}^{(B)} \right)  \nonumber\\
& \stackrel{(a)}{\leqslant} I\left( \bs{W}_1; \bs{\tilde{Y}}_{1}^{(1)}, \ldots,\bs{\tilde{Y}}_{1}^{(B)} |  \bs{\Omega}_1 \right) + T \delta_1(N_1')  \nonumber\\
& =  h\left(\bs{\tilde{Y}}_{1}^{(1)}, \ldots,\bs{\tilde{Y}}_{1}^{(B)} |  \bs{\Omega}_1 \right) - h\left(\bs{\tilde{Y}}_{1}^{(1)}, \ldots,\bs{\tilde{Y}}_{1}^{(B)} |  \bs{W}_{1}, \bs{\Omega}_1 \right) \nonumber\\
& \quad + T \delta_1(N_1') \nonumber \\
&  \stackrel{(b)}{\leqslant}   \frac{T}{2} \log\big( 2 \pi e \big(  1 + \SNR_{1} + \INR_{12} + 2 \sqrt{\SNR_1 \INR_{12} }\big)\big) \nonumber
\\
& \quad - h\left(\bs{\tilde{Y}}_{1}^{(1)}, \ldots,\bs{\tilde{Y}}_{1}^{(B)} |  \bs{W}_{1}, \bs{\Omega}_1 \right) + T \delta_1(N_1'),
\end{align}
where, \newline
$(a)$ follows from Fano's inequality, as the rate $\tilde{R}_{1}$ is achievable by assumption, and thus, the all message indices $W_1^{(1)}, \ldots, W_1^{(B)}$ can be reliably decodable from $\bs{\tilde{Y}}_1^{(1)}, \ldots, \bs{\tilde{Y}}_1^{(B)}$ and $\bs{\Omega}_1$ after the deviation of transmitter-receiver pair $1$; and \newline
$(b)$ follows from the fact that  for all $b \in \lbrace 1, \ldots, B\rbrace$ and for all $m \in  \lbrace 1, \ldots, T \rbrace$, $h(\tilde{Y}_{1,m}^{(b)}|\tilde{Y}_{1,1}^{(1)},\ldots,\tilde{Y}_{1,m-1}^{(b)}, \bs{\Omega}_1 ) \leqslant h(\tilde{Y}_{1,m}^{(b)}) \leqslant \frac{1}{2}\log\big( 2 \pi e \big(1 + \SNR_{1} + \INR_{12} + 2 \sqrt{\SNR_1 \INR_{12} } \big) \big)$.

To refine this upper bound, the term $h\left(\bs{\tilde{Y}}_{1}^{(1)}, \ldots,\bs{\tilde{Y}}_{1}^{(B)} |  \bs{W}_{1}, \bs{\Omega}_1 \right)$ in \eqref{EqConditionTildeR1} can be lower bounded.  Denote by $W_{2,C}^{(b)}$ and $\Omega_{2}^{(b)}$ the common message index and the random index of transmitter-receiver pair $2$ during block $b$ after the deviation of transmitter-receiver pair $1$, with $\bs{W}_{2,C} = \left(W_{2,C}^{(1)}, \ldots, W_{2,C}^{(B)} \right)$ and $\bs{\Omega}_{2,C} = \left(\Omega_{2,C}^{(1)}, \ldots, \Omega_{2,C}^{(B)} \right)$. Hence, the following holds:
\begin{IEEEeqnarray}{l}
\label{EqConditionhY1W1}
T (R_{2,C} + R_{2,R})   =  H(\bs{W}_{2,C}, \bs{\Omega}_2)\\
\nonumber
 \stackrel{(d)}{=}  H(\bs{W}_{2,C}, \bs{\Omega}_2 | \bs{W}_{1}, \bs{\Omega}_1) \\
\nonumber
 =  I(\bs{W}_{2,C}, \bs{\Omega}_2; \bs{\tilde{Y}}_{1}^{(1)}, \ldots, \bs{\tilde{Y}}_{1}^{(B)} | \bs{W}_{1}, \bs{\Omega}_1) \\
\nonumber
  + H\left( \bs{W}_{2,C}, \bs{\Omega}_2| \bs{\tilde{Y}}_{1}^{(1)}, \ldots, \bs{\tilde{Y}}_{1}^{(B)}, \bs{W}_1, \Omega_1 \right)\\
\nonumber
 \stackrel{(e)}{=}   I(\bs{W}_{2,C}, \bs{\Omega}_2; \bs{\tilde{Y}}_{1}^{(1)}, \ldots, \bs{\tilde{Y}}_{1}^{(B)} | \bs{W}_{1}, \bs{\Omega}_1) \\
\nonumber
  + H\Big( \bs{W}_{2,C}, \bs{\Omega}_2| \bs{\tilde{Y}}_{1}^{(1)}, \ldots, \bs{\tilde{Y}}_{1}^{(B)}, \bs{W}_1, \Omega_1,  \bs{\tilde{X}}_{1}^{(1)}, \ldots, \bs{\tilde{X}}_{1}^{(B)} \Big)
  \end{IEEEeqnarray}
\begin{IEEEeqnarray}{l}
\nonumber
 \stackrel{(f)}{=}   I(\bs{W}_{2,C}, \bs{\Omega}_2; \bs{\tilde{Y}}_{1}^{(1)}, \ldots, \bs{\tilde{Y}}_{1}^{(B)} | \bs{W}_{1}, \bs{\Omega}_1) \\
\nonumber
  + H\Big( \bs{W}_{2,C}, \bs{\Omega}_2| \bs{\hat{Y}}_{1}^{(1)}, \ldots, \bs{\hat{Y}}_{1}^{(B)}, \bs{W}_1, \Omega_1    \bs{\tilde{X}}_{1}^{(1)}, \ldots, \bs{\tilde{X}}_{1}^{(B)} \Big)\\
\nonumber
 \stackrel{(g)}{\leqslant}  I(\bs{W}_{2,C}, \bs{\Omega}_2; \bs{\tilde{Y}}_{1}^{(1)}, \ldots, \bs{\tilde{Y}}_{1}^{(B)} | \bs{W}_{1}, \bs{\Omega}_1)    + T \delta(N_2)\\
\nonumber
 \stackrel{(h)}{\leqslant}   h(\bs{\tilde{Y}}_{1}^{(1)}, \ldots, \bs{\tilde{Y}}_{1}^{(B)} | \bs{W}_{1}, \bs{\Omega}_1)  + T \left( \delta(N_2) -  \frac{1}{2}\log_{2}(2 \pi e) \right),
\end{IEEEeqnarray}
where, \newline
$(d)$ follows from the independence of the index messages $W_1^{(b)}$, $\Omega_{1}^{(b)}$, $W_2^{(b)}$, $\Omega_{2}^{(b)}$, for all $b \in \lbrace 1, \ldots, B \rbrace$; \newline
$(e)$ follows from the fact that the output of transmitter $i$ at the $m$-th channel use of block $b$ is a deterministic function of $W_{i}^{(b)}$ and the previous channel outputs $\tilde{Y}_{1}^{(1)}, \ldots, \tilde{Y}_{m-1}^{(b)}$; \newline
$(f)$ follows from the signal construction in \eqref{EqGICsignals}, such that $\tilde{Y}_{1,m}^{(b)} = \hat{Y}_{1,m}^{(b)} + h_{11}\tilde{X}_{1,m}^{(b)}$, with $\hat{Y}_{1,m}^{(b)} = h_{12} X_{2,m}^{(b)} + Z_{1,m}^{(b)}$; \newline
$(g)$ follows from Fano's inequality as the message indices $W_{2,C}^{(1)}, \ldots, W_{2,C}^{(B)}$ and $\Omega_{2}^{(1)}, \ldots, \Omega_{2}^{(B)}$ can be reliably decoded by receiver $1$ using the signals $\bs{\hat{Y}}_{1}^{(1)}, \ldots, \bs{\hat{Y}}_{1}^{(B)}$, and $\bs{\Omega}_1$ as transmitter-receiver pair $2$ did not change its transmit-receive configuration and by assumption of the lemma the rate tuple $\bs{R}$ is achievable; and finally, \newline
$(h)$ follows from the fact that $h(\bs{\tilde{Y}}_{1}^{(1)}, \ldots, \bs{\tilde{Y}}_{1}^{(B)} | \bs{W}_{1}, \bs{\Omega}_1, \bs{W}_{2,C}, \bs{\Omega}_2) > \frac{T}{2} \log_{2}(2 \pi e)$. 

Substituting \eqref{EqConditionhY1W1} into \eqref{EqConditionTildeR1}, it follows that:
\begin{IEEEeqnarray}{lcl}
\label{EqConditionNEgicfb4}
\tilde{R}_{1} & \leqslant &   \frac{1}{2} \log\big( 1 + \SNR_{1} + \INR_{12} + 2 \sqrt{\SNR_1 \INR_{12} } \big)  \\
\nonumber
& & - \left( R_{2,C} + R_{2,R}\right) + \delta_{1}(N'_1) +  \delta_2(N_2).
\end{IEEEeqnarray}
Note that $ \delta_{1}(\tilde{N}_1)$ and $\delta_2(N_2)$ are monotonically decreasing functions of $\tilde{N}_1$ and $N_2$, respectively. Hence, in the asymptotic regime, it follows that:
\begin{IEEEeqnarray}{lcl}
\nonumber
\tilde{R}_{1} & \leqslant &  \frac{1}{2}  \log\big( 1 + \SNR_{1} + \INR_{12} + 2 \sqrt{\SNR_1 \INR_{12} } \big)  \\
& & - \left( R_{2,C} + R_{2,R}\right).
\end{IEEEeqnarray}
The same can be proved for the other transmitter-receiver pair $2$ and this completes the proof.
\end{IEEEproof}

Note that if there exists  an $\eta \geqslant 0$ and a  rate tuple $\bs{R} = (R_{1,C} , R_{1,R} ,R_{1,P},R_{2,C} , R_{2,R} ,R_{2,P})$  achievable with the randomized Han-Kobayashi coding scheme with feedback,  such that $\tilde{R}_{i} - (R_{i,C} + R_{i,P}) < \eta$, then the rate pair $(R_1,R_2)$, with $R_1 = R_{1,C} + R_{1,P}$ and $R_2 = R_{2,C} + R_{2,P}$, is achievable at an $\eta$-NE. The following lemma formalizes this observation.

\begin{lemma}\label{LemmaGICFBd}
\emph{Let $\eta \geqslant 1$  and let the rate tuple $\bs{R} = (R_{1,C} , R_{1,R} ,R_{1,P},R_{2,C} , R_{2,R} ,R_{2,P})$ be achievable with the randomized Han-Kobayashi coding scheme with feedback. Let also $\rho \in [0,1]$ and $\forall i \in \lbrace 1, 2 \rbrace$, 
\begin{IEEEeqnarray}{l}
\label{EqConditionNEgicfb1}
R_{i,P} + R_{i,C}  \geqslant  \frac{1}{2} \log\left( 1 + \frac{\SNR_{i}}{1 +\INR_{ij}}\right) \\
\label{EqConditionNEgicfb2}
R_{i,C}  + R_{i,P} + R_{j,C} + R_{j,R}  \\
\nonumber
\qquad =  \frac{1}{2} \log\big( 1 + \SNR_{i}  + \INR_{ij}  + 2 \rho \sqrt{\SNR_i \INR_{ij} } \big) -\frac{1}{2}  .
\end{IEEEeqnarray}
Then, the rate pair $(R_1,R_2)$, with $R_i = R_{i,C} + R_{i,P}$ is a utility pair achieved at an $\eta$-NE.
}
\end{lemma}
The proof of Lemma \ref{LemmaGICFBd} follows the same steps as in the proof of Lemma \ref{LemmaLDICFBd}.
\begin{IEEEproof}
Condition \eqref{EqConditionNEgicfb1} is a necessary condition for the rate tuple $(R_1,R_2)$ to be an $\eta$-NE, for any value of $\eta \geqslant 1$ (Lemma \ref{LemmaGICFBa}). 
 Let $s_i^* \in \mathcal{A}_i$ be a transmit-receive configuration in which communication takes place using the randomized Han-Kobayashi scheme with feedback and $R_{1,R}$ and $R_{2,R}$ are chosen according to condition \eqref{EqConditionNEgicfb2}, with $i=1$ and $i=2$, respectively. From the assumptions of the lemma such a transmit-receive configuration pair  $(s_1^*,s_2^*)$ is an $\eta$-NE and 
\begin{IEEEeqnarray}{l}
\label{EqMaxUtilityNEa}
 u_i(s_1^*,s_2^*)  = R_i  =  R_{i,C} + R_{i,P} \\
\nonumber
 = \frac{1}{2}  \log\big( 1 + \SNR_{i}  + \INR_{ij}  + 2 \rho \sqrt{\SNR_i \INR_{ij} } \big) \\
 \nonumber
 - (R_{j,C} + R_{j,R}) -\frac{1}{2} ,
\end{IEEEeqnarray}  
where the last equality holds from \eqref{EqConditionNEgicfb2}.
Then, from Def. \ref{DefEtaNE}, it holds that for all $i \in \lbrace 1, 2 \rbrace$ and for all transmit-receive configuration $s_i \neq s_i^* \in \mathcal{A}_i$, the utility improvement is upper bounded by $\eta$, that is,  
\begin{equation}
\label{EqNEdiff}
u_i(s_i,s_j^*)  - u_i(s_i^*,s_j^*) \leqslant \eta. 
\end{equation}
Without loss of generality, let $ i = 1$ be the deviating transmitter-receiver pair and assume it achieves the highest improvement (Lemma \ref{LemmaGICFBc}), that is, 
\begin{IEEEeqnarray}{lcl}
\nonumber
u_1(s_1,s_2^*) & = & \frac{1}{2}  \log\big( 1 + \SNR_{1} + \INR_{12} + 2 \sqrt{\SNR_1 \INR_{12} } \big) \\
\label{EqMaxUtilityNE}
& &  - \left( R_{2,C} + R_{2,R}\right).
\end{IEEEeqnarray}
Hence, replacing \eqref{EqMaxUtilityNEa} and \eqref{EqMaxUtilityNE} in \eqref{EqNEdiff}, it yields 
\begin{IEEEeqnarray}{l}
\label{EqConditionNEinGICFB}
  u_1(s_1,s_2^*) - u_1(s_1^*,s_2^*) \nonumber\\
 =   \frac{1}{2}  \log\big( 1 + \SNR_{1} + \INR_{12} + 2 \sqrt{\SNR_1 \INR_{12} } \big) \nonumber
\\
  - \frac{1}{2}  \log\big( 1 + \SNR_{1} + \INR_{12} + 2 \rho \sqrt{\SNR_1 \INR_{12} } \big) + \frac{1}{2}  \nonumber
\\
 \stackrel{(a)}{\leqslant}  1  \leqslant  \eta,
\end{IEEEeqnarray}
where $(a)$ follows from the fact that $\Delta = \frac{1}{2}  \log\big( 1 + \SNR_{1} + \INR_{12} + 2 \sqrt{\SNR_1 \INR_{12} } \big) - \frac{1}{2}  \log\big( 1 + \SNR_{1} + \INR_{12} + 2 \rho \sqrt{\SNR_1 \INR_{12} } \big) + 1$ satisfies the following inequality:
\begin{IEEEeqnarray}{lcl}
\nonumber
\Delta & = & \frac{1}{2}  \log\left( 1 + \frac{2 (1-\rho) \sqrt{\SNR_1 \INR_{12} } }{1 + \SNR_{1} + \INR_{12} + 2 \rho \sqrt{\SNR_1 \INR_{12} } }\right) \\
\nonumber
&  &  + \frac{1}{2}  \\
\nonumber
&\leqslant & \frac{1}{2}  \log\left( 1 + \frac{2 \sqrt{\SNR_1 \INR_{12} } }{1 + \SNR_{1} + \INR_{12} }\right)  + \frac{1}{2} \\
\nonumber
& \leqslant &  \frac{1}{2}  \log\left( 1 + \frac{\SNR_{1} + \INR_{12}}{1 +\SNR_{1} + \INR_{12} }\right) + \frac{1}{2}  \\
\nonumber
& \leqslant &  \frac{1}{2}  \log\left(2\right) + \frac{1}{2} \\
& = & 1.
\end{IEEEeqnarray}
This verifies that any rate improvement by unilateral deviation of the transmit-receive configuration $(s_1^*,s_2^*)$ is upper bounded by any $\eta$ arbitrarily close to $1$, i.e., $\eta \geqslant 1$.
The same can be proved for the other transmitter-receiver pair and this completes the proof.
\end{IEEEproof}
 
The following lemma shows that all the rate pairs $(R_1,R_2) \in \underRfb \cap \Bgicfb$ can be achieved at an $\eta$-NE, with $\eta \geqslant 1$, by a randomized Han-Kobayashi scheme with feedback, and thus, $\underRfb \cap \Bgicfb \subseteq \Ngicfb$.
\begin{lemma}\label{LemmaGICFBe}\emph{
For all rate pairs $(R_{1}, R_{2}) \in  \underRfb \cap \Bgicfb$, there always exists at least one $\eta$-NE transmit-receive configuration pair $(s_1^*,s_2^*) \in \mathcal{A}_1\times\mathcal{A}_2$, with $\eta \geqslant 1$, such that $u_1(s_1^*,s_2^*) = R_{1}$ and $u_2(s_1^*,s_2^*) = R_{2} $.
}
\end{lemma}
\begin{proof}
A rate tuple $(R_{1,C}, R_{1,R}, R_{1,P},R_{2,C}, R_{2,R}, R_{2,P})$ that is achievable with the randomized Han-Kobayashi scheme with feedback satisfies the inequalities in \eqref{EqAchievableRatesGICFB1} - \eqref{EqAchievableRatesGICFB3}. Additionally, any rate tuple $(R_{1,C}, R_{1,R}, R_{1,P},R_{2,C}, R_{2,R}, R_{2,P})$ that satisfies \eqref{EqAchievableRatesGICFB1} - \eqref{EqAchievableRatesGICFB3}, \eqref{EqConditionNEgicfb1} and \eqref{EqConditionNEgicfb2} is an $\eta$-NE (Lemma \ref{LemmaGICFBd}). 
A Fourier-Motzkin elimination of inequalities \eqref{EqAchievableRatesGICFB1} - \eqref{EqAchievableRatesGICFB3}, \eqref{EqConditionNEgicfb1} and \eqref{EqConditionNEgicfb2} leads to the following set of inequalities:
\begin{IEEEeqnarray}{l}
\label{EqNEGICFB1}
R_{1,R}  \leqslant \frac{1}{2}  \log\left( 1 + (1-\rho) \INR_{21} \right) - \frac{1}{2}  \\
\label{EqNEGICFB2}
R_{2,R}  \leqslant  \frac{1}{2}  \log\left( 1 + (1-\rho) \INR_{12} \right) -\frac{1}{2}   \\
\label{EqNEGICFB3}
R_{1} + R_{2,R}  \leqslant  \\
\nonumber
\quad  \frac{1}{2}  \log\big( 1 + \SNR_1 + \INR_{12}  + 2 \rho \sqrt{\SNR_1 \INR_{12}} \big) -\frac{1}{2}  \\
\nonumber
R_{1} + R_{1,R}  \leqslant  \frac{1}{2} \log \left(2 + \frac{\SNR_1}{\INR_{21}} \right) \\
\label{EqNEGICFB4}
\quad   +  \frac{1}{2} \log\left( 1 + (1-\rho) \INR_{21} \right) -1 \\
\label{EqNEGICFB5}
R_{1}   \geqslant   \frac{1}{2} \log\left( 1 + \frac{ \SNR_{1}}{1 + \INR_{12}} \right)  \\ 
\nonumber
R_{1} + R_{2,R}  \geqslant  \\
\nonumber
\quad \frac{1}{2} \log\big( 1 + \SNR_1 + \INR_{12} + 2 \rho \sqrt{\SNR_1 \INR_{12}} \big)\\
\label{EqNEGICFB6}
\quad - \frac{1}{2}  \log\left( 1 + (1-\rho) \INR_{12} \right), \\
\nonumber
R_{2} + R_{1,R}  \leqslant  \\
\nonumber
\quad  \frac{1}{2} \log\big( 1 + \SNR_2 + \INR_{21} + 2 \rho \sqrt{\SNR_2 \INR_{21}} \big)\\
\label{EqNEGICFB7}
\quad   -\frac{1}{2},\\
\nonumber
R_{2} + R_{2,R}  \leqslant  \frac{1}{2}  \log \left(2 + \frac{\SNR_2}{\INR_{12}} \right)\\
\label{EqNEGICFB8}
\qquad    + \frac{1}{2} \log\left( 1 + (1-\rho) \INR_{12} \right)  -\frac{1}{2} \\
\label{EqNEGICFB9}
R_{2}   \geqslant   \frac{1}{2} \log\left( 1 + \frac{ \SNR_{2}}{1 + \INR_{21}} \right) \\ 
\nonumber
R_{2} + R_{1,R}  \geqslant \frac{1}{2}   \\
\nonumber
\qquad \log\big( 1 + \SNR_2 + \INR_{21}+ 2 \rho \sqrt{\SNR_2 \INR_{21}} \big)\\
\label{EqNEGICFB10}
\qquad  -\frac{1}{2}  \log\left( 1 + (1-\rho) \INR_{21} \right) \\
\label{EqNEGICFB11}
R_{1} + R_{2} + R_{2,R}  \geqslant   \frac{1}{2}  \log \left(2 + \frac{\SNR_2}{\INR_{12}} \right) \\
\nonumber
\qquad + \frac{1}{2} \log\big( 1 + \SNR_1 + \INR_{12}   + 2 \rho \sqrt{\SNR_1 \INR_{12}} \big) -1 \\
\label{EqNEGICFB12}
R_{1} + R_{2} + R_{1,R}  \geqslant \frac{1}{2}  \log \left(2 + \frac{\SNR_1}{\INR_{21}} \right) \\
\nonumber
\qquad   +\frac{1}{2}  \log\big( 1 + \SNR_2 + \INR_{21}  + 2 \rho \sqrt{\SNR_2 \INR_{21}} \big) -1.
\end{IEEEeqnarray}
Finally, it is verified by inspection that for any rate pair $(R_1, R_2)  \in  \underRfb \cap  \Bgicfb$, there always exists a $R_{1,R} \geqslant 0$ and a $R_{2,R} \geqslant 0$ such that inequalities \eqref{EqNEGICFB1} - \eqref{EqNEGICFB12} are also satisfied. Thus, $\forall (R_{1}, R_{2}) \in  \underRfb \cap \Bgicfb$, the rate pair $(R_{1}, R_{2})$ can be achieved at an $\eta$-NE, with $\eta \geqslant 1$. This completes the proof.
\end{proof}
It is worth highlighting that for all rate pairs $\left( R_1, R_2 \right) \in \underRfb \cap \Bgicfb$, it also holds that $\left( R_1 + 2, R_2  + 2\right) \in \overRfb \cap \Bgicfb$. This is formally stated by the following lemma.
\begin{lemma}[$\eta$-NE region within $1$ bit] \label{LemmaGICFBf}\emph{
Let $\eta \geqslant 1$.
Then, the  $\eta$-NE region $\Ngicfb$ of the Gaussian interference channel with perfect output feedback is approximated to within $1$ bit  by the regions $\underRfb \cap \Bgicfb$ and $\overRfb \cap \Bgicfb$. 
}
\end{lemma}
\begin{proof}
From Theorem $4$ in \cite{Suh-TIT-2011}, it follows that $\forall (R_1,R_2) \in \overRfb$ it holds that $(R_1 - 1,R_2 -1) \in \underRfb$. Hence, since $\underRfb \cap \Bgicfb \subseteq \underRfb$ and $\overRfb \cap \Bgicfb \subseteq \overRfb$, it follows that $\forall (R_1,R_2) \in \overRfb\cap \Bgicfb$ it holds that $(R_1 - 1,R_2 - 1) \in \underRfb \cap \Bgicfb$. This completes the proof.
\end{proof}

\section{Conclusions}\label{SecConclusions}

In this paper, the $\eta$-NE region of the two-user decentralized linear deterministic interference channel with feedback has been fully characterized for $\eta \geqslant 0$ arbitrarily small. Using this initial result, the $\eta$-NE region of the two-user decentralized Gaussian interference channel with feedback has been approximated to within $1$ bit for $\eta \geqslant 1$. 
The main conclusions of this work are twofold: $(i)$ The $\eta$-NE region achieved with feedback is larger than or equal to the NE region without feedback. $(ii)$ The use of feedback allows the achievability at an $\eta$-NE of rate pairs that are at a maximum gap of $1$-bit per channel use of the outer bounds of the capacity region even when the network is fully decentralized. 

Some interesting questions that remain unsolved and become directions for future work are the following: 
$(i)$  Theorem \ref{TheoremMainResultLDICFB} and Theorem \ref{TheoremMainResultGICFB} do not provide any insight into the uniqueness of the $\eta$-NE strategy pair $(s_1, s_2)$ that generates the rate pair $(R_1(s_1,s_2),R_2(s_1,s_2))$. Indeed, it is possible that another equilibrium transmit-receive configuration pair $(s_1',s_2')$ generates the same rate pair, i.e.,  $(R_1(s_1,s_2),R_2(s_1,s_2)) = (R_1(s_1',s_2'),R_2(s'_1,s_2'))$. Theorem \ref{TheoremMainResultLDICFB} and Theorem \ref{TheoremMainResultGICFB} do not characterize the set of NE transmit-receive configurations but the set of rate pairs observed at any of the possible $\eta$-NE of the decentralized GIC with feedback. The characterization of all the transmit-receive configuration pairs $(s_1,s_2)$ that are $\eta$-NEs is still an open problem.
$(ii)$ An interesting question that remains open is how to achieve an NE in the case in which players do not know the set of strategies of all the other players and only local information is available. The intuition obtained from the results presented here leads to imply that other tools different from those brought from information theory and game theory would be needed to solve these questions. 
Finally, $(iii)$ despite recent results in the centralized case \cite{Quintero-ITW-2015, SyQuoc-TIT-2013, Quintero-TR-2015} and decentralized case \cite{Perlaza-ICSSP-2014, Perlaza-ISIT-2014}, it is worth mentioning that very little can be said about the benefits of feedback in the NE region of interference channel when the feedback links are impaired by noise or when the number of transmitter-receiver pairs is larger than two.

\appendices
\section{Proof of Lemma \ref{LemmaSumMaxNE}}\label{AppProofOfLemmaSumMaxNE}

 \begin{figure}[t]
 \centerline{\epsfig{figure=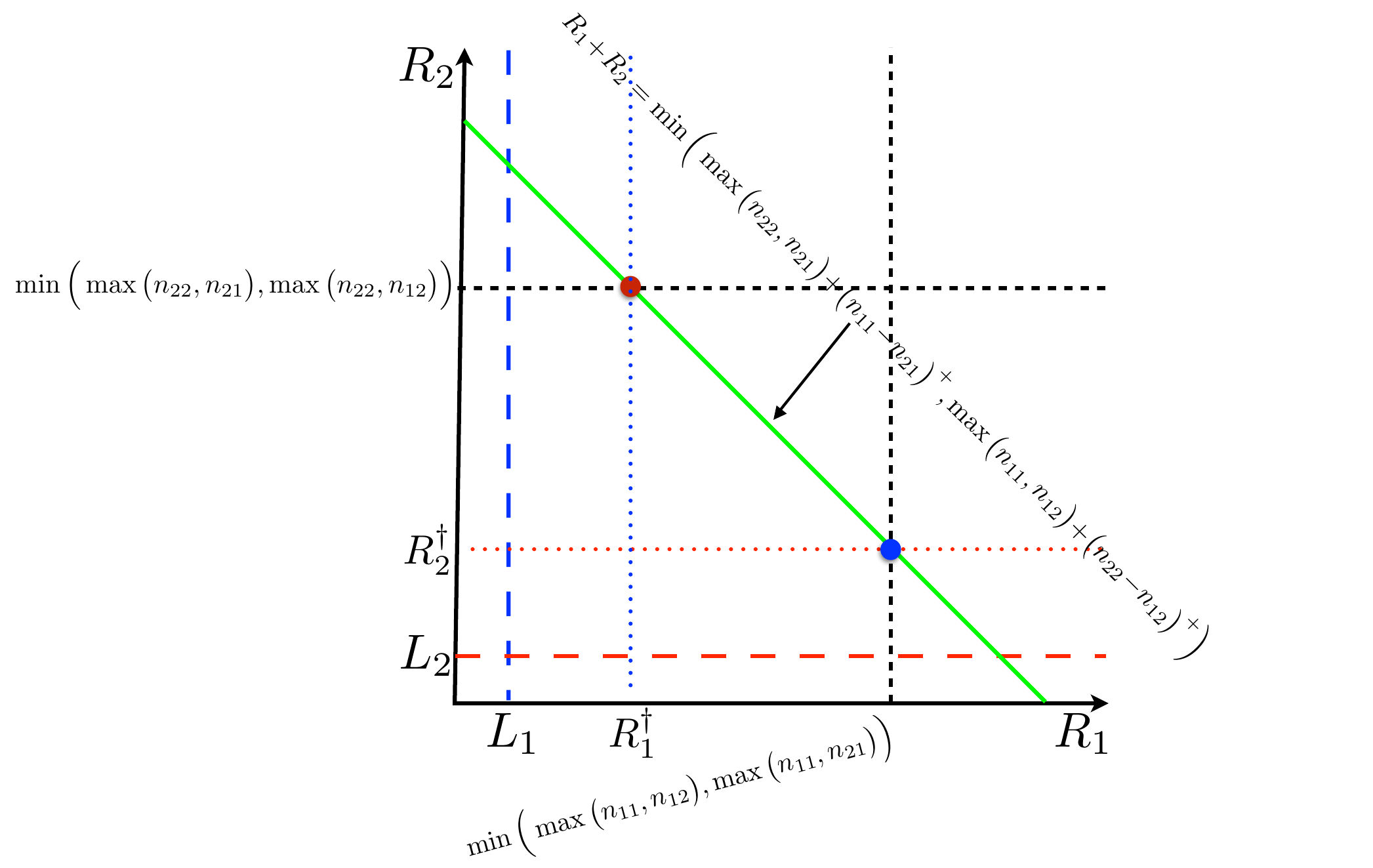,width=.4\textwidth}}
  \caption{Ilustration of the rates $R_{1}^{\dagger}$ and $R_{2}^{\dagger}$.}
  \label{FigProofSumRateNE}
\end{figure}

This appendix provides a proof of Lemma \ref{LemmaSumMaxNE}. 
The proof relies on the fact that the capacity region $\Cldicfb$ and the $\eta$-NE region $\Nldicfb$ are defined by the intersection of at most $5$ half-planes:  one upper bound and one lower bound (positive rates) for each individual rate, $R_1$ and $R_2$; and one upper-bound on the sum  $R_1+R_2$. The intersection of these half-planes defines a polygon that can be a pentagon, a quadrilateral or a triangle as shown in Fig. \ref{FigProofSumRateNE}.  The upper-bound $R_i \leqslant U_i$ is never strictly active. In particular, note that the hyperplane that defines the upper-bound on $R_1+R_2$ satisfies the following equalities:
\begin{IEEEeqnarray}{lcl}
\label{EqHyperplane3}
R_1+R_2&  =  \min\Big( &  \max\big(n_{22}, n_{21}\big) + \big(n_{11} - n_{21}\big)^+ , \\
\nonumber
& & \max\big(n_{11}, n_{12}\big) + \big(n_{22} - n_{12}\big)^+ \Big),\\
\nonumber
& = \min\Big( &  \max\big(n_{22}, n_{21}\big) + \big(n_{11} - n_{21}\big)^+ , \\
\nonumber
& & \max\big(n_{22}, n_{12}\big) + \big(n_{11} - n_{12}\big)^+ \Big)\\
\nonumber
& = \min\Big( & \max\big(n_{11}, n_{21}\big) + \big(n_{22} - n_{21}\big)^+ , \\
\nonumber
& & \max\big(n_{11}, n_{12}\big) + \big(n_{22} - n_{12}\big)^+ \Big).
\end{IEEEeqnarray}
Denote by $R_1^{\dagger}$ (resp. $R_2^{\dagger}$), the maximum rate at which transmitter $1$ (resp. transmitter $2$) can transmit information to receiver $1$ (resp. receiver $2$) when transmitter $2$ (resp. transmitter $1$) transmits information at a rate $R_2 = \min\Big( \max\big( n_{22}, n_{21} \big), \max\big( n_{22}, n_{12}\big) \Big)$ (resp. $R_1 = \min\Big( \max\big( n_{11}, n_{12} \big), \max\big( n_{11}, n_{21}\big) \Big)$); see Fig. \ref{FigProofSumRateNE}.

Hence, all the sum-optimal rate pairs of $\Cldicfb$ are included in the $\eta$-NE region $\Nldicfb$ if $L_1 \leqslant R_{1}^{\dagger}$ and $L_2 \leqslant R_{2}^{\dagger}$, where, 
\begin{IEEEeqnarray}{lcl}
R_{i}^{\dagger} & = & \min\Big( \max(n_{jj}, n_{ij}) + \left(n_{ii} - n_{ij} \right)^+ ,  \max(n_{jj},n_{ji})  \\
\nonumber
& & + \left(n_{ii} - n_{ji}\right)^+  \Big) - \min\big( \max(n_{jj}, n_{ji}),\max(n_{jj}, n_{ij})   \big).
\end{IEEEeqnarray}
Hence, $L_i \leqslant R_{i}^{\dagger}$ if and only if
\begin{IEEEeqnarray}{lcl}
\nonumber
 \left(n_{ii} - n_{ij} \right)^+ -  \left(n_{ii} - n_{ji}\right)^+  & \leqslant & \Big( (n_{ji} - n_{jj}) - (n_{ij} - n_{jj}) \Big)^+,\\
\end{IEEEeqnarray}
which completes the proof.

\section{Proof of Lemma \ref{LemmaLDICFBb} }\label{AppProofOfLemmaLDICFBb}
This appendix provides a description of the randomized Han-Kobayashi scheme with feedback used in Sec. \ref{SecLDICFB} and provides a proof of Lemma \ref{LemmaLDICFBb}.

\textbf{Codebook Generation:} 
Fix a joint probability distribution $p(U, U_1,U_2, X_1, X_2) = p(U)p(U_1|U)p(U_2 | U)p(X_1|U,U_1)p(X_2|U,U_2)$.
Generate $2^{N (R_{1,C} +R_{1,R} + R_{2,C} +R_{2,R})}$ i.i.d. $N$-length codewords $\bs{u}(s, r) = \left(u_{1}(s,r), \ldots, u_{N}(s,r)\right)$ according to $p(\bs{u}(s,r)) = \ds\prod_{m =1}^N p(u_{m}(s,r))$, with $s \in \lbrace 1, \ldots, 2^{N (R_{1,C} +R_{1,R} )}\rbrace$ and $r \in \lbrace 1, \ldots, 2^{N ( R_{2,C} +R_{2,R} )}\rbrace$. 

\noindent
For encoder $1$, for each codeword $\bs{u}(s,r)$,  generate  $2^{N (R_{1,C} +R_{1,R})}$ i.i.d. $N$-length codewords $\bs{u}_1(s,r,k) = \left(u_{1,1}(s,r,k), \ldots, u_{1,N}(s,r,k)\right)$  according to $p(\bs{u}_1(s,r,k)|\bs{u}(s,r)) = \ds\prod_{m =1}^N p(u_{1,m}(s,r,k)|u_{m}(s,r))$, with $k \in \lbrace 1, \ldots, 2^{N (R_{1,C} +R_{1,R} )}\rbrace$.  For each pair of codewords $(\bs{u}(s,r),\bs{u}_{1}(s,r,k))$, generate  $2^{N(R_{1,P})}$ i.i.d. $N$-length codewords $\bs{x}_1(s,r,k,l) = \left(x_{1,1}(s,r,k,l), \ldots, x_{1,N}(s,r,k,l)\right)$ according to $p\left(\bs{x}_1(s,r,k,l)|\bs{u}(s,r),\bs{u}_{1}(s,r,k)\right) = \ds\prod_{m =1}^N p(x_{1,m}(s,r,k,l)|u_{m}(s,r),u_{1,m}(s,r,k))$, with $l \in \lbrace 1, \ldots, 2^{N R_{1,P}}\rbrace$. 

For encoder $2$, for each codeword $\bs{u}(s,r)$,  generate $2^{N (R_{2,C} +R_{2,R})}$ i.i.d. $N$-length codewords $\bs{u}_2(s,r,q) = \left(u_{2,1}(s,r,q), \ldots, u_{2,N}(s,r,q)\right)$  according to $p(\bs{u}_2(s,r,q)|\bs{u}(s,r)) = \ds\prod_{m =1}^N p(u_{2,m}(s,r,q)|u_{m}(s,r))$, with $q \in \lbrace 1, \ldots, 2^{N(R_{2,C} +R_{2,R} )}\rbrace$.  For each pair of codewords $(\bs{u}(s,r),\bs{u}_{2}(s,r,q))$, generate  $2^{N(R_{2,P})}$ i.i.d. $N$-length codewords $\bs{x}_2(s,r,q,z) = \left(x_{2,1}(s,r,q,z), \ldots, x_{2,N}(s,r,q,z)\right)$ according to $p(\bs{x}_{2}(s,r,q,z)|\bs{u}(s,r),\bs{u}_{2}(s,r,q)) = \ds\prod_{m =1}^N p(x_{2,m}(s,r,q,z)|u_{m}(s,r),u_{2,m}(s,r,q))$, with $z \in \lbrace 1, \ldots, 2^{NR_{2,P}}\rbrace$.\\
\textbf{Encoding}: 
Denote by $W_{i}^{(t)} \in \lbrace 1, \ldots, 2^{N (R_{i,C} + R_{i,P})}\rbrace$ and $\Omega_{i}^{(t)} \in \lbrace 1, \ldots, 2^{N R_{i,R} } \rbrace$ the message index and the random message index of transmitter $i$ during block $t$, respectively.  
Let $W_{i}^{(t)} = (W_{i,C}^{(t)},W_{i,P}^{(t)})$ be the message index composed of the common message index $W_{i,C}^{(t)}\in \lbrace 1, \ldots, 2^{N R_{i,C}} \rbrace$ and private message $W_{i,P}^{(t)}\in \lbrace 1, \ldots, 2^{N R_{i,P}} \rbrace$. 
Let also $ \left(W_{i,C}^{(t)},\Omega_{i}^{(t)} \right) \in \lbrace 1, \ldots, 2^{N (R_{i,C}  + R_{i,R})} \rbrace$ be the (joint) index of the common message and common random message sent by transmitter $i$ during block $t$, respectively.  
Consider Markov encoding with a length of $T$ blocks. 
At block $t \in \lbrace 1, \ldots, T \rbrace$, transmitter $1$ sends the codeword $\bs{x}_1^{(t)} = \bs{x}_1\left( (W_{1,C}^{(t-1)},\Omega_{1}^{(t-1)}) (W_{2,C}^{(t-1)},\Omega_{2}^{(t-1)}), (W_{1,C}^{(t)},\Omega_{1}^{(t)}), W_{1,P}^{(t)}\right)$, where $W_{1,C}^{(0)} = s^*$, $\Omega_{1}^{(0)} = \omega_{1}^*$, $W_{2,C}^{(0)} = r^*$, $\Omega_{2}^{(0)} = \omega_{2}^*$, $W_{1,C}^{(T)} = k^*$, $\Omega_{1}^{(T)} = \omega_{1}^*$, $W_{2,C}^{(T)} = l^*$ and $\Omega_{2}^{(T)} = \omega_{2}^*$. The $4$-tuple $(s^*,k^*,\omega_{1}^*,r^*,l^*,\omega_{2}^*) \in  \lbrace 1, \ldots, 2^{N R_{1,C}} \rbrace^2 \times \lbrace 1, \ldots, 2^{N R_{1,R}} \rbrace \times \lbrace 1, \ldots, 2^{N R_{2,C}} \rbrace^2 \times \lbrace 1, \ldots, 2^{N R_{2,R}} \rbrace$  is predefined and known at both receivers and transmitters. 
It is worth noting that the message index $W_{2,C}^{(t-1)}$ is obtained by transmitter $1$ via the feedback of $\bs{y}_{1}^{(t-1)}$ at the end of block $t-1$. That is, for $t > 1$, $(W_{2,C}^{(t-1)}, \Omega_{2}^{(t-1)}) = (\hat{W}_{2,C}^{(t-1)},\hat{\Omega}_{2}^{(t-1)}) \in \lbrace 1, \ldots, 2^{n(R_{2c} + R_{2r})}\rbrace$, which are the only indices that satisfy
\begin{IEEEeqnarray}{ll}
\label{EqDecodingW2r}
\Bigg(& \bs{u}\Big((W_{1,C}^{(t-2)},\Omega_{1}^{(t-2)}) , (W_{2,C}^{(t-2)},\Omega_{2}^{(t-2)})\Big), \\
\nonumber
& \bs{u}_1  \Big((W_{1,C}^{(t-2)},\Omega_{1}^{(t-2)}) , (W_{2,C}^{(t-2)},\Omega_{2}^{(t-2)}), (W_{1,C}^{(t-1)},\Omega_{1}^{(t-1)})\Big),  \\
\nonumber 
& \bs{x}_1 \Big((W_{1,C}^{(t-2)},\Omega_{1}^{(t-2)}) , (W_{2,C}^{(t-2)},\Omega_{2}^{(t-2)}), (W_{1,C}^{(t-1)},\Omega_{1}^{(t-1)}),\\
\nonumber
& W_{1,P}^{(t-1)}\Big), \bs{u}_2\Big((W_{1,C}^{(t-2)},\Omega_{1}^{(t-2)}),  (W_{2,C}^{(t-2)},\Omega_{2}^{(t-2)}),  \\
\nonumber
&(\hat{W}_{2,C}^{(t-1)}, \hat{\Omega}_{2}^{(t-1)}) \Big),
\bs{y}_1^{(t-1)}  \Bigg) \in   \mathcal{A}_{e}^{(N)},
\end{IEEEeqnarray}
where $\mathcal{A}_{e}^{(n)}$ represents the set of jointly typical sequences under the assumption that $W_{2,C}^{(t-2)}$ and $\Omega_{2}^{(t-2)}$ have been decoded without errors at transmitter $1$ during the previous block. Transmitter $1$ knows $W_{1,C}^{(t-2)}$, $\Omega_{1}^{(t-2)}$, $W_{2,C}^{(t-2)}$, $\Omega_{2}^{(t-2)}$, $W_{1,C}^{(t-1)}$ and $\Omega_{1}^{(t-1)}$. 
Moreover, the random message indices $\Omega_{1}^{(1)},\ldots, \Omega_{1}^{(T)}$ are known by transmitter and receiver $1$ and thus, it does not represent an effective transmission of information. Conversely, indices $\Omega_{2}^{(2)},\ldots, \Omega_{2}^{(T-1)}$ are ignored and thus, must be estimated.
 Transmitter $2$ follows a similar encoding scheme.\\
\textbf{Decoding}: Both receivers decode their messages at the end of block $T$ in a backward decoding fashion. At each step $t \in \lbrace 1, \ldots, T \rbrace$, receiver $1$ estimates the message indices $(\hat{W}_{1,C}^{(T-t)},\Omega_{1}^{(T-t)}, \hat{W}_{2,C}^{(T-t)},\hat{\Omega}_{2}^{(T-t)}, \hat{W}_{1,P}^{(T-t)}) \in \lbrace 1, \ldots, 2^{N R_{1,C}}\rbrace \times \lbrace 1, \ldots, 2^{N R_{1,R}}\rbrace \times \lbrace 1, \ldots, 2^{N R_{2,C}}\rbrace \times \lbrace 1, \ldots, 2^{N R_{2,R}}\rbrace \times \lbrace 1, \ldots, 2^{N R_{1,P}}\rbrace$. The tuple $(\hat{W}_{1,C}^{(T-t)},\Omega_{1}^{(T-t)}, \hat{W}_{2,C}^{(T-t)},\hat{\Omega}_{2}^{(T-t)}, \hat{W}_{1,P}^{(T-t)})$ is the unique tuple that satisfies
\begin{IEEEeqnarray}{lll}
\label{EqDecodingW1cW2cW1p}
\Bigg(& \bs{u}\Big( & (\hat{W}_{1,C}^{(T-t)}, \Omega_{1}^{(T-t)}) , (\hat{W}_{2,C}^{(T-t)},\hat{\Omega}_{2}^{(T-t)})\Big), \\
\nonumber
& \bs{u}_1\Big(& (\hat{W}_{1,C}^{(T-t)}, \Omega_{1}^{(T-t)}) , (\hat{W}_{2,C}^{(T-t)},\hat{\Omega}_{2}^{(T-t)}), \\
\nonumber
& &(W_{1,C}^{(T-(t-1))}, \Omega_{1}^{(T-(t-1))})\Big),  \\
\nonumber 
& \bs{x}_1\Big(& (\hat{W}_{1,C}^{(T-t)},\Omega_{1}^{(T-t)}) , (\hat{W}_{2,C}^{(T-t)},\hat{\Omega}_{2}^{(T-t)}),  \\
\nonumber
& & (W_{1,C}^{(T-(t-1))},\Omega_{1}^{(T-(t-1))}), \hat{W}_{1,P}^{(T-t)}\Big),  \\
\nonumber
&\bs{u}_2\Big(& (\hat{W}_{1,C}^{(T-t)},\Omega_{1}^{(T-t)}), (\hat{W}_{2,C}^{(T-t)},\hat{\Omega}_{2}^{(T-t)}),  \\
\nonumber
& & (W_{2,C}^{(T-(t-1))}, \Omega_{2}^{(T-(t-1))}) \Big), \bs{y}_1^{(T-t)}  \Bigg)  \in \mathcal{A}_{e}^{(N)},
\end{IEEEeqnarray}
where $\mathcal{A}_{e}^{(n)}$ represents the set of jointly typical sequences.
Receiver $2$ follows a similar decoding scheme.\\
\textbf{Probability of Error Analysis:} An error might occur during the coding phase at the beginning of block $t$ if the common message index $W_{j,C}^{(t-1)}$ is not correctly decoded at transmitter $i$. For instance, this error might occur at transmitter $1$ because: $(i)$ there does not exist an index $\hat{r} \in \lbrace 1, \ldots, 2^{N (R_{2,C} + R_{2,R} )}\rbrace$ that satisfies \eqref{EqDecodingW2r}, or $(ii)$ several indices simultaneously satisfy \eqref{EqDecodingW2r}. From the asymptotic equipartion property (AEP) \cite{Cover-Book-1991}, the probability of error due to $(i)$ tends to zero when $N$ grows to infinity. The probability of error due to $(ii)$ can be made arbitrarily close to zero when $N$ grows to infinity, if  
\begin{eqnarray}\label{EqConditionNoError1}
R_{i,C}+R_{i,R} & \leqslant &  I\left( U_i ; Y_j | X_j, U \right).
\end{eqnarray} 

An error might occur during the (backward) decoding step $t$ if the messages $W_{1,C}^{(T-t)}$, $W_{2,C}^{(T-t)}$ and $W_{1,P}^{(T-t)}$ are not decoded correctly given that the message indices $W_{1,C}^{(T-(t-1))}$ and $W_{2,C}^{(T-(t-1))}$ were correctly decoded in the previous decoding step $t-1$. 
These errors might arise for two reasons: $(i)$ there does not exist a tuple $(\hat{W}_{1,C}^{(T-t)},\Omega_{1}^{(T-t)}, \hat{W}_{2,C}^{(T-t)},\hat{\Omega}_{2}^{(T-t)}, \hat{W}_{1,P}^{(T-t)})$ that satisfies \eqref{EqDecodingW1cW2cW1p}, or $(ii)$ there exist several tuples $(\hat{W}_{1,C}^{(T-t)},\Omega_{1}^{(T-t)}, \hat{W}_{2,C}^{(T-t)},\hat{\Omega}_{2}^{(T-t)}, \hat{W}_{1,P}^{(T-t)})$ that simultaneously satisfy \eqref{EqDecodingW1cW2cW1p}. From the AEP, the probability of an error at receiver $1$ due to $(i)$ tends to zero when $N$ tends to infinity. Consider the error due to $(ii)$ and define the following event during the decoding of block $t$, 
\begin{IEEEeqnarray}{llll}
\nonumber
E_{(s, r, k)}^{(t)} = & \Big\lbrace \Big( & 
\bs{u}\Big( & (\hat{W}_{1,C}^{(T-t)},\Omega_{1}^{(T-t)}), (\hat{W}_{2,C}^{(T-t)},\hat{\Omega}_{2}^{(T-t)})\Big),\\
\nonumber 
& & \bs{u}_1\Big( & (\hat{W}_{1,C}^{(T-t)},\Omega_{1}^{(T-t)}), (\hat{W}_{2,C}^{(T-t)},\hat{\Omega}_{2}^{(T-t)}), \\
\nonumber
& & &(W_{1,C}^{(T-(t-1))},\Omega_{1}^{(T-(t-1))}) \Big),\\
\nonumber
& & \bs{x}_1\Big(& (\hat{W}_{1,C}^{(T-t)},\Omega_{1}^{(T-t)}), (\hat{W}_{2,C}^{(T-t)},\hat{\Omega}_{2}^{(T-t)}), \\
\nonumber
& & &(W_{1,C}^{(T-(t-1))},\Omega_{1}^{(T-(t-1))}), \hat{W}_{i,P}^{(T-t)}\Big), \\
\nonumber
& & \bs{u}_2\Big( & (\hat{W}_{1,C}^{(T-t)},\Omega_{1}^{(T-t)}), (\hat{W}_{2,C}^{(T-t)},\hat{\Omega}_{2}^{(T-t)}), \\
\nonumber
& & & (W_{2,C}^{(T-(t-1))},\Omega_{2}^{(T-(t-1))}) \Big), \bs{y}_1^{(T-t)} \Big) \\
& & \in & \mathcal{A}_{e}^{(n)}.
\Big\rbrace,
\end{IEEEeqnarray}
where $s = (\hat{W}_{1,C}^{(T-t)},\Omega_{1}^{(T-t)})$, $r = (\hat{W}_{2,C}^{(T-t)},\hat{\Omega}_{2}^{(T-t)})$ and $k = \hat{W}_{i,P}^{(T-t)}$, with $(s,r,k) \in \lbrace 1, \ldots, 2^{N(R_{1,C} + R_{1,R} )} \rbrace \times \lbrace 1, \ldots, 2^{N(R_{2,C} + R_{2,R} )} \rbrace \times \lbrace 1, \ldots, 2^{N R_{1,P}} \rbrace$.
Assume also that at step $t$ the indices $(s, r, k)$ are $(1,1,1)$ without loss of generality, due to the symmetry of the code.
Then, the probability of error due to $(ii)$ during step $t$, $p_e$, can be  bounded as follows:
\begin{IEEEeqnarray}{lcl}
\nonumber
p_e  & = &\pr{\ds\bigcup_{(s,r,k) \neq (1,1,1)} E_{(s, r, k)}^{(t)} }\\
\nonumber
& \leqslant & \ds\sum_{\scriptscriptstyle s\neq1, r \neq 1, k \neq 1} \pr{ E_{(s, r, k)}^{(t)} } + \ds\sum_{\scriptscriptstyle s = 1, r \neq 1, k \neq 1} \pr{ E_{(s, r, k)}^{(t)} } \\
\nonumber
&  & + \ds\sum_{\scriptscriptstyle s\neq1, r = 1, k \neq 1} \pr{ E_{(s, r, k)}^{(t)} } + \ds\sum_{\scriptscriptstyle s \neq 1, r \neq 1, k = 1} \pr{ E_{(s, r, k)}^{(t)} } \\
\nonumber
&  & + \ds\sum_{\scriptscriptstyle s = 1, r = 1, k \neq 1} \pr{ E_{(s, r, k)}^{(t)} } + \ds\sum_{\scriptscriptstyle s \neq 1, r = 1, k = 1} \pr{ E_{(s, r, k)}^{(t)} }\\
&  & + \ds\sum_{\scriptscriptstyle s =1, r \neq 1, k = 1} \pr{ E_{(s, r, k)}^{(t)} }
\end{IEEEeqnarray}
\begin{IEEEeqnarray}{lcl}
\nonumber
& \leqslant & 2^{N ( R_{1,C} + R_{1,R} + R_{2,C} + R_{2,R} + R_{1,P} - I(U, U_2, X_1; Y_1) + 4\epsilon) }\\
\nonumber
 & & +2^{N ( R_{2,C} + R_{2,R} + R_{1,P}- I(U, U_2, X_1; Y_1) + 4\epsilon) }\\
 \nonumber
 & & +2^{N ( R_{1,C} + R_{1,R} + R_{1,P} - I(U, U_2, X_1; Y_1) + 4\epsilon) }\\
 \nonumber
 & & + 2^{N ( R_{1,C} + R_{1,R} + R_{2,C} + R_{2,R}  - I(U, U_2, X_1; Y_1) + 4\epsilon) }\\
 \nonumber
 & & +2^{N (R_{1,P}  - I( X_1;Y_1 |U, U_1 , U_2) + 4\epsilon) }\\
 \nonumber
 & & +2^{N ( R_{1,C} + R_{1,R} - I(U, U_2, X_1; Y_1) + 4\epsilon) }\\
\label{EqConditionNoError2}
 & & +2^{N ( R_{2,C} + R_{2,R} - I(U, U_2, X_1; Y_1) + 4\epsilon) }.
\end{IEEEeqnarray}
Now, from \eqref{EqConditionNoError1} and \eqref{EqConditionNoError2}, given that $I\left( U_i ; Y_j | X_j, U \right) < I\left(U, U_i, X_j; Y_j\right)$, the probability of error due to $(ii)$ can be made arbitrarily small if the following conditions hold:
\begin{equation} \label{EqRateRegion0}
\left\lbrace  
\begin{array}{rl}
R_{2,C} + R_{2,R}  \leqslant & I \left( U_2 ; Y_1 | X_1, U \right)\\
R_{1,P}  \leqslant & I \left( X_1 ; Y_1 | U, U_1, U_2 \right)\\
R_{1,C} +R_{1,R} + R_{1,P} \\
+ R_{2,C} + R_{2,R} \leqslant & I \left( U ,U_2 , X_1 ; Y_1 \right).
\end{array}
\right.
\end{equation}
The common random message index $\Omega_{i}^{(t)}$ is assumed to be known at both transmitter $i$ and receiver $i$. Therefore, the set of inequalities in \eqref{EqRateRegion0} reduces to the following:
\begin{equation} \label{EqRateRegion1}
\left\lbrace  
\begin{array}{rl}
R_{2,C} + R_{2,R}  \leqslant & I \left( U_2 ; Y_1 | X_1, U \right)\\
R_{1,P}  \leqslant & I \left( X_1 ; Y_1 | U, U_1, U_2 \right)\\
R_{1,C} + R_{1,P} + R_{2,C} + R_{2,R} \leqslant & I \left( U ,U_2 , X_1 ; Y_1 \right).
\end{array}
\right.
\end{equation}
The same analysis is carried out for transmitter $2$ and thus,
\begin{equation}\label{EqRateRegion2}
\left\lbrace  
\begin{array}{rl}
R_{1,C} + R_{1,R}   \leqslant & I \left( U_1 ; Y_2 | X_2, U \right)\\
R_{2,P} \leqslant & I \left( X_2 ; Y_2 | U, U_1, U_2 \right)\\
R_{2,C} + R_{2,P} + R_{1,C} + R_{1,R} \leqslant & I\left( U ,U_1 , X_2 ; Y_2 \right).
\end{array}
\right.
\end{equation}
From the probability of error analysis, it follows that the rate-pairs achievable with the proposed randomized coding scheme with feedback are those simultaneously satisfying conditions \eqref{EqRateRegion1} and \eqref{EqRateRegion2}. Indeed, when $R_{1,R} = R_{2,R} = 0$, the coding scheme described above reduces to the coding scheme presented in \cite{Suh-TIT-2011} and the achievable region corresponds to the entire capacity region of $\Cldicfb$.
In terms of the linear deterministic model, it follows that $\forall i \in \lbrace 1, 2 \rbrace$,
\begin{equation} 
\left\lbrace 
\begin{array}{lcl} 
R_{i,C} + R_{i,R}   & \leqslant & n_{ji}\\
R_{i,P}  & \leqslant & \left(n_{ii} - n_{ji} \right)^+ \\
R_{i,C} + R_{i,P} + R_{j,C} + R_{j,R} & \leqslant & \max(n_{ii},n_{ij}),
\end{array}
\right. 
\end{equation}
which completes the proof of Lemma \ref{LemmaLDICFBb}.
\section{Proof of Lemma \ref{LemmaGICFBb}} \label{AppProofOfLemmaGICFBb}

The code generation presented in Appendix \ref{AppProofOfLemmaLDICFBb} is general and thus, it applies for both the LD-IC-FB model and the GIC-FB model. Therefore, from the analysis of the error probability, the rate tuples achievable in the Gaussian interference channel by the randomized Han-Kobayashi scheme with feedback satisfy  the inequalities in \eqref{EqRateRegion1} and \eqref{EqRateRegion2}. That is,
\begin{equation} 
\left\lbrace 
\begin{array}{lcl} 
R_{i,C} + R_{i,R}   & \leqslant & I\left( U_{i};Y_{j}| X_{j}, U\right)\\
R_{j,P}  & \leqslant & I\left(X_{i};Y_{i}| U, U_{1},U_{2} \right)\\
R_{i,C} + R_{i,P} + R_{j,C} + R_{j,R} & \leqslant & I\left( U,U_{j},X_{i};Y_{i}\right).
\end{array}
\right. 
\end{equation}
Suppose that transmitter $i$ uses the Gaussian input distribution 
\begin{IEEEeqnarray}{rcl}
\label{EqGaussianInputDistribution1}
U &\sim &\mathcal{N}(0,\rho),\\
\label{EqGaussianInputDistribution2}
U_i & \sim & \mathcal{N}(0,\lambda_{i,C})  \mbox{ and } \\
\label{EqGaussianInputDistribution3}
X_{i,P} &\sim & \mathcal{N}(0,\lambda_{i,P}),
\end{IEEEeqnarray}
with $\rho + \lambda_{i,C} + \lambda_{i,P} = 1$ and let $\lambda_{i,P}$ be determined by \eqref{EqLambdaip}.
Let also 
\begin{equation}
\label{EqConstructionX_i}
X_i = U + U_{i} + X_{i,P},
\end{equation}
 and assume that $U$, $U_{1}$, $U_{2}$, and $X_{1,P}$ are mutually independent. 
Then the following holds:
\begin{IEEEeqnarray}{lcl}
\label{EqHKgicfb4}
I(U,U_j,X_i;Y_i) & = & h(Y_i) - h(Y_i|U,U_j,X_i)\\
\nonumber
&= & \frac{1}{2}\log\Big( 1 + \SNR_{i} + \INR_{ij} \\
\nonumber
&   &+ 2 \rho \sqrt{\SNR_{i} \INR_{ij}} \Big)\\
\nonumber
&  &- \frac{1}{2} \log\Big( 1 + \lambda_{j,P}\INR_{ij} \Big).   
\end{IEEEeqnarray}

\noindent 
The term $I(U_i;Y_j | U,X_j)$ satisfies the condition
\begin{IEEEeqnarray}{lcl}
\label{EqHKgicfb5}
I(U_i;Y_j | U,X_j) & = & h(Y_j  | U,X_j ) - h(Y_j | U,U_i,X_j)\\
\nonumber
& = & \frac{1}{2} \log\Big( 1 + (1-\rho)\INR_{ji}  \Big) \\
\nonumber
&  &- \frac{1}{2} \log\Big( 1 + \lambda_{i,P}\INR_{ji}  \Big).
\end{IEEEeqnarray}
Similarly, the term $I(X_i;Y_i | U,U_1U_2)$ satisfies the condition
\begin{IEEEeqnarray}{lcl}
\nonumber
I(X_i;Y_i | U,U_1,U_2) & = & h(Y_i  | U,U_1,U_2 ) - h(Y_i | U,U_1,U_2,X_i)\\
\nonumber
& = & \frac{1}{2} \log\Big( \lambda_{i,P} \SNR_{i} + \lambda_{j,P} \INR_{ij} + 1 \Big)\\
\label{EqHKgicfb6}
&  & - \frac{1}{2} \log\Big( 1 + \lambda_{j,P} \INR_{ij} \Big).
 \end{IEEEeqnarray}
Finally it follows from \eqref{EqHKgicfb4}, \eqref{EqHKgicfb5} and \eqref{EqHKgicfb6} that $\forall i \in \lbrace 1, 2 \rbrace$,
\begin{IEEEeqnarray}{l} 
\nonumber
R_{i,C} + R_{i,R} \leqslant  \frac{1}{2} \log\big(1+ (1-\rho)\INR_{ji} \big) \\
\qquad - \frac{1}{2} \log\big(1+ \lambda_{i,P}\INR_{ji}, \big), \\
\nonumber
R_{i,P}   \leqslant \frac{1}{2} \log\Big( \lambda_{i,P} \SNR_{i} + \lambda_{j,P} \INR_{ij}   + 1 \Big)\\
\qquad  - \frac{1}{2} \log\Big( 1 + \lambda_{j,P} \INR_{ij} \Big), \mbox{ and }\\
\nonumber
R_{i,C} + R_{i,P} + R_{j,C} + R_{j,R}  \leqslant  \frac{1}{2} \log\Big( 1 + \SNR_{i} + \INR_{ij} \\
\qquad   + 2 \rho \sqrt{\SNR_{i} \INR_{ij}} \Big) - \frac{1}{2} \log\Big( 1 + \lambda_{j,P}\INR_{ij} \Big). 
\end{IEEEeqnarray}
In the specific scenario in which $\INR_{ji} \geqslant 1$, the inequalities above can be written as follows:
\begin{IEEEeqnarray}{cl} 
\nonumber
R_{i,C} + R_{i,R}& \leqslant \frac{1}{2} \log\big(1+ (1-\rho)\INR_{ji} \big)  - \frac{1}{2}\\
\nonumber
R_{i,P}  & \leqslant \frac{1}{2} \log\Big( 2 + \frac{\SNR_{i}}{\INR_{ji}} \Big) - \frac{1}{2}\\
\nonumber
R_{i,C} + R_{i,P} + R_{j,C} + R_{j,R} & \leqslant \frac{1}{2} \log\Big(1 + \SNR_{i} + \INR_{ij} \\
\nonumber
& + 2 \rho \sqrt{\SNR_{i} \INR_{ij}} \Big) -\frac{1}{2},
\end{IEEEeqnarray}
and this concludes the proof.
\balance
\section*{Acknowledgment}
The authors would like to thank the anonymous reviewers for some helpful suggestions that substantially improved this paper. 
\bibliographystyle{IEEEtran}
\bibliography{IT-GT}

\begin{IEEEbiographynophoto}{Samir M. Perlaza} (S'07-M'11-SM'15) is a research scientist with the Institut National de Recherche en Informatique et en Automatique (INRIA), France, and a visiting research collaborator at the School of Applied Science at Princeton University (NJ, USA). He received the M.Sc. and Ph.D. degrees from Ecole Nationale Sup\'{e}rieure des T\'{e}l\'{e}communications (Telecom ParisTech), Paris, France, in 2008 and 2011, respectively. 
Previously, from 2008 to 2011, he was a Research Engineer at France T\'{e}l\'{e}com - Orange Labs (Paris, France). He has held long-term academic appointments at the Alcatel-Lucent Chair in Flexible Radio at Sup\'{e}lec (Gif-sur-Yvette, France); at Princeton University (Princeton, NJ) and at the University of Houston (Houston, TX). His research interests lie in the overlap of signal processing, information theory, game theory and wireless communications. 

Dr. Perlaza has been distinguished by the European Commission with an Al$\beta$an Fellowship in 2006 and a Marie Sk\l{}odowska-Curie Fellowship in 2015. He was also one of the recipients of the the Best Student Paper Award at International Conference on Cognitive Radio Oriented Wireless Networks CROWNCOM in 2009.
\end{IEEEbiographynophoto}

\begin{IEEEbiographynophoto}{Ravi Tandon}(S'03-M'09) received the B.Tech. degree in electrical engineering from IIT Kanpur, Kanpur, India, in 2004, and the Ph.D. degree in electrical and computer engineering from the University of Maryland, College Park, MD, USA, in 2010. From 2010 to 2012, Dr. Tandon was a Post-Doctoral Research Associate with Princeton University, Princeton, NJ, USA. Since 2012, he has been with the Virginia Polytechnic Institute and State University, Blacksburg, VA, USA, where he is currently a Research
Assistant Professor with the Discovery Analytics Center and the Department of Computer Science.  From Fall 2015, he will be joining the Electrical and Computer Engineering Department at the University of Arizona as an Assistant Professor.

His research interests are in the areas of network information theory for wireless networks, information theoretic security, machine learning, and cloud storage systems. He is a co-recipient of the Best Paper Award at the Communication Theory Symposium at the 2011 IEEE Global Communications Conference.
\end{IEEEbiographynophoto}

\begin{IEEEbiographynophoto}{H. Vincent Poor} (S'72-M'77-SM'82-F'87) received the Ph.D. degree in electrical engineering and computer science from Princeton University in 1977.  From 1977 until 1990, he was on the faculty of the University of Illinois at Urbana-Champaign. Since 1990 he has been on the faculty at Princeton, where he is the Dean of Engineering and Applied Science, and the Michael Henry Strater University Professor of Electrical Engineering. He has also held visiting appointments at several other institutions, most recently at Imperial College and Stanford. Dr. Poor's research interests are in the areas of information theory, stochastic analysis and statistical signal processing, and their applications in wireless networks and related fields. Among his publications in these areas is the recent book \emph{Mechanisms and Games for Dynamic Spectrum Allocation} (Cambridge University Press, 2014).

Dr. Poor is a member of the National Academy of Engineering and the National Academy of Sciences, and is a foreign member of Academia Europaea and the Royal Society. He is also a fellow of the American Academy of Arts and Sciences, the Royal Academy of Engineering (U. K.), and the Royal Society of Edinburgh.  In 1990, he served as President of the IEEE Information Theory Society, in 2004-07 as the Editor-in-Chief of these TRANSACTIONS, and in 2009 as General Co-chair of the IEEE International Symposium on Information Theory, held in Seoul, South Korea. He received a Guggenheim Fellowship in 2002 and the IEEE Education Medal in 2005. Recent recognition of his work includes the 2014 URSI Booker Gold Medal, and honorary doctorates from several universities in Asia and Europe.

\end{IEEEbiographynophoto}

\begin{IEEEbiographynophoto}{Zhu Han} (S'01ÐM'04-SM'09-F'14) received the B.S. degree in electronic engineering from Tsinghua University, in 1997, and the M.S. and Ph.D. degrees in electrical engineering from the University of Maryland, College Park, in 1999 and 2003, respectively. 

From 2000 to 2002, he was an R\&D Engineer of JDSU, Germantown, Maryland. From 2003 to 2006, he was a Research Associate at the University of Maryland. From 2006 to 2008, he was an assistant professor in Boise State University, Idaho. Currently, he is a Professor in Electrical and Computer Engineering Department as well as Computer Science Department at the University of Houston, Texas. His research interests include wireless resource allocation and management, wireless communications and networking, game theory, wireless multimedia, security, and smart grid communication. Dr. Han is coauthor of 2015 EURASIP Best Paper Award for the Journal on Advances in Signal Processing.  Dr. Han is the winner of IEEE Fred W. Ellersick Prize 2011. Dr. Han is an NSF CAREER award recipient 2010. Dr. Han is IEEE Distinguished lecturer since 2015. 
\end{IEEEbiographynophoto}
\balance
\end{document}